\documentclass{article}

\usepackage{graphicx} 
\usepackage{amsmath}
\usepackage{amsfonts}
\usepackage{amssymb}
\usepackage{amsthm}
\usepackage{xcolor}
\usepackage{thmtools} 
\usepackage{thm-restate}
\usepackage{algorithm}
\usepackage[noend]{algpseudocode}
\usepackage{algorithmicx}
\usepackage[scale=0.75,top=3cm]{geometry} 
\usepackage[hidelinks]{hyperref}
\usepackage{cleveref}
\usepackage{bbm}
\usepackage{tikz}
\usetikzlibrary{calc}


\newcommand{\bigO}{\mathcal{O}} 
\newcommand{\calG}{\mathcal{G}} 

\newcommand{\rt}{\mathrm{root}}

\DeclareMathOperator{\poly}{poly} 
\DeclareMathOperator{\greedy}{greedy} 
\DeclareMathOperator{\reverse}{reverse} 
\DeclareMathOperator{\supp}{supp} 
\DeclareMathOperator{\lcp}{lcp} 

\DeclareMathOperator{\iso}{iso} 
\DeclareMathOperator{\home}{home}
\DeclareMathOperator{\cou}{cou}

\DeclareMathOperator{\ancs}{ancs} 
\DeclareMathOperator{\id}{id} 

\DeclareMathOperator{\rep}{rep} 

\DeclareMathOperator{\conn}{conn} 

\let\cong\undefined
\DeclareMathOperator{\dil}{dil}
\DeclareMathOperator{\cong}{cong}

\newcommand{\calD}{\mathcal{D}}

\algdef{SE}
[CLASS]
{Class}
{EndClass}
[1]
{\textbf{class} \textsc{#1}}
{\textbf{end class}}

\algdef{SE}
[DATA]
{Data}
{EndData}
[0]
{\textbf{data members}}

\algnewcommand{\LeftComment}[1]{\State {\color{blue}\(\triangleright\) #1}}

\renewcommand*\Call[2]{\textproc{#1}(#2)}

\newtheorem{theorem}{Theorem}[section]
\newtheorem{lemma}[theorem]{Lemma}

\theoremstyle{definition}
\newtheorem{definition}[theorem]{Definition}

\title{Maintaining Routing Structures under Deletions via Self-Pruning}

\begin{document}

\author{
    Bernhard Haeupler\thanks{Partially funded from the Ministry of Education and Science of Bulgaria (support for INSAIT, part of the Bulgarian National Roadmap for Research Infrastructure).  Also partially funded through the European Research Council (ERC) under the European Union's Horizon 2020 research and innovation program (ERC grant agreement 949272). }\\\texttt{bernhard.haeupler@inf.ethz.ch}\\INSAIT, Sofia University\\``St. Kliment Ohridski''\\\& ETH Z\"urich \and
    Antti Roeyskoe\thanks{Partially funded through the European Research Council (ERC) under the European Union's Horizon 2020 research and innovation program (ERC grant agreement 949272). Parts of this work were done while the author was visiting INSAIT. Partially funded from the Ministry of Education and Science of Bulgaria (support for INSAIT, part of the Bulgarian National Roadmap for Research Infrastructure).}\\\texttt{antti.roeyskoe@inf.ethz.ch}\\ETH Z\"urich}
\date{}

\maketitle

\begin{abstract}
\noindent Expanders are powerful algorithmic structures with two key properties: they are
\begin{enumerate}
    \item[a)] routable: for any multi-commodity flow unit demand, there \emph{exists} a routing with low congestion over short paths, where a demand is unit if the amount of demand sent / received by any vertex is at most the number of edges adjacent to it.
    \item[b)] stable / prunable: for any (sequence of) edge 
    failures, there \emph{exists} a proportionally small subset of vertices that can be disabled, such that the graph induced on the remaining vertices is an expander. 
\end{enumerate}
Two natural algorithmic problems correspond to these two existential guarantees: \textit{expander routing}, i.e. computing a low-congestion routing for a unit multi-commodity demand on an expander, and \textit{expander pruning}, i.e., maintaining the subset of disabled vertices under a sequence of edge failures.

This paper considers the combination of the two problems: maintaining a routing for a unit multi-commodity demand under pruning steps. 
This is done through the introduction of a family of expander graphs that, like hypercubes, are easy to route in, and are \textit{self-pruning}: for an online sequence of edge deletions, a simple self-contained algorithm can find a few vertices to prune with each edge deletion, such that the remaining graph always remains an easy-to-route-in expander in the family. 

Notably, and with considerable technical work, this self-pruning can be made \emph{worst-case}, i.e., such that every single adversarial deletion only causes a small number of additional deletions. Our results also allow tight constant-factor control over the length of routing paths (with the usual trade-offs in congestion and pruning ratio) and therefore extend to constant-hop and length-constrained expanders in which routing over constant length paths is crucial.
\end{abstract}

\thispagestyle{empty} 
\newpage
\setcounter{page}{1}


\section{Introduction}

We give a family of (length-constrained expander) graphs that like hypercubes are easy to route in and are \textit{self-pruning}: for an online sequence of edge deletions, a simple self-contained algorithm can find a few vertices to prune with each edge deletion, such that the remaining graph is always an easy-to-route-in expander in the family.

\textit{Expander graphs} are graphs where any unit (i.e. degree-respecting) demand can be routed with low congestion over short paths. Specifically, a $\phi$-expander is a graph on which every unit demand can be routed with congestion $\log(n) / \phi$ over paths of length at most $\log(n) / \phi$. Even on graphs that are not expanders, the power of expanders can be accessed by computing an \textit{expander decomposition}, which deletes a small number of edges (relative to $1 / \phi$), guaranteeing that afterward, every connected component is an $\phi$-expander.

The algorithmic problem of exploiting this existential routing property of expanders is called \textit{expander routing} and has been intensely studied~\cite{RandomWalkExpanderRouting17, FasterRandomExpanderRouting18, DetExpanderRouting20, FasterDetExpanderRouting24, MCFlowEmulator24}. There are particular easy-to-route-in expander topologies, such as the \textit{hypercube}, where the $2^d$ vertices $V = \{0, 1\}^d$ correspond to binary strings of length $d$, and there is an edge between two vertices if and only if their corresponding strings differ in exactly one bit. On a hypercube, \textit{Valiant routing} \cite{ValiantRouting} can route any unit demand with congestion $\bigO(d)$ over paths of length $\bigO(d)$ by selecting for every demand pair a uniformly random midpoint vertex, and constructing the routing path by walking from the midpoint to both endpoints by greedily flipping differing bits. This routing strategy has the additional advantage of being \textit{oblivious}: the routing path for a particular demand pair does not depend at all on the rest of the demand. While Valiant routing is randomized, deterministic routing algorithms on a hypercube exist as well (e.g., based on sorting networks). On an unstructured $\phi$-expander, expander routing is equivalent up to a logarithmic loss in length and congestion to embedding a hypercube on the graph. Once a hypercube has been embedded, routing on the graph can be done by first routing on the hypercube, then projecting the routing onto the original graph. This approach only results in overhead in congestion and path length equal to the congestion and path length of the embedding, and has the additional advantages of being oblivious, and being able to efficiently answer future expander routing instances, as the hypercube does not need to be re-embedded.

In addition to being excellent routers, another property of expanders crucial to many applications is their \textit{robustness to edge failures}. Specifically, for any sequence $e_1, e_2, \dots$ of online edge deletions from a $\phi$-expanding graph, it is possible to expand an initially empty set of vertices $P$ such that the remaining graph induced on $V \setminus P$ always remains an expander of the same asymptotic quality $\bigO(\phi)$, and the volume of $P$ is always at most $\bigO(i / \phi)$, where $i$ is the number of edge deletions processed so far \cite{EvenBetterExpanderPruning19}. The algorithmic problem of selecting the vertices to add to $P$ corresponding to this existential property is called \textit{expander pruning}. Like expander routing it is one of the fundamental algorithmic tools when working with expanders and has been studied in depth~\cite{ExpanderPrunability04, DynamicPruning17, OneShotPruning17, BetterExpanderPruning17, EvenBetterExpanderPruning19, WorstCasePruning25}.

While algorithms exist for both problems of expander routing and expander pruning, issues arise when attempting to algorithmically utilize the routing properties and prunability of expanders in tandem: even if the graph is initially an easy-to-route-in expander such as a hypercube, after the deletion of some edges and an application of pruning, even though the graph is still guaranteed to be an expander, the easy approaches for routing on it are no longer guaranteed to work: for example, with Valiant routing, a greedy path might no longer exist.

This motivates the following natural question:

\begin{quote}
\textit{Is it possible to perform pruning on an easy-to-route-in expander, such that the remaining graph not only always remains expanding, but in fact always remains an easy-to-route-in expander?}
\end{quote}

Of course, one would also want to \emph{efficiently} perform the pruning, though even an inefficient approach would be interesting. There is a similar natural question when the initial graph is unstructured:

\begin{quote}
\textit{Is it possible to efficiently perform pruning on an expander while maintaining a short low-congestion embedding of an easy-to-route-in expander on the remaining graph?}
\end{quote}

The expander pruning property as stated above and all pruning algorithms until very recently only bounded the \textit{asymptotic} volume of the pruned vertex set $P$. A stronger, \emph{worst-case} guarantee on the number of vertices entering $P$ during any the deletion of an edge is very desirable but turns out to be a very technically challenging problem (see \cite{WorstCasePruning25} for an in-depth discussion on worst-case pruning) with \cite{WorstCasePruning25} being the first expander pruning algorithm with such a guarantee. Can the above goals be achieved even with \emph{worst-case} pruning?

In many settings and applications, there is more importance placed on the length of the routing rather than its congestion (see e.g. discussions in \cite{LCCutmatchGame25}). The object corresponding to an expander where the two parameters of routing length and routing congestion are "unbalanced" is called a \textit{router}: specifically, a $h$-length $\kappa$-congestion \textit{router} is a graph on which any unit demand can be routed with congestion $\kappa$ over paths of length $h$. For an example of a router of length $o(\log n)$, consider the "larger-alphabet" generalization of a hypercube with vertex set $[k]^d$ and an edge between two vertices if their strings differ in exactly one character. This graph is a $\bigO(d)$-length $\bigO(d)$-congestion router where each vertex has degree $(k - 1) d$, on which Valiant routing works exactly as on a hypercube. There are corresponding existential guarantees and algorithms in this setting to expander routing \cite{MCFlowEmulator24}, expander pruning \cite{NewLCED24, ConstApproxDynamicDistOracle24} and expander decomposition \cite{OrigLCED22, NewLCED24}. Can the above goals be achieved for instead pruning a \emph{router}?
\section{Our Results}

We introduce a family of routers called $(k, d, \tau)$-semi-hypercubes, which are graphs with up to $n = k^d$ vertices of maximum degree $(k - 1) d = n^{\bigO(1 / d)}$. For small enough $d = \bigO(\sqrt{\log(n) / \log \log(n)})$, these graphs have the following properties:

\begin{itemize}
    \item \textbf{Easy-to-route-in:} on any $(k, d, \tau)$-semi-hypercube, there is an oblivious routing of length $\bigO(d^3)$ that can route any unit demand with congestion $n^{\bigO(1 / d)}$. A path from this oblivious routing can be sampled in expected time $\bigO(d^3)$. We show this in \Cref{thm:obliv-routing-general} in \Cref{sec:rand-routing}.

    Additionally, an explicit routing of length $2^{\bigO(d)}$ and congestion $n^{\bigO(1 / d)}$ for a dynamic demand can be maintained deterministically with low recourse on a fully dynamic graph guaranteed to initially be and always remain a $(k, d, \tau)$-semi-hypercube. We show this in \Cref{thm:det-routing} in \Cref{sec:det-routing}. 
    \item \textbf{Self-Pruning:} there is a deterministic, simple and self-contained vertex pruning algorithm with a worst-case pruning ratio and work per update of $n^{\bigO(1 / d)}$ for maintaining that the graph remains a $(k, d, \tau)$-semi-hypercube. We show this in \Cref{thm:shc-strong-pruning} in \Cref{sec:pruning-formal-proof}. An overview of the pruning algorithm is given in \Cref{sec:selfpruning-overview}.
\end{itemize}

Note that in any graph with maximum degree $n^{\bigO(1 / d)}$, the distance between almost all vertex pairs must be $\Omega(d)$, thus $(k, d, \tau)$-semi-hypercubes are routers with only a polynomial loss in length to the optimum length of routers with the same maximum degree.

Given a router with an embedded $(k, d, \tau)$-semi-hypercube, we can maintain the router and the embedded semi-hypercube on it under worst-case pruning by relying on the self-pruning property of the semi-hypercube, as we show in \Cref{sec:pruning-through-embedding}. Combining this with an algorithm for routing on a router for the initial embedding, we obtain the following result. 

\begin{restatable*}{corollary}{pruningobliviousrouting}\label{cor:pruning-oblivious-routing}
Let $G = (V, E)$ be a $h$-length $\kappa$-congestion router for the all-one node weighting $\mathbbm{1}_V$. Then, there is a deterministic algorithm that for a given positive integer $d = \bigO\left(\sqrt{\frac{\log n}{\log \log n}}\right)$ maintains a remaining vertex set $V' \subseteq V$ (initially $V' = V$) and an oblivious routing $R$ on $G[V']$ with
\begin{itemize}
    \item congestion $\kappa \cdot (h \cdot n^{\bigO(1 / d)})$ on $\mathbbm{1}_{V'}$-respecting demands,
    \item path length $h \cdot \bigO(d^3)$, and
    \item efficiently sampleable paths: a path can be sampled from $R(s, t)$ with expected work $\bigO(h \cdot d^3)$
\end{itemize}
under pruning updates with a worst-case pruning ratio of $\kappa h \cdot n^{\bigO(1 / d)}$: the query to the data structure gives an edge $e^{\mathrm{del}} \in E[V']$ to be deleted. The algorithm then selects a set of vertices $V^{\mathrm{prune}} \subseteq V'$ of size $|V^{\mathrm{prune}}| \leq \kappa h \cdot n^{\bigO(1 / d)}$. The remaining vertex set is updated to $V' \gets V' \setminus V^{\mathrm{prune}}$ and the edge set to $E \gets E - e^{\mathrm{del}}$.

The data structure has work per update of $\kappa h \cdot n^{\bigO(1 / d)}$. The initialization of the data structure takes work $\kappa h \cdot m^{1 + \bigO(1 / d)}$, plus the work required to route a single integral $n^{1 + \bigO(1 / d)}$-commodity $n^{\bigO(1 / d)}$-load demand on $G$ over paths of length $\bigO(h)$ with congestion $\kappa \cdot n^{\bigO(1 / d)}$.
\end{restatable*}

Note that the losses of $\poly(d)$ in path length and $n^{\bigO(1 / d)}$ in congestion are the desired type of tradeoff when working with length-constrained objects such as routers, similar to $\poly \log(n)$-losses in expansion in the balanced setting. The initialization of the data structure is a simple offline expander routing problem, for which one can use any of the current state-of-the-art algorithms.

While the statement of \Cref{cor:pruning-oblivious-routing} is at its strongest when applied to pruning routers of length $h$ up to $\bigO\left(\sqrt{\frac{\log n}{\log \log n}}\right)$, the result for normal expanders is still interesting: assuming the initial graph is a constant-degree $\phi$-expander, applying \Cref{cor:pruning-oblivious-routing} immediately obtains an expander pruning algorithm with a \emph{worst-case} pruning ratio and worst-case work per update of $(n^{o(1)} / \phi^2)$ for maintaining that the graph remains a $(\phi^2 / n^{o(1)})$-expander, \emph{and an efficiently sampleable oblivious routing on the graph}. While the $n^{o(1)}$-loss in expansion is not great, the remaining graph is in fact maintained to be a \emph{$\bigO(\log^{1 + o(1)}(n) / \phi)$-length router} of congestion $(\phi^2 / n^{o(1)})$, thus there is barely any loss in the length of the oblivious routing.

We also obtain in \Cref{cor:pruning-deterministic-routing} a similar result as \Cref{cor:pruning-oblivious-routing}, that instead of maintaining query access to an oblivious routing, explicitly maintains with low recourse a routing for a dynamic demand of bounded load on the remaining graph.

\begin{restatable*}{corollary}{pruningdeterministicrouting}\label{cor:pruning-deterministic-routing}
Let $G = (V, E)$ be a $h$-length $\kappa$-congestion router for the all-one node weighting $\mathbbm{1}_V$, $L \in \mathbb{N}$ a demand load bound, and $D$ an integral load-$L$ demand on $G$. Then, there is a deterministic algorithm that for a given positive integer $d = \bigO\left(\sqrt{\frac{\log n}{\log \log n}}\right)$ maintains
\begin{itemize}
    \item a remaining vertex set $V' \subseteq V$ (initially $V' = V$), such that $G[V']$ is a $\kappa \cdot (h \cdot n^{\bigO(1 / d)})$-congestion $h \cdot \bigO(d^3)$-length router for the all-one node weighting $\mathbbm{1}_{V'}$ on the remaining vertex set, and
    \item an explicit routing of $D$ over paths in $G[V']$ of congestion $\kappa L \cdot (h \cdot n^{\bigO(1 / d)})$ and length $h \cdot 2^{\bigO(d)}$, supporting $|p| \poly \log(n)$-work queries for the path $p$ assigned to a given demand pair,
\end{itemize}
under updates each consisting of two steps:
\begin{itemize}
    \item Pruning step: here, the query to the data structure gives an edge $e^{\mathrm{del}} \in E[V']$ to be deleted. The algorithm responds by selecting vertices $V^{\mathrm{prune}} \subseteq V'$ to \textit{prune}. The remaining vertex set is updated to $V' \gets V' \setminus V^{\mathrm{prune}}$ and the edge set to $E \gets E - e^{\mathrm{del}}$.
    \item Rerouting step: here, the query to the data structure gives sets of demand pairs $D^+$ and $D^-$ to respectively add and remove from the demand $D$, such that $D \gets D \cup D^+ \setminus D^-$ is supported with load $L$ on the post-pruning-step remaining vertex set $V'$. The algorithm responds by updating the routing, and returning the set $I^{\Delta}$ of identifiers of demand pairs for which the assigned path changed.
\end{itemize}
The algorithm has the following worst-case guarantees:
\begin{itemize}
    \item Worst case pruning ratio: $|V^{\mathrm{prune}}| \leq \kappa h \cdot n^{\bigO(1 / d)}$.
    \item Rerouting step recourse: $|I^{\Delta}| \leq L \cdot \kappa h \cdot n^{\bigO(1 / d)} + |D^-| \cdot 2^{\bigO(d)} + |D^+|$.
    \item Work per update: $(L \cdot \poly(\kappa h) \cdot n^{\bigO(1 / d)}) \cdot (L + |D^+| + |D^-|)$.
\end{itemize}
The initialization of the data structure takes work $L \cdot \kappa h \cdot m^{1 + \bigO(1 / d)}$, plus the work required to route a single integral $n^{1 + \bigO(1 / d)}$-commodity $n^{\bigO(1 / d)}$-load demand on $G$ over paths of length $\bigO(h)$ with congestion $\kappa \cdot n^{\bigO(1 / d)}$.
\end{restatable*}
\section{Related Work}

\textbf{Expander Pruning.} The first efficient algorithms for expander pruning were provided concurrently by \cite{OneShotPruning17} for the one-shot expander pruning problem (where instead of an online sequence of edge deletions, the whole set of edge deletions is given at once) and by \cite{DynamicPruning17} for the dynamic expander pruning problem. Subsequently, \cite{BetterExpanderPruning17} improved on both of these results, specifically improving the work per update of dynamic expander pruning from the $\bigO(n^{1/2 - \epsilon})$ for some small constant $\epsilon$ of \cite{DynamicPruning17} to $n^{o(1)}$. Asymptotically near-optimal dynamic expander pruning was achieved by \cite{EvenBetterExpanderPruning19}, with amortized work of $\bigO(\log(m) / \phi^2)$ per update and an amortized pruning ratio of $\bigO(1 / \phi)$, maintaining that the pruned graph $G[V \setminus P]$ always remains a $\phi / 6$-expander as long as it was initially $\phi$-expanding, and that the volume of the pruned vertex set $P$ after the first $i$ edge deletions is always at most $8i / \phi$.

All of the dynamic expander pruning results in this line of work~\cite{DynamicPruning17,BetterExpanderPruning17,bernstein2022fully,hua2023maintaining} bound only the \emph{amortized} pruning ratio, with no nontrivial guarantees on the \emph{worst-case} pruning ratio, i.e. the worst-case growth of the monotone set of pruned vertices $P$ during the deletion of any single edge (see also \cite{WorstCasePruning25} for more details). While worst-case pruning guarantees are very desirable for a multitude of reasons~\cite{WorstCasePruning25}, such results are much more technical and harder to obtain with \cite{WorstCasePruning25} being the first expander pruning algorithm giving a worst-case pruning ratio guarantee. Specifically, they show the following: given an $m$-edge $\phi$-expander graph and a sequence of up to $\tilde{\Omega}(\phi \cdot m)$ edge deletions to $G$, there is a deterministic algorithm that processes each edge deletion in time $\tilde{\bigO}(1 / \phi^2)$ and adds at most $\tilde{\bigO}(1 / \phi^2)$ vertices to the pruned vertex set $P$, such that $G[V \setminus P]$ always remains an $\tilde{\Omega}(\phi)$-expander.

We note that while the worst-case expander pruning result of \cite{WorstCasePruning25} was made public a earlier, our results were obtained independently and through completely different techniques. The focus and advantages of the results are different as well: while \cite{WorstCasePruning25} obtains stronger guarantees for regular expanders, the focus of this paper is on pruning while maintaining easy routability in the remaining graph, and smaller, even just constant, losses in length of routes on it, rather than super-constant polylogarithmic factors in the expansion of the graph. This constant length-loss is absolutely crucial in the context of algorithms for length-constrained expanders~\cite{LCCutmatchGame25,ConstApproxDynamicDistOracle24}.

\textbf{Expander Routing.} Expander routing was pioneered by \cite{RandomWalkExpanderRouting17}, who gave a randomized CONGEST algorithm that can handle any packet routing problem where each vertex $v$ is the source and destination of at most $\deg(v) \cdot 2^{\bigO(\sqrt{\log(n) \log \log(n)})}$ messages in $\tau_{\mathrm{mix}}(G) \cdot 2^{\bigO(\sqrt{\log(n) \log \log(n)})}$ rounds, where $\tau_{\mathrm{mix}}(G)$ is the mixing time of a random walk on $G$. For a $\phi$-expander, $\Omega(\frac{1}{\phi}) \leq \tau_{\mathrm{mix}}(G) \leq \bigO(\frac{\log n}{\phi})$ \cite{MixingTime89}, thus as a particular consequence, they gave a randomized algorithm for routing a degree-respecting demand on a $\phi$-expander with congestion, path length, and depth $\frac{1}{\phi} \cdot 2^{\bigO(\sqrt{\log(n) \log \log(n)})}$ and work $m$ times this. The overhead was subsequently lowered to $2^{\bigO(\sqrt{\log n})}$ by \cite{FasterRandomExpanderRouting18}. These algorithms are inherently randomized, as they rely on having each packet take a random walk to perform the routing.

The first \emph{deterministic} expander routing of \cite{DetExpanderRouting20} instead built on the deterministic cut-matching game and recursive approach of \cite{FastDetCutmatchGame20}. Their approach was to first partition the vertex set into some number $k$ of disjoint parts $V_1, V_2, \dots, V_k$, each with either $\lfloor n / k \rfloor$ or $\lceil n / k \rceil$ vertices. They then embed an expander into each vertex set $V_i$, such that the total congestion of these embeddings is low. To do this, they run simultaneous cut-matching games on each part $V_i$, selecting the matchings to add to the embeddings through a $k$-commodity flow instance, where the $i$th commodity is a demand between the two halves the cut-matching game algorithm split the vertices $V_i$ into. They then compute a $k$-commodity flow to route each packet into the part containing its destination, with a commodity for each destination part, and recursively route the packets in each part to their destinations on the embedded expanders. With this, they obtained a $\poly(1 / \phi) \cdot 2^{\log^{2 / 3 + o(1)}(n)}$-round deterministic packet routing algorithm on a $\phi$-expander. \cite{FasterDetExpanderRouting24} later refined this overhead to $\poly(1 / \phi) \cdot 2^{\bigO(\sqrt{\log(n) \log \log(n)})}$.

Traditionally, an advantage of routing on an expander has been that solving a multi-commodity flow problem was possible with work additively, not multiplicatively dependent on the number of commodities. Recently, \cite{MCFlowEmulator24} gave an algorithm for constant-approximate min-cost multicommodity flow with this property on general graphs. Applying this algorithm on a $\phi$-expander, one can obtain a $\bigO(\frac{\log n}{\phi})$-congestion routing over paths of length $\bigO(\frac{\log n}{\phi})$ for a unit demand, matching the asymptotic guarantees of the expander, with work $(m + k)^{1 + \epsilon}$ and depth $(m + k)^{\epsilon}$ for arbitrarily small constant $\epsilon > 0$. 

\textbf{Routing in a Dynamic Graph.} We now consider work related to maintaining routing-related data structures in a decremental or fully dynamic graph. Perhaps surprisingly, while even on a fully dynamic graph there are results for approximating the \emph{value} of a feasible (multicommodity) flow, to the best of our knowledge, no known approach with update time $o(n)$ provides access to low-congestion and low-length \emph{routing paths}.

An \textit{expander hierarchy} \cite{DynamicExpanderHierarchy21} is a hierarchical graph clustering formed by repeatedly computing a \textit{boundary-linked expander decomposition} of the graph and contracting each of its connected components into a single vertex, keeping parallel edges and removing self-loops, until the remaining graph contains no edges. \cite{DynamicExpanderHierarchy21} show that for appropriate parameters for the boundary-linked expander decomposition, the natural decomposition tree $T$ of the hierarchical clustering is a \textit{tree flow sparsifier} of quality $n^{o(1)}$, meaning that for any (multi-commodity) demand $D$ on $G$, (i) if the demand can be routed on $G$ with congestion $1$, it can be routed on $T$ with congestion $1$, and (ii) if the demand can be routed on $T$ with congestion $1$, it can be routed on $G$ with congestion $n^{o(1)}$. They also show how to maintain an expander hierarchy (and the associated tree) of a fully dynamic graph with $n^{o(1)}$ time and recourse per update. As solving multi-commodity flow problems on a tree is very easy, \cite{DynamicExpanderHierarchy21} show that this allows one to, among other things, query for the $n^{o(1)}$-approximate congestion of optimally routing a multi-commodity demand $D$ on $G$ with work $\bigO(|\supp(D)| \log^{1 / 6} n)$. All prior approaches for this problem, even in a decremental setting, took either $\Omega(n)$ update or query time.

While the dynamic expander hierarchy structure is very powerful, it has two major drawbacks: firstly, it cannot approximate distances \cite{OrigLCED22}, and secondly, it cannot be used to obtain routing paths. The problem with the latter is the expander case: if the graph is already a good-enough expander, the hierarchy is trivial, and the associated tree is simply a star. While this guarantees any unit demand can be routed on the graph with low congestion over short paths, it provides no help in actually performing the routing.

This former drawback was a major motivating factor that lead to the development of \textit{length-constrained expanders} \cite{OrigLCED22}. By dynamically maintaining a \textit{length-constrained expander hierarchy}, \cite{ConstApproxDynamicDistOracle24} obtained a constant-approximate distance oracle in a fully-dynamic graph with worst-case work of $n^{\epsilon}$ for arbitrarily small constant $\epsilon$ per graph update. While this result supports path queries, it does not give any congestion guarantees for these routes.


The recent almost-linear time decremental min-cost flow result of \cite{FastDecrementalFlow24} implicitly maintains a routing of $(1 + \epsilon)$-approximate congestion and average length for a single-commodity demand in a decremental directed graph. While they do not provide a way to access this flow, it is not clear that the algorithm without major changes could not be adapted to allow for e.g. efficiently sampling a routing path from the path decomposition of the flow. However, even if this was possible, this algorithm could not be used to maintain a routing of a multi-commodity demand, and updating the structure when an edge is deleted inherently takes amortized, not worst-case work. 

\textbf{Semi-Hypercubes.} the "larger-alphabet" generalization of a hypercube with vertex set $[k]^d$ and an edge between two vertices if their strings differ in exactly one character, is a simple example of an excellent router. This graph can also be seen as a recursive construction, where the vertex set is partitioned into $k$ equal-size parts, a perfect matching is added between each part, and a structure of dimension $d - 1$ is then built on each part. Even if these perfect matchings are not orthogonal (i.e. in that they match copies of the same vertex in different parts to each other, resulting in the neat "hamming-distance-1" description of the edge set), the graph remains a router of the same excellent quality, on which Valiant routing still works without changes. This router occurs naturally in every recursive approach to cut-matching games \cite{FastDetCutmatchGame20, LCCutmatchGame25}, expander routing \cite{RandomWalkExpanderRouting17, FasterRandomExpanderRouting18, DetExpanderRouting20, FasterDetExpanderRouting24, MCFlowEmulator24}, and also in \cite{JuliaDynamicHierarchy23, JuliaNiceRouterPruning25, NewLCED24}. We are not aware of an existing name for this router, and refer to it as a $(k, d)$-semi-hypercube, "semi" as the matchings are not orthogonal. $(k, d)$-semi-hypercubes are not new to this paper; our key contribution is observing a set of conditions on an induced subgraph of a $(k, d)$-semi-hypercube that are sufficiently strong for the subgraph to still be a good easy-to-route in router, while being soft enough to allow for worst-case pruning. The family of graphs satisfying these conditions is the family of \textit{$(k, d, \tau)$-semi-hypercubes} we introduce in this paper.

\section{Preliminaries}


\textbf{Graphs.} In this paper, we consider undirected, unit-edge-length, uncapacitated graphs $G = (V, E)$ with vertex sets $V$ and edge sets $E$.

\textbf{Node Weightings.} A \textit{node weighting} $A : V \rightarrow \mathbb{R}_{\geq 0}$ is simply a function that assigns to each vertex $v$ a weight $A(v)$. For two node weightings $A, A'$ on the same vertex set $V$, we write $A \preceq A'$ if $A(v) \leq A'(v)$ for all $v \in V$. For a graph $G$, the \textit{degree node weighting} $\deg_G$ is the node weighting assigning to each vertex $v$ weight $\deg_G(v)$ equal to its degree in $G$. The all-one node weighting $\mathbbm{1}_V$ is the node weighting assigning to each vertex the weight $\mathbbm{1}_V(v) = 1$.

\textbf{Demands.} A demand $D : V \times V$ is a function assigning a demand value $D(u, v)$ to each vertex pair, representing how much flow value should be sent from $u$ to $v$. The \textit{load} of a demand $D$ on $V$ is the node weighting on $V$ defined as $\mathrm{load}_D(v) := \sum_{v \in V} D(u, v) + D(v, u)$. We say the demand has load $L$ if $\mathrm{load}_D(v) \leq L$ for all $v \in V$. For a node weighting $A$, the demand $D$ is $A$-respecting if $\mathrm{load}_D \preceq A$, i.e. the total demand value involving any vertex $u$ is at most $A(u)$. A demand on a graph $G$ is called \textit{unit} or equivalently \textit{degree-respecting} if it is $\deg_G$-respecting. 

An \textit{integral demand} assigns an integral value $D(u, v) \in \mathbb{Z}_{\geq 0}$ to each vertex pair. For algorithms, we represent integral demands as sets of triples $(u, v, \id)$, where $\id$ is a unique identifier specific to the triplet, with a particular pair $u, v$ appearing as many times as how much $(u, v)$-demand there is.

\textbf{Routing a demand.} A multiset of paths $P$ on a graph $G$ \textit{routes} an integral demand $D$, if the set contains exactly $D(u, v)$-many $(u, v)$-paths. The \textit{congestion} $\cong(P, e) := |\{p \in P : e \in p\}|$ of an edge $e$ is the total number of paths using it, and the congestion $\cong(P) := \max_e \cong(P, e)$ of the path set is the maximum congestion of an edge. The \textit{dilation} $\dil(P) := \max_{p \in P} |p|$ of the path set is the length of the longest path in the set. For routing algorithms, we denote by $P(\id)$ the $(u, v)$-path assigned to the triplet $(u, v, \id)$. 

A multicommodity flow $F$ assigning nonnegative flow values $F(p) \geq 0$ to paths $p$ in $G$ routes a demand $D$, if for all $u, v$, the total flow value over $(u, v)$-paths in $G$ is exactly $D(u, v)$, i.e. $\sum_{(u, v)\text{-path }p} F(p) = D(u, v)$. The congestion $\cong(F, e) := \sum_{p : e \in p} F(p)$ of an edge is the total flow value of paths using it, and the congestion $\cong(F) := \max_e \cong(F, e)$ of the path set is the maximum congestion of an edge. We say the flow is \textit{length-$h$} if every path in the support of the flow has length at most $h$, i.e. $\max_{p \in \supp(F)} \mathrm{len}(p) \leq h$.

\textbf{Routers.}  A graph $G = (V, E)$ is a \textit{$h$-length $\kappa$-congestion router} for a node weighting $A$ on $V$, if for any $A$-respecting demand $D$, there is a multicommodity flow of congestion $\kappa$ and length $h$ routing $D$.

\textbf{Oblivious Routing.} An oblivious routing $R = \{R(s, t)\}_{s, t \in V}$ on a graph $G$ is a collection of distributions of $(s, t)$-paths $R(s, t)$ between vertex pairs $(s, t)$.
    
For a demand $D$, the congestion $\cong(R, D, e)$ of an edge $e$ is defined as the expected number of paths using the edge in the path set $P$ formed by sampling $D(s, t)$ paths from $R(s, t)$ for every $(s, t) \in \supp(D)$, and the congestion $\cong(R, D) := \max_e \cong(R, D, e)$ is the maximum congestion of an edge.
    
For a class of demands $\calD$ (for example, the set of degree-respecting demands on $G$), the oblivious routing has congestion $\cong(R, \calD, e) := \max_{D \in \calD} \cong(R, D, e)$. The dilation $\dil(R) := \max_{p \in \supp(R)} |p|$ of the oblivious routing is the length of the longest path in its support.

\textbf{String Notation.} For an alphabet $\Sigma$ and nonnegative integer $d$, we denote by $\Sigma^{d}$ the set of length-$d$ strings with characters from the alphabet $\Sigma$, and likewise by $\Sigma^{\leq d}$ the set of strings of length at most $d$ (including the empty string), and by $\Sigma^{< d}$ the set of strings of length strictly less than $d$. We denote the concatenation operation on strings by $+$, i.e. for a string $\sigma \in \Sigma^{d}$ and character $i \in \Sigma$, the string resulting from concatenating $i$ to $\sigma$ is $\sigma + i \in \Sigma^{d + 1}$. We denote the longest common prefix of a pair of strings $\sigma, \sigma'$ by $\lcp(\sigma, \sigma')$.


%
\section{Semi-Hypercubes and Pruning Overview}

\subsection{Semi-Hypercubes}



A $d$-dimensional hypercube for $d \geq 1$ can be defined recursively as the graph formed by taking two copies of a $(d - 1)$-dimensional hypercube, and adding a perfect matching between them connecting the two copies of each vertex. Each vertex in a $d$-dimensional hypercube has degree $d$.

A $(k, d)$-semi-hypercube for $d \geq 1$ is a graph of size $k^d$ constructed similarly, by instead taking $k$ individual $(k, d - 1)$-semi-hypercubes, then adding a perfect matching between the vertex sets of each pair of these smaller semi-hypercubes. Thus, each vertex in a $(k, d)$-semi-hypercube has degree $(k - 1) d$.

Other than the "branching factor" $k$ not being restricted to two, a semi-hypercube is more general than a hypercube in that the matchings connecting the individual smaller semi-hypercubes can be arbitrary perfect matchings, rather than being restricted to connecting copies of the same vertex. The individual $(k, d - 1)$-semi-hypercubes do not necessarily have to be identical to each other either.

For the definition of a $(k, d)$-semi-hypercube below, the size-$k^d$ vertex set of the semi-hypercube is fixed to be specifically $[k]^d$, i.e. the set of $d$-length strings of elements of $[k] = \{1, 2, \dots, k\}$. For a string $\sigma \in [k]^{\leq d}$, we refer to the subset $V_\sigma := \{v \in V : \text{$\sigma$ is a prefix of $v$}\}$ as the \textit{cluster} $\sigma$. We refer to the clusters $\sigma + i$ for $i \in [k]$ as the \textit{child clusters} of the cluster $\sigma$. Each cluster $\sigma$ of \textit{depth} $d' = |\sigma|$ will correspond to a $(k, d - d')$-semi-hypercube, up to vertex labeling, made up of the child clusters $\sigma + i$ that correspond to $(k, d - d' - 1)$-semi-hypercubes, and a complete matching between each pair of those children.

\begin{definition}[$(k, d)$-semi-hypercube]
For any positive integers $k, d$, an undirected graph $G = (V, E)$ is a $(k, d)$-semi-hypercube if
\begin{itemize}
    \item The vertex set $V = [k]^d$ consists of $d$-tuples of elements of $[k] = \{1, \dots, k\}$.
    \item For every cluster $\sigma \in [k]^{< d}$, the edge set $E_{\sigma, i, i'} := E[V_{\sigma + i}, V_{\sigma + i'}]$ between each pair of distinct child clusters $\sigma + i$ and $\sigma + i'$ ($i, i' \in [k]$) is a perfect matching. 
\end{itemize}
\end{definition}

\begin{figure}[H]
    \centering
    \begin{tikzpicture}[main/.style={draw, circle, minimum size=0.2cm, fill=black!50}, cluster/.style={draw, circle, minimum size=2cm, color=black!50}, edges/.style={draw}, txt/.style = {draw=none}]

    \coordinate(C1) at (-2, -2);
    \coordinate(C2) at (2, -2);
    \coordinate(C3) at (2, 2);
    \coordinate(C4) at (-2, 2);
    
    \draw[color=black!50] (0, 0) ellipse (4cm and 4cm);
    
    \node[cluster] (cluster1) at (C1) {};
    \node[cluster] (cluster2) at (C2) {};
    \node[cluster] (cluster3) at (C3) {};
    \node[cluster] (cluster4) at (C4) {};
    
    \draw[-] (cluster1) to[out=-15, in=185] (cluster2);
    \draw[-] (cluster1) to[out=-5, in=175] (cluster2);
    \draw[-] (cluster1) to[out=5, in=165] (cluster2);
    \draw[-] (cluster1) to[out=15, in=195] (cluster2);

    \draw[-] (cluster1) to[out=30, in=210] (cluster3);
    \draw[-] (cluster1) to[out=40, in=220] (cluster3);
    \draw[-] (cluster1) to[out=50, in=240] (cluster3);
    \draw[-] (cluster1) to[out=60, in=230] (cluster3);

    \draw[-] (cluster1) to[out=75, in=285] (cluster4);
    \draw[-] (cluster1) to[out=85, in=265] (cluster4);
    \draw[-] (cluster1) to[out=95, in=275] (cluster4);
    \draw[-] (cluster1) to[out=105, in=255] (cluster4);
    
    \draw[-] (cluster2) to[out=75, in=275] (cluster3);
    \draw[-] (cluster2) to[out=85, in=265] (cluster3);
    \draw[-] (cluster2) to[out=95, in=285] (cluster3);
    \draw[-] (cluster2) to[out=105, in=255] (cluster3);
    
    \draw[-] (cluster4) to[out=-15, in=195] (cluster3);
    \draw[-] (cluster4) to[out=-5, in=175] (cluster3);
    \draw[-] (cluster4) to[out=5, in=165] (cluster3);
    \draw[-] (cluster4) to[out=15, in=185] (cluster3);

    \draw[-] (cluster2) to[out=120, in=320] (cluster4);
    \draw[-] (cluster2) to[out=130, in=310] (cluster4);
    \draw[-] (cluster2) to[out=140, in=300] (cluster4);
    \draw[-] (cluster2) to[out=150, in=330] (cluster4);

    \node[main] (c11) at ($(C1)+(-0.3, -0.3)$) {};
    \node[main] (c12) at ($(C1)+(-0.3, 0.3)$) {};
    \node[main] (c13) at ($(C1)+(0.3, -0.3)$) {};
    \node[main] (c14) at ($(C1)+(0.3, 0.3)$) {};
    \draw (c11) -- (c12);
    \draw (c11) -- (c13);
    \draw (c11) -- (c14);
    \draw (c12) -- (c13);
    \draw (c12) -- (c14);
    \draw (c13) -- (c14);

    \node[main] (c21) at ($(C2)+(-0.3, -0.3)$) {};
    \node[main] (c22) at ($(C2)+(-0.3, 0.3)$) {};
    \node[main] (c23) at ($(C2)+(0.3, -0.3)$) {};
    \node[main] (c24) at ($(C2)+(0.3, 0.3)$) {};
    \draw (c21) -- (c22);
    \draw (c21) -- (c23);
    \draw (c21) -- (c24);
    \draw (c22) -- (c23);
    \draw (c22) -- (c24);
    \draw (c23) -- (c24);

    \node[main] (c31) at ($(C3)+(-0.3, -0.3)$) {};
    \node[main] (c32) at ($(C3)+(-0.3, 0.3)$) {};
    \node[main] (c33) at ($(C3)+(0.3, -0.3)$) {};
    \node[main] (c34) at ($(C3)+(0.3, 0.3)$) {};
    \draw (c31) -- (c32);
    \draw (c31) -- (c33);
    \draw (c31) -- (c34);
    \draw (c32) -- (c33);
    \draw (c32) -- (c34);
    \draw (c33) -- (c34);
    
    \node[main] (c41) at ($(C4)+(-0.3, -0.3)$) {};
    \node[main] (c42) at ($(C4)+(-0.3, 0.3)$) {};
    \node[main] (c43) at ($(C4)+(0.3, -0.3)$) {};
    \node[main] (c44) at ($(C4)+(0.3, 0.3)$) {};
    \draw (c41) -- (c42);
    \draw (c41) -- (c43);
    \draw (c41) -- (c44);
    \draw (c42) -- (c43);
    \draw (c42) -- (c44);
    \draw (c43) -- (c44);

    
    \end{tikzpicture}
    \caption{A $(4, 2)$-semi-hypercube. The matching edges are not drawn all the way between the edges' endpoint vertices for clarity. Each gray circle pictured is a cluster containing the vertices (filled gray) inside it, with each vertex also being referred to as a 'leaf cluster' containing only the vertex itself.} \label{fig:semi-hypercube-example}
\end{figure}
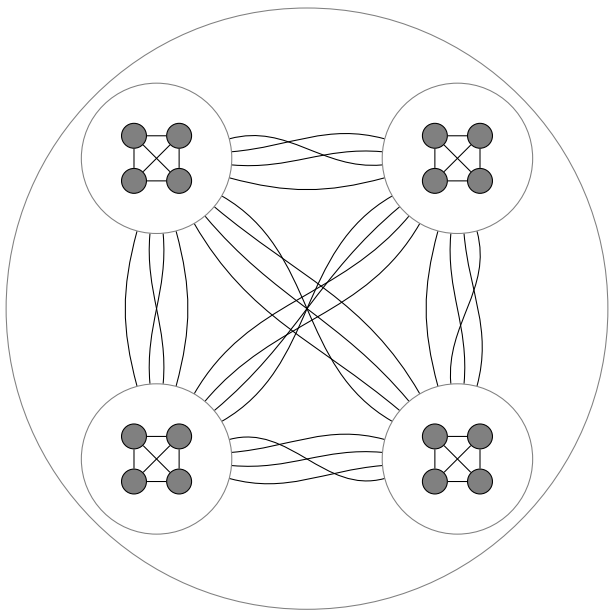

We call the \textit{cluster tree} of the $(k, d)$-semi-hypercube the tree with vertices $\sigma \in [k]^{\leq d}$, the empty string as the root vertex, the actual vertices $[k]^d$ of the semi-hypercube as leaf vertices, and the $k$ child clusters $\sigma + i$ as the children of each non-leaf $\sigma$. For a cluster $\sigma$, we call a cluster $\sigma''$ a \textit{descendant cluster} of $\sigma$ if $\sigma''$ is a (strict) descendant of $\sigma$ in this tree (i.e. $\sigma$ is a strict prefix of $\sigma''$), and a cluster $\sigma'$ an \textit{ancestor cluster} of $\sigma$ if $\sigma$ is a descendant of $\sigma'$ (or equivalently, $\sigma'$ is a strict prefix of $\sigma$).

Valiant-like oblivious routing works on a $(k, d)$-semi-hypercube as it does on a hypercube, with length $\bigO(d)$ and congestion $\bigO(d)$ on unit demands: pick a random vertex as a midpoint, then greedily route from it to both endpoints. To construct the greedy path from a vertex $u$ to a vertex $v$, walk down the cluster tree towards $v$ from the root cluster until you find a cluster $\sigma$ for which the current path endpoint $u'$ and the destination vertex $v$ are in different child clusters $V_{\sigma + i} \ni u'$ and $V_{\sigma + i'} \ni v$. Then, there is a unique edge incident to the current path endpoint in the complete matching between $V_{\sigma + i}$ and $V_{\sigma + i'}$. After concatenating it to the path, the current path endpoint and $v$ are in the same child cluster $\sigma + i'$.

In this paper, we consider two families of induced subgraphs of a $(k, d)$-semi-hypercube, one of which is a subset of the other. Each family has an additional parameter $\tau \in [0, 1)$, the \textit{criticality threshold}. The family satisfying a stricter condition is that of \textit{noncritical $(k, d, \tau)$-semi-hypercubes}. For $\tau$ small enough and $k$ large enough relative to $d$, these graphs are routers of the same asymptotic quality as a $(k, d)$-semi-hypercube, and there is a simple vertex pruning algorithm with an amortized pruning ratio for maintaining that a graph that is initially a $(k, d)$-semi-hypercube remains a noncritical $(k, d, \tau)$-semi-hypercube. The more general family is that of \textit{$(k, d, \tau)$-semi-hypercubes}. These graphs lose some quality as routers compared to noncritical semi-hypercubes, but can be maintained under pruning with a \emph{worst-case} pruning ratio.

\subsection{Amortized Pruning of Noncritical Semi-Hypercubes} \label{sec:pruning-amortized}

A \textit{noncritical $(k, d, \tau)$-semi-hypercube} is an induced subgraph $G[V]$ of a $(k, d)$-semi-hypercube $G = ([k]^d, E)$ satisfying a certain property that guarantees that an easy-to-sample-from oblivious routing of the same asymptotic length and congestion $\bigO(d)$ as on a $(k, d)$-semi-hypercube exists on it. This property concerns the remaining sizes $|V_\sigma|$ of clusters: specifically, every cluster $\sigma$ must either be empty, or contain at least a $(1 - \tau)$-fraction of the vertices it does in $G$. We call a cluster violating these conditions a \textit{$\tau$-critical cluster}.
\begin{definition}[$\tau$-critical cluster]
For a criticality threshold $\tau \in [0, 1)$, a cluster $\sigma$ in an induced subgraph $G[V]$ of a $(k, d)$-semi-hypercube $G = ([k]^d, E)$ is \textit{$\tau$-critical} if it is nonempty and has lost strictly more than a $\tau$-fraction of the vertices it originally had in $G$, i.e. $0 < |V_\sigma| < (1 - \tau) k^{d - |\sigma|}$.
\end{definition}

\begin{restatable}[Noncritical $(k, d, \tau)$-Semi-Hypercube]{definition}{noncritsemihypercube} \label{def:noncrit-semi-hyper-cube}
For $k, d \in \mathbb{N}$ and $\tau \in [0, 1)$, a \textit{noncritical $(k, d, \tau)$-semi-hypercube} is an induced subgraph of a $(k, d)$-semi-hypercube with no $\tau$-critical clusters.
\end{restatable}

We show in \Cref{sec:rand-routing} the above-claimed oblivious routing property, stated formally below in \Cref{lem:obliv-routing-noncritical}. 

\begin{restatable*}{lemma}{oblivroutingnoncritical}\label{lem:obliv-routing-noncritical}
    Let $G = (V, E)$ be a \textbf{noncritical} $(k, d, \tau)$-hypercube for $\frac{1}{\tau}, k \geq 27d$. Then, there exists an oblivious routing $R$ on $G$ of length $\bigO(d)$ that can route any unit demand on $V$ with congestion $\bigO(d)$. A path from this oblivious routing can be sampled in expected running time $\bigO(d)$.
\end{restatable*}

We now consider \emph{vertex} pruning for maintaining that a graph that is initially a $(k, d)$-semi-hypercube remains a noncritical $(k, d, \tau)$-semi-hypercube. In this setting, there is a sequence of updates, each giving a vertex $v \in V$ to be deleted. Then, the pruning scheme can select an additional subset of vertices $V^{\mathrm{prune}} \subseteq (V - v)$ to prune. The vertices $\{v\} \cup V^{\mathrm{prune}}$ (and all edges incident to them) are then permanently removed from the graph. The graph $G[V]$ should always be a $(k, d, \tau)$-semi-hypercube after each update.

The quality of the pruning scheme to be optimized is the \textit{pruning ratio}, of which there are two variants. The \textit{worst-case} pruning ratio of a scheme is the maximum number of pruned vertices in any single update, i.e. a bound on the \emph{maximum} size of $V^{\mathrm{prune}}$. The \textit{amortized} pruning ratio of a scheme is the maximum ratio of total pruned vertices to the number of updates, i.e. a bound on the \emph{average} size of $V^{\mathrm{prune}}$.

For maintaining that a graph remains a noncritical semi-hypercube, we cannot obtain a nontrivial worst-case pruning ratio. However, there is an obvious pruning algorithm with an amortized pruning ratio guarantee: once a cluster becomes $\tau$-critical, prune all of its remaining vertices. This is clearly necessary, and optimal as $\tau$-criticality is a monotone property on the remaining vertices within a cluster.

This strategy obtains an amortized pruning ratio of $\tau^{-d} - 1$, which is easy to show by induction on $d$. Notice that the deletions and pruning on each child cluster of the root cluster are done independently of the other child clusters, until the point when the root cluster has lost strictly more than a $\tau$-fraction of its original $n = k^d$ vertices, at which point every remaining vertex in the graph is pruned. By induction, before that occurs, the amortized pruning ratio is at worst $\tau^{-(d - 1)} - 1$, i.e. the total number of vertices removed is at most $\tau^{-(d - 1)}$ times the number of deletions so far. Thus, for the root cluster to shrink to that size, strictly more than $(\tau \cdot n) / \tau^{-(d - 1)} = \tau^d \cdot n$ vertices have to have been deleted, giving an amortized pruning ratio of $\tau^{-d} - 1$, as desired. 

\begin{lemma}
    For given $k, d \in \mathbb{N}$ and $\tau \in (0, 1)$, there is a deterministic vertex pruning scheme with amortized pruning ratio $(\tau^{-d} - 1)$ for maintaining that a graph that is initially a $(k, d)$-semi-hypercube remains a noncritical $(k, d, \tau)$-semi-hypercube.
\end{lemma}

The above simple scheme easily generalizes to the \textit{batched} worst-case pruning setting with a limited number of batches. In this setting, the updates come in a small number $b$ of batches. Each batch gives a subset of vertices $V^{\mathrm{del}} \subseteq V$ to be deleted, after which the pruning scheme selects a subset of vertices $V^{\mathrm{prune}} \subseteq V \setminus V^{\mathrm{del}}$ to be pruned, after which the vertices $V^{\mathrm{del}} \cup V^{\mathrm{prune}}$ (and all edges incident to them) are removed from the graph. The desired property of the graph should hold \emph{after every batch}, i.e. the deletions in a batch do not have to be processed one-by-one. The worst-case pruning ratio in this setting is the maximum ratio of pruned vertices to deleted vertices $|V^{\mathrm{prune}}| / |V^{\mathrm{del}}|$ in any single batch.


The simple adaptation of the amortized scheme to this setting is to start with a stricter criticality threshold, then increase it with every batch until it reaches the desired threshold $\tau$ after the last batch. If the criticality threshold starts at and is increased by $\tau / b$ after each batch, one obtains a pruning ratio of $(\tau / b)^{-d} - 1$ with similar analysis as the amortized algorithm. \Cref{lem:batched-pruning} is proven in \Cref{sec:amortized-formal-pruning} for completeness.

\begin{restatable}[Batched pruning of a noncritical semi-hypercube]{lemma}{batchedpruning} \label{lem:batched-pruning}
For given $k, d, b \in \mathbb{N}$ and $\tau \in (0, 1)$, there is a deterministic algorithm with a worst-case pruning ratio $(\tau / b)^{-d} - 1$, and work and depth of $\bigO(kd \cdot (\tau / b)^{-d} \cdot |V^{\mathrm{del}}|)$ and $\bigO(d \log n)$ respectively per update for maintaining that a graph that is initially a $(k, d)$-semi-hypercube remains a noncritical $(k, d, \tau)$-semi-hypercube in the batched vertex pruning setting with up to $b$ batches.
\end{restatable}

\subsection{Worst-Case Pruning of Semi-Hypercubes}\label{sec:selfpruning-overview}


Notice that the pruning ratio of the batched pruning algorithm of \Cref{lem:batched-pruning} becomes trivial as the batch bound $b$ increases to $k$. In this subsection, we now consider pruning with a worst-case pruning ratio bound, \emph{without any limit on the number of updates}. It is easy to see that unlike in the amortized or batched setting, maintaining a \emph{noncritical} $(k, d, \tau)$-semi-hypercube with a nontrivial worst-case pruning ratio is impossible: in particular, in the update where the root cluster becomes $\tau$-critical, at least $\tau n$ vertices have to be pruned. 
Thus, it is necessary to relax the noncriticality condition.

As the first change from noncritical to general $(k, d, \tau)$-semi-hypercubes, we now allow $\tau$-critical clusters, but require that each cluster has \emph{at most one} critical child cluster. 

\textbf{Maintaining $\leq 1$ Critical Child.} Recall a $\tau$-critical cluster is a \emph{nonempty} cluster with strictly less than $(1 - \tau) n_0$ vertices (where $n_0$ is its initial size). In a decremental setting, a critical cluster only becomes noncritical upon the loss of its final vertex. To maintain that a fixed cluster has at most one critical child, we thus need to empty out any child cluster small enough to be critical before any other child cluster becomes too small. This gives rise to the following intuitive strategy: whenever a vertex is removed from some cluster $\sigma$, if that vertex was not removed from the smallest child cluster of $\sigma$, prune some number of vertices from the smallest child cluster of $\sigma$.

Our strategy is essentially this.  We maintain for each cluster $\sigma$ a \textit{target child} $t_\sigma \in [k]$ initialized arbitrarily, and maintain that at any time, no child cluster of $\sigma$ other than possibly the target child cluster $\sigma + t_\sigma$ can be critical. Whenever a vertex is removed from a child cluster of $\sigma$ other than the target child cluster, we perform $\rho$ prunes (where $\rho$ is a parameter whose value is determined later) on the target child cluster. Whenever a cluster's target child cluster becomes empty, the smallest nonempty child cluster is chosen as the new target child. 

Specifically which vertices should be pruned? From the point of view of a cluster that wants to maintain its own at-most-one-critical-child property, the individual vertices in any child cluster are all interchangeable. Thus, if a cluster wants to prune a vertex from some specific child cluster, the exact vertex to be pruned can be determined to benefit that child.

This motivates having a function $\Call{Trim}{\sigma}$, that on a non-leaf cluster $\sigma$, recursively calls $\Call{Trim}{\sigma + t_\sigma}$ to prune a vertex from the target child cluster of $\sigma$, and on a leaf $\sigma$, prunes that leaf.

Note that the only clusters whose child cluster sizes are affected by the deletion of a vertex are the ancestor clusters of that vertex. We can iterate over these ancestors from leaf to root, performing these trims.

Thus, for $\rho$ large enough, we obtain a worst-case pruning algorithm with pruning ratio $(\rho + 1)^d - 1$ for maintaining that each cluster has at most one critical child cluster: while processing descendant clusters of $\sigma$, there is some total number of deleted and pruned vertices. To maintain the at-most-one critical child property of $\sigma$, we make up to $\rho$ times that many calls to $\Call{Trim}{\sigma}$, multiplying the total amount by $(\rho + 1)$. 

\begin{algorithm}[H]
    \caption{Worst-case pruning for maintaining $\leq 1$ critical child cluster}
    \label{alg:strong-self-pruning-intro}
    \begin{algorithmic}[1]
        \LeftComment{Set the new target child cluster of $\sigma$ (to the smallest nonempty child).}
        \Function{Retarget}{$\sigma$}
            \State $t_\sigma \gets \arg\min_{i \in [k] : V_{\sigma + i} \neq \emptyset} |V_{\sigma + i}|$
        \EndFunction
        \State
        \Function{Trim}{$\sigma$}
            \If{$\sigma$ is a leaf}
                \State $V \gets V \setminus \{\sigma\}$
                \State \Return $\sigma$
            \Else
                \State Let $v_{\mathrm{trim}} := \Call{Trim}{\sigma + t_\sigma}$
                \If{$V_{\sigma + t_\sigma} = \emptyset$} $\Call{Retarget}{\sigma}$ \EndIf
                \State \Return $v_{\mathrm{trim}}$
            \EndIf
        \EndFunction
        \State
        \LeftComment{Delete vertex $v$, return pruned vertices}
        \Function{Delete}{$v$}
            \State $V \gets V \setminus \{v\}$
            \State Let $\sigma \gets v$, $V^\mathrm{prune} \gets \emptyset$
            \While{$|\sigma| > 0$}
                \State $\sigma \gets \Call{Parent}{\sigma}$
                \If{$v$ is a descendant of $\sigma + t_\sigma$}
                    \If{$V_{\sigma + t_\sigma} = \emptyset$} $\Call{Retarget}{\sigma}$ \EndIf
                \Else
                    \State Let $\mathrm{trim\_cou} := \Call{min}{\rho \cdot (|V^\mathrm{prune}| + 1), |V_\sigma|}$
                    \For{$\mathrm{trim\_cou}$ times} $V^\mathrm{prune} \gets V^\mathrm{prune} \cup \{\Call{Trim}{\sigma}\}$ \EndFor
                \EndIf
            \EndWhile
            \State \Return $V^\mathrm{prune}$
        \EndFunction
    \end{algorithmic}
\end{algorithm}

It remains to bound the necessary value of $\rho$. We will show that it suffices to select $\rho = O(\log(k) / \tau)$. For this purpose, fix some cluster $\sigma$ with initial child cluster sizes $n_0 := k^{d - |\sigma| - 1}$, and let $x = O(n_0 / \rho)$ be the maximum number of vertices that can be removed from child clusters other than the current target child cluster before the current target child cluster becomes empty.

Suppose that the strategy of an adversary trying to create a critical non-target child cluster is to always simply delete a vertex from the smallest nonempty non-target child cluster. Then, we would need $x \leq \tau n_0$ to hold, for which it even suffices to select $\rho = O(1 / \tau)$.

However, there exists a better strategy for the adversary: each round, delete a vertex from the \emph{largest} non-target child cluster. Then, when the first target child cluster becomes empty, each other child cluster has lost $\Omega\left(x / (k - 1)\right)$ vertices. Repeating the argument, when all but two of the $k$ initial child clusters have become empty, each child cluster has lost $\Omega(x \log k)$ vertices. Thus, it is necessary to select $\rho = \Omega(\log(k) / \tau)$.

To show that $\rho = O(\log(k) / \tau)$ suffices, consider the moment when $\Call{Retarget}{\sigma}$ is called while there are $k'$ nonempty child clusters, and denote by $B_{k', j}$ the maximum total number of missing vertices among some $j$ specific nonempty child clusters of $\sigma$ at that moment. We want $B_{k', 1} \leq \tau n_0$ to hold for all $k' \in [k]$.

Initially, $B_{k, j} = 0$ for all $j$. For all $k' \in [k - 1]$ and $j \in [k']$, since $\Call{Retarget}{\sigma}$ always selects the smallest nonempty child cluster as the target child cluster, we have
\begin{equation}\label{eq:missing-vertex-bound}
    B_{k', j} \leq x + \frac{j}{j + 1} \cdot B_{k' + 1, j + 1}
\end{equation}
Thus, $B_{k', 1} \leq x \cdot H_{k - k'} = O(x \log k)$ where $H_i$ is the $i$\textsuperscript{th} harmonic number, and $\rho = \bigO(\log(k) / \tau)$ suffices.

We have now shown that for a cluster $\sigma$, if $\rho = \bigO(\log(k) / \tau)$ prunes are performed on the target child cluster of $\sigma$ whenever a vertex is removed from a non-target child cluster of $\sigma$, the cluster $\sigma$ will always have at most one critical child cluster. Thus, we have a simple pruning scheme in \Cref{alg:strong-self-pruning-intro} for maintaining that each cluster has at most one critical child cluster, with a worst-case pruning ratio of $(\rho + 1)^d - 1 = \bigO(\log(k) / \tau)^d$.

\textbf{Connectivity of the Critical Child.} Unfortunately, allowing each cluster even just one critical child breaks the routing properties that hold for a noncritical semi-hypercube. Even though each other nonempty child cluster of a cluster $\sigma$ is noncritical and thus has lost at most a $\tau$-fraction of its vertices, if those vertices were exactly those with an edge to the remaining vertices in the critical child cluster of $\sigma$, each vertex in the critical child cluster could have cluster-$\sigma$-degree zero: see \Cref{fig:crit-child-degree}.

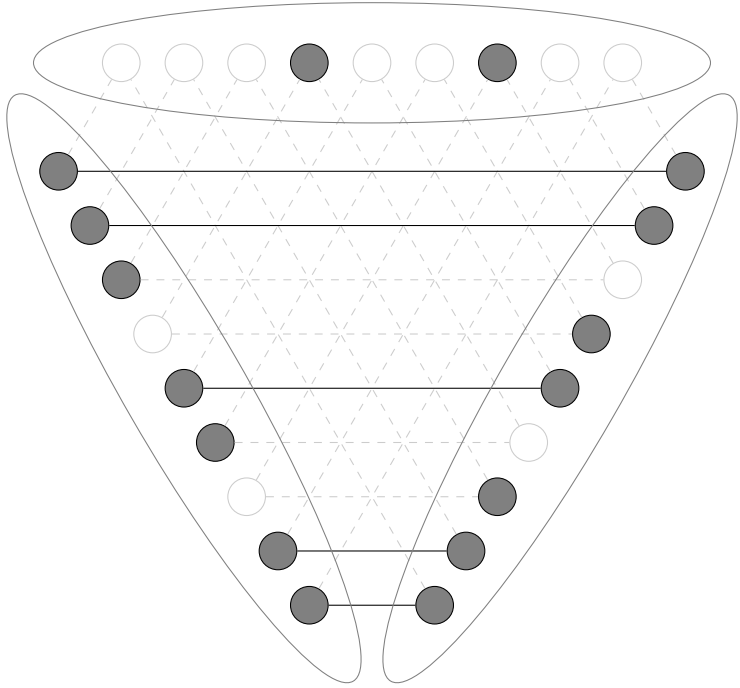
\begin{figure}[h]
    \centering
    \begin{tikzpicture}[main/.style={draw, circle, minimum size=0.5cm}, minor/.style={draw, circle, minimum size=0.5cm}, edges/.style={draw}, txt/.style = {draw=none}]

    \coordinate(T11) at (-5, 4);
    \coordinate(T21) at (5, 4);
    \coordinate(T12) at (0, -4.66);

    \coordinate(C11) at ($2/12*(T11)+10/12*(T12)$);
    \coordinate(C19) at ($10/12*(T11)+2/12*(T12)$);

    \coordinate(C21) at ($2/12*(T21)+10/12*(T12)$);
    \coordinate(C29) at ($10/12*(T21)+2/12*(T12)$);
    
    \coordinate(C31) at ($2/12*(T11)+10/12*(T21)$);
    \coordinate(C39) at ($10/12*(T11)+2/12*(T21)$);

    \coordinate(C12) at ($1/8*(C11)+7/8*(C19)$);
    \coordinate(C13) at ($2/8*(C11)+6/8*(C19)$);
    \coordinate(C14) at ($3/8*(C11)+5/8*(C19)$);
    \coordinate(C15) at ($4/8*(C11)+4/8*(C19)$);
    \coordinate(C16) at ($5/8*(C11)+3/8*(C19)$);
    \coordinate(C17) at ($6/8*(C11)+2/8*(C19)$);
    \coordinate(C18) at ($7/8*(C11)+1/8*(C19)$);
    \coordinate(C22) at ($1/8*(C21)+7/8*(C29)$);
    \coordinate(C23) at ($2/8*(C21)+6/8*(C29)$);
    \coordinate(C24) at ($3/8*(C21)+5/8*(C29)$);
    \coordinate(C25) at ($4/8*(C21)+4/8*(C29)$);
    \coordinate(C26) at ($5/8*(C21)+3/8*(C29)$);
    \coordinate(C27) at ($6/8*(C21)+2/8*(C29)$);
    \coordinate(C28) at ($7/8*(C21)+1/8*(C29)$);
    \coordinate(C32) at ($1/8*(C31)+7/8*(C39)$);
    \coordinate(C33) at ($2/8*(C31)+6/8*(C39)$);
    \coordinate(C34) at ($3/8*(C31)+5/8*(C39)$);
    \coordinate(C35) at ($4/8*(C31)+4/8*(C39)$);
    \coordinate(C36) at ($5/8*(C31)+3/8*(C39)$);
    \coordinate(C37) at ($6/8*(C31)+2/8*(C39)$);
    \coordinate(C38) at ($7/8*(C31)+1/8*(C39)$);

    \node[main, fill=black!50] (c11) at (C11) {};
    \node[main, fill=black!50] (c12) at (C12) {};
    \node[main, fill=black!50] (c13) at (C13) {};
    \node[main, color=black!20] (c14) at (C14) {};
    \node[main, fill=black!50] (c15) at (C15) {};
    \node[main, fill=black!50] (c16) at (C16) {};
    \node[main, color=black!20] (c17) at (C17) {};
    \node[main, fill=black!50] (c18) at (C18) {};
    \node[main, fill=black!50] (c19) at (C19) {};

    \node[main, fill=black!50] (c21) at (C21) {};
    \node[main, fill=black!50] (c22) at (C22) {};
    \node[main, color=black!20] (c23) at (C23) {};
    \node[main, fill=black!50] (c24) at (C24) {};
    \node[main, fill=black!50] (c25) at (C25) {};
    \node[main, color=black!20] (c26) at (C26) {};
    \node[main, fill=black!50] (c27) at (C27) {};
    \node[main, fill=black!50] (c28) at (C28) {};
    \node[main, fill=black!50] (c29) at (C29) {};

    \node[main, color=black!20] (c31) at (C31) {};
    \node[main, color=black!20] (c32) at (C32) {};
    \node[main, color=black!20] (c33) at (C33) {};
    \node[main, fill=black!50] (c34) at (C34) {};
    \node[main, color=black!20] (c35) at (C35) {};
    \node[main, color=black!20] (c36) at (C36) {};
    \node[main, fill=black!50] (c37) at (C37) {};
    \node[main, color=black!20] (c38) at (C38) {};
    \node[main, color=black!20] (c39) at (C39) {};
    


    
    \draw[dashed, color=black!20] (c11) -- (c31);
    \draw[dashed, color=black!20] (c12) -- (c32);
    \draw[dashed, color=black!20] (c13) -- (c33);
    \draw[dashed, color=black!20] (c14) -- (c34);
    \draw[dashed, color=black!20] (c15) -- (c35);
    \draw[dashed, color=black!20] (c16) -- (c36);
    \draw[dashed, color=black!20] (c17) -- (c37);
    \draw[dashed, color=black!20] (c18) -- (c38);
    \draw[dashed, color=black!20] (c19) -- (c39);
    
    \draw[dashed, color=black!20] (c29) -- (c31);
    \draw[dashed, color=black!20] (c28) -- (c32);
    \draw[dashed, color=black!20] (c27) -- (c33);
    \draw[dashed, color=black!20] (c26) -- (c34);
    \draw[dashed, color=black!20] (c25) -- (c35);
    \draw[dashed, color=black!20] (c24) -- (c36);
    \draw[dashed, color=black!20] (c23) -- (c37);
    \draw[dashed, color=black!20] (c22) -- (c38);
    \draw[dashed, color=black!20] (c21) -- (c39);

    \draw (c11) -- (c21);
    \draw (c12) -- (c22);
    \draw[dashed, color=black!20] (c13) -- (c23);
    \draw[dashed, color=black!20] (c14) -- (c24);
    \draw (c15) -- (c25);
    \draw[dashed, color=black!20] (c16) -- (c26);
    \draw[dashed, color=black!20] (c17) -- (c27);
    \draw (c18) -- (c28);
    \draw (c19) -- (c29);

    \draw[color=black!50] (C35) ellipse (4.5cm and 0.8cm);
    \draw[color=black!50, rotate around={60:(C25)}] (C25) ellipse (4.5cm and 0.8cm);
    \draw[color=black!50, rotate around={120:(C15)}] (C15) ellipse (4.5cm and 0.8cm);
    
    \end{tikzpicture}
    \caption{A subgraph of a $(3, 3)$-semi-hypercube induced on the filled vertices, where no vertex in the small child cluster (top) has edges to the cluster's sibling clusters, despite each of those siblings missing only two vertices, as the missing vertices happen to be precisely those that were initially incident to the remaining vertices in the small child. The child clusters' internal structures are not shown.} \label{fig:crit-child-degree}
\end{figure}

We need an additional property guaranteeing the connectivity of the critical child cluster to the other children of its parent. Note that between any pair of noncritical child clusters, there is an almost-complete matching, with at least a $(1 - 2\tau)$-fraction of vertices in each child being matched to some vertex in the other. One might hope for the same property between the critical child and each of its individual nonempty siblings, with at least a large fraction of the vertices in the critical child being matched to some vertex in the sibling, but unfortunately, this turns out to be too strong of a property to maintain. Instead, we settle for requiring this \emph{on average} over the sibling clusters, which is still sufficient for routing out of the critical child cluster with low congestion.

With this, we have arrived at the definition of a $(k, d, \tau)$-semi-hypercube:

\begin{restatable}[$(k, d, \tau)$-semi-hypercube]{definition}{semihypercube} \label{def:semi-hyper-cube}
For $k, d \in \mathbb{N}$ and $\tau \in (0, 1)$, an induced subgraph $G[V]$ of a $(k, d)$-semi-hypercube $G = ([k]^d, E)$ is a $(k, d, \tau)$-semi-hypercube, if for every cluster $\sigma \in [k]^{< d}$,
\begin{enumerate}
    \item At most one child cluster $\sigma + t$ of the cluster $\sigma$ is $\tau$-critical, i.e. is of size $|V_{\sigma + t}| < (1 - \tau) k^{d - |\sigma| - 1}$.
    \item If such a child cluster $\sigma + t$ exists, the average cluster-$\sigma$-degree of $V_{\sigma + t}$ is at least $\frac{9}{10} k'$ where $k'$ is the number of nonempty, noncritical child clusters of $\sigma$, i.e. $|E[V_{\sigma + t}, V_{\sigma} \setminus V_{\sigma + t}]| \geq \frac{9}{10} \cdot k' \cdot |V_{\sigma + t}|$. 
\end{enumerate}
\end{restatable}

We later show in \Cref{thm:obliv-routing-general} that even with these relaxations to a noncritical $(k, d, \tau)$-semi-hypercube, an easy-to-sample-from oblivious routing of length $\bigO(d^3)$ and congestion on unit demands of $\bigO(kd \cdot 10^d)$ exists on any $(k, d, \tau)$-semi-hypercube with $\tau$ small enough and $k$ large enough relative to $d$. While the congestion of the oblivious routing is worse than with a non-critical semi-hypercube, the loss in length is still crucially only polynomial.

The fraction $\frac{9}{10}$ in the definition is a fixed constant rather than being $(1 - 2 \tau)$ or a parameter, as a larger fraction, even as large as $(1 - n^{-(1 - \epsilon)})$ for constant $\epsilon > 0$, would not improve the routing properties of the semi-hypercube (this can be observed later from the proof of \Cref{lem:obliv-routing-lower-bound}).

We now return to considering how to maintain through pruning with a worst-case pruning ratio guarantee that the graph remains a $(k, d, \tau)$-semi-hypercube. Specifically, we want to maintain that the following two properties hold for each cluster $\sigma$:
\begin{enumerate}
    \item The cluster has no critical child cluster, other than possibly its target child cluster
    \item Suppose the cluster $\sigma$ is the target child cluster of its parent cluster $\sigma'$. Then, the average cluster-$\sigma'$-degree of vertices in $V_\sigma$ is at least $\frac{9}{10}$ times the number of nonempty sibling clusters of $\sigma$. \label{prop:deg-ratio}
\end{enumerate}


Note that while the second property in the definition of a $(k, d, \tau)$-semi-hypercube is stated from the point of view of the parent cluster, for the pruning algorithm, it is more advantageous to look at it as a property of the child as above. To obtain an even easier to work-with form of the property, we now define \textit{marked vertex sets} of clusters, and restate the property as a bound on the fraction of marked vertices in $\sigma$. 

\textbf{Marked Vertices.} The set of \textit{marked vertices} $M_\sigma$ of a cluster $\sigma$ that is the target child cluster of its parent cluster $\sigma'$ is the subset of vertices in the cluster $\sigma$ with cluster-$\sigma'$-degree strictly less than $\frac{19}{20} k'$, where $k'$ is the number of nonempty non-$\sigma$ children of $\sigma'$.

Very importantly, note that each cluster has an individual set of marked vertices specific to that cluster. For any cluster $\sigma$, all the sets $M_{\sigma + i}$ of non-target children $\sigma + i$ are empty, and the marked vertex set $M_\rt$ of the root cluster is always empty. It \emph{does not hold} that $M_\sigma = \bigcup_{i \in [k]} M_{\sigma + i}$.

With this definition of marked vertices of a cluster, \hyperref[prop:deg-ratio]{Property 2} follows from the following, simpler condition:
\begin{enumerate}
    \item[2'.] At most a $\frac{1}{20}$-fraction of vertices in the cluster $\sigma$ are marked in the cluster, i.e. $|V_\sigma \cap M_\sigma| \leq \frac{1}{20} |V_\sigma|$. \label{prop:mark-ratio}
\end{enumerate}

We can observe three constraints on how vertices become marked:
\begin{itemize}
    \item The ratio of marked vertices $|V_\sigma \cap M_\sigma| / |V_\sigma|$ in a cluster $\sigma$ when it initially becomes the target child cluster of its parent cluster $\sigma'$ is very small, specifically at most $O(\tau)$, as the cluster must be noncritical when initially becoming its parent's target child.
    \item The removal of a vertex $v$ from a child cluster of $\sigma'$ other than $\sigma$ results in at most one new vertex in $\sigma$ entering $M_\sigma$, as the edge set between any pair of child clusters of $\sigma'$ forms a matching.
\end{itemize}
Furthermore, already in \Cref{alg:strong-self-pruning-intro}, whenever a vertex is removed from a child cluster other than the target child cluster of $\sigma'$, that removal is followed by $\rho$ calls to $\Call{Trim}{\sigma'}$, each recursively calling $\Call{Trim}{\sigma}$ unless the cluster $\sigma$ becomes empty. Thus, whenever a vertex enters $M_\sigma$ outside the "initialization" of $M_\sigma$ when the cluster $\sigma$ becomes the target child cluster of its parent cluster, $\rho$ trims to the cluster $\sigma$ follow.

Even with the trims following any new marking, as the target child cluster is selected in \Cref{alg:strong-self-pruning-intro} based only on the sizes of the child clusters, the ratio of marks can degrade from the initial $O(\tau)$. To maintain both that no non-target child is critical and the ratio of marks, $\Call{Retarget}{\sigma}$ needs to consider both the size and ratio of marked vertices of a child cluster when deciding the next target. 

\textbf{Alternating Targeting.} Suppose that instead of always selecting the smallest nonempty child cluster as the new target child cluster, $\Call{Retarget}{\sigma}$ selected the smallest nonempty child cluster as the new target child cluster half the time (say, when the number of remaining nonempty clusters is odd), and the child cluster $\sigma + t$ containing the most vertices in $M_\sigma$, i.e. of maximum $|V_{\sigma + t} \cap M_\sigma|$, the other half of the time.

Asymptotically, this would not change the value of $\rho$ required to maintain that non-target child clusters are never critical: \Cref{eq:missing-vertex-bound} would hold when $k'$ is even, and otherwise we'd still have $B_{k', j} \leq x + B_{k' + 1, j}$.

We can bound the maximum number of vertices in the cluster $\sigma$ marked by the parent of $\sigma$ at any moment $\Call{Retarget}{\sigma}$ is called similarly. Denote by $C_{k', j}$ the maximum total number of vertices in $M_\sigma$ among some $j$ specific child clusters of $\sigma$ at the moment $\Call{Retarget}{\sigma}$ is called with $k'$ nonempty child clusters remaining. At the moment the cluster $\sigma$ initially becomes the target child cluster of its parent, causing the initial assignment of marks, the cluster $\sigma$ has lost at most a $\tau$ fraction of its original vertices, and in particular has at least $(1 - \tau)k$ nonempty child clusters. Thus, for $k' + 1 \leq (1 - \tau)k$,
\begin{equation}\label{eq:marked-vertex-bound}
    C_{k', j} \leq x + \left\{ \begin{array}{ll}
        (C_{k' + 1, j + 1}) \cdot \frac{j}{j + 1} & \text{when $k'$ is odd}\\
        C_{k' + 1, j} &\text{when $k'$ is even}
    \end{array}\right.
\end{equation}
Combined with a bound for larger $k'$ simply taking the maximum of $x$ plus the number of initial marks $\bigO(\tau \cdot n_0 k)$, we can bound the total number of marked vertices when $\Call{Retarget}{\sigma}$ is called while there are $k'$ nonempty child clusters, $C_{k', k'}$, by $O(x \log(k)) + O(\tau \cdot n_0 k')$. At that moment, since no child cluster is critical, the total size of the cluster $\sigma$ is at least $(1 - \tau) n_0 k' = \Omega(n_0 k')$, giving a ratio of marks of $O(\log(k) / \rho + \tau)$. With the selection $\rho = \Omega(\log(k) / \tau)$ as before, this ratio of marks is $O(\tau)$.

However, bounding $C_{k', k'}$ does not guarantee \hyperref[prop:mark-ratio]{Property 2'}, as the ratio of marks must be at most $\frac{19}{20}$ after every delete-operation, not just whenever $\Call{Retarget}{\sigma}$ is called. While there are $k' \geq 2$ remaining nonempty clusters, the number of marked vertices is at most $x + C_{k', k'}$ while the number of remaining vertices is at least $(1 - \tau)n_0  \cdot (k' - 1) = \Omega(n_0 k')$, thus the ratio is still $O(\tau)$. However, the real issue is when the cluster $\sigma$ becomes \textit{degenerate}, which is a term we use to refer to nonempty clusters with exactly one child cluster.

\begin{definition}[Degenerate Cluster]
A non-leaf cluster $\sigma$ in a $(k, d, \tau)$-semi-hypercube is \textit{degenerate} if it has exactly one nonempty child cluster, i.e. $V_\sigma = V_{\sigma + t}$ for some child cluster $\sigma + t$ of $\sigma$.
\end{definition}

To control the ratio of vertices in $M_\sigma$ in a degenerate cluster $\sigma$, we have to \emph{pass down the marks} to the child of $\sigma$.

\textbf{Passing Down Marks.} Recall that with marks defined as before, if the parent cluster $\sigma$  of some cluster $\sigma + t$ is degenerate, no vertex in the cluster $\sigma + t$ is marked, i.e. $M_{\sigma + t} = \emptyset$, as the cluster $\sigma$ has no nonempty child clusters other than $\sigma + t$. We will change this: once $\Call{Retarget}{\sigma}$ is called with only one child cluster $\sigma + t$ remaining, we assign $M_{\sigma + t} \gets M_\sigma$, "passing down the marks". Afterwards, whenever a vertex enters $M_\sigma$, we also mark it in the child, adding it to $M_{\sigma + t}$.

Then, the initial number of marked vertices in $\sigma + t$ is still at most $C_{1, 1} = \bigO(\tau n_0)$. However, this bound in the initial fraction of marked vertices is a constant factor higher than the original $\bigO(\tau)$-fraction bound when $\Call{Retarget}{\sigma}$ is called with multiple nonempty child clusters remaining. As the semi-hypercube generally has non-constant depth $d$, the actual maximum ratio of marked vertices with this naive strategy becomes $\tau \cdot 2^{\bigO(d)}$. For the ratio of marks to be bounded by the constant $\frac{1}{20}$, $\tau$ would need to be set to $2^{O(-d)}$, and the trim multiplier $\rho$ selected so far to be $\bigO(\log(k) / \tau)$ would become exponential in $d$, giving a pruning ratio of $\bigO(\log k)^{d^2}$. This pruning ratio is trivial for $d = \Theta(\sqrt{\log n})$, which is unacceptable.

Thus, rather than a constant loss at every depth, we need to aim for a multiplicative loss of $(1 + \bigO(1 / d))$ in the maximum ratio of marked vertices, which multiplies up to only a constant loss over all depths.

This cannot be achieved by alternating targeting between the smallest and most marked child with every retarget: suppose that the initial ratio of vertices in a newly-target cluster $\sigma$ marked by the parent $\sigma'$ is $r > 0$, and that those marked vertices are evenly distributed over $k / 2$ children of $\sigma$, for a ratio of roughly $2r$ in each. Whenever $\Call{Retarget}{\sigma}$ selects the next target child based on size, it is possible for it to select a child with no marked vertices, thus once the cluster $\sigma$ becomes degenerate, it is possible that one of these $k / 2$ children is the only remaining nonempty child, and the initial ratio of marks on that lone child is twice as large as it was for $\sigma$.

However, this example motivates a simple solution: \emph{unbalanced} alternating targeting. Instead of selecting the next target child to be the one with the most marked vertices half the time, we can select it on that basis all but a $O(1 / d)$-fraction of the time, say whenever the number $k'$ of remaining nonempty child clusters is not $1$ modulo $d$.

For child sizes, \Cref{eq:missing-vertex-bound} still holds when $k'$ is $0$ modulo $d$, and otherwise we'd still have $C_{k', j} \leq x + C_{k' + 1, j}$, giving $C_{k', 1} \leq dx \cdot H_k = O(dx \log k)$. Thus, to still maintain that non-target children never become critical, it suffices to increase the trimming multiplier $\rho$ by a multiplicative $d$-factor.

The pruning scheme is stated as an algorithm in \Cref{alg:shc-strong-self-pruning}.

\textbf{Shadow Mark Sets.} Note that while \Cref{alg:shc-strong-self-pruning} is a pruning algorithm with a worst-case pruning ratio guarantee, it is not clear how to implement it with worst-case, not amortized work. Specifically, when a cluster $\sigma$ initially becomes the target child cluster of its parent child cluster $\sigma'$ and the set of marked vertices $M_\sigma$ are initialized inside $\Call{Retarget}{\sigma'}$, the construction of the mark set takes work proportional to the size of the cluster $\sigma$, which can have up to $n / k$ vertices.

To solve this, for the formal pruning result, we consider a different algorithm in \Cref{alg:efficient-self-pruning}, that instead of maintaining just the mark set $M_\sigma$ for every cluster, maintains every possible mark set that could ever be assigned to the cluster.

Recall that $M_\sigma$ is either the set of vertices with cluster-$\sigma'$-degree strictly less than some value (as the degree is always an integer, specifically strictly less than $\lceil \frac{19}{20} k' \rceil$, where $k'$ is the number of sibling clusters the cluster $\sigma$ has), or the marked vertex set $M_{\sigma'}$ of the parent cluster in case the parent cluster is degenerate. Thus, we maintain sets $M^{\mathrm{shadow}}_{\sigma, \sigma', j}$, containing the subset of vertices in the cluster $\sigma$, with cluster-$\sigma'$-degree strictly less than $j$. We additionally store for each cluster its mark source information $\mathrm{minfo}_\sigma$, which just stores the ancestor cluster $\sigma'$ and the $j$ that determine the mark set $M_\sigma$.

In \Cref{sec:pruning-formal-proof}, we give a formal proof of the properties of \Cref{alg:efficient-self-pruning}, specifically showing \Cref{thm:shc-strong-pruning}. 

\begin{restatable*}[Algorithmic Self-Pruning]{theorem}{selfpruning}\label{thm:shc-strong-pruning}
    Let $G = ([k]^d, E)$ be a $(k, d)$-semi-hypercube with $k \geq 16 d$. Then, \Cref{alg:efficient-self-pruning}, initialized with $k, d, \tau, E$ where $\tau \leq \frac{1}{4350d}$, maintains an induced $(k, d, \tau)$-semi-hypercube $G[V] = (V, E[V])$, with initially $V = [k]^d$, in a vertex pruning setting. Specifically, it can process any number of updates $\Call{Delete}{v}$ which do the following:
    \begin{enumerate}
        \item The query $\Call{Delete}{v^{\mathrm{del}}}$ to the data structure gives a vertex $v^{\mathrm{del}} \in V$ to be deleted.
        \item The data structure prunes additional vertices $V^{\mathrm{prune}} \subseteq V - v^{\mathrm{del}}$.
        \item The vertex set is updated to $V \gets V \setminus (V^{\mathrm{prune}} \cup \{v\})$. 
    \end{enumerate}
    The algorithm is deterministic, and has the following guarantees:
    \begin{itemize}
        \item Worst-case pruning ratio: the set of pruned vertices has size $|V^{\mathrm{prune}}| \leq \mathcal{O}\left(\frac{\log(n)}{\tau}\right)^d$.
        \item Worst-case work: $\Call{Delete}{v^{\mathrm{del}}}$ takes work $\poly(k \cdot (\log(n) / \tau)^d)$.
    \end{itemize}
\end{restatable*}

\begin{algorithm}[H]
    \caption{Semi-hypercube self-pruning}
    \label{alg:shc-strong-self-pruning}
    \begin{algorithmic}[1]
        \Class{SelfPruningSHC}
            \Data
                \State Constants $k, d, \rho \in \mathbb{N}$, $\tau \in (0, \frac{1}{4350 d})$, $E \subseteq \binom{[k]^d}{2}$
                \State Vertex set $V \subseteq [k]^d$
                \State Target child clusters $t_{\sigma} \in [k]$, $\sigma \in [k]^{< d}$
                \State Marked vertex subsets $M_{\sigma} \subseteq V_{\sigma}$, $\sigma \in [k]^{\leq d}$
            \EndData

            \Function{initialize}{$E$, $k$, $d$, $\tau$}
                \State $\rho := \left\lceil 3 d H_k / \tau \right\rceil$
                \State $V \gets [k]^d$
                \State $t_{\sigma} \gets k$ for all $\sigma \in [k]^{< d}$
                \State $M_{\sigma} \gets \emptyset$ for all $\sigma \in [k]^{< d}$
            \EndFunction

            \State

            \Function{Mark}{$\sigma$, $M^\mathrm{mark}$}
                \State $M_{\sigma} \gets M_{\sigma} \cup M^\mathrm{mark}$
                \If{the cluster $\sigma$ is degenerate} $\Call{Mark}{\sigma + t_{\sigma}, M^\mathrm{mark}}$ \EndIf
            \EndFunction

            \State
            
            \Function{Retarget}{$\sigma$}
                \State Let $S := \{i : V_{\sigma + i} \neq \emptyset\}$
                \If{$S = \{i\}$}
                    \State $t_{\sigma} \gets i$
                    \State $\Call{Mark}{\sigma + t_{\sigma}, M_{\sigma} \cap V_\sigma}$
                \ElsIf{$S \neq \emptyset$}
                    \If{$|S| = 1 \mod (2d + 1)$} Let $t_{\sigma} \gets \arg\min_{i \in S} |V_{\sigma + i}|$
                    \Else\ $t_{\sigma} \gets \arg\max_{i \in S} |V_{\sigma + i} \cap M_{\sigma}|$ \EndIf
                    \State $\Call{Mark}{\sigma + t_{\sigma}, \{v \in V_{\sigma + t_{\sigma}} : \Call{degree}{v, E_\sigma[V]} < \frac{19}{20} (|S| - 1)\}}$
                \EndIf
            \EndFunction

            \State

            \Function{Trim}{$\sigma$}
                \If{$\sigma$ is a leaf}
                    \State $V \gets V \setminus \{\sigma\}$
                    \State \Return $\sigma$
                \Else
                    \State Let $v_{\mathrm{trim}} := \Call{Trim}{\sigma + t_\sigma}$
                    \If{$V_{\sigma + t_\sigma} = \emptyset$} $\Call{Retarget}{\sigma}$ \EndIf
                    \State \Return $v_{\mathrm{trim}}$
                \EndIf
            \EndFunction

            \State

            \Function{Delete}{$v$}
                \State $V \gets V \setminus \{v\}$
                \State Let $\sigma \gets v$, $V^\mathrm{prune} \gets \emptyset$
                \While{$|\sigma| > 0$}
                    \State $\sigma \gets \Call{Parent}{\sigma}$
                    \If{$V_{\sigma + t_\sigma} = \emptyset$} $\Call{Retarget}{\sigma}$
                    \Else
                        \State Let $k' \gets |\{j \neq t_{\sigma} : V_{\sigma + j} \neq \emptyset\}|$
                        \State $\Call{Mark}{\sigma + t_{\sigma}, \{v \in V_{\sigma + t_{\sigma}} \setminus M_{\sigma + t_{\sigma}} : \Call{degree}{v, E_\sigma[V]} < \frac{19}{20} k'\}}$
                    \EndIf
                    \State Let $\mathrm{trim\_cou} := \Call{min}{\rho \cdot (|V^\mathrm{prune}| + 1), |V_\sigma|}$
                    \For{$\mathrm{trim\_cou}$ times} $V^\mathrm{prune} \gets V^\mathrm{prune} \cup \{\Call{Trim}{\sigma}\}$ \EndFor
                \EndWhile
                \State \Return $V^\mathrm{prune}$
            \EndFunction
        \EndClass
    \end{algorithmic}
\end{algorithm}

\section{Randomized Routing on a Semi-Hypercube}\label{sec:rand-routing}

In this section, we show that semi-hypercubes are excellent routers by giving a easy-to-sample-from oblivious routing with low congestion and path length on semi-hypercubes.

First, in \Cref{sec:noncritical-routing}, we show that noncritical semi-hypercubes can route unit demands with congestion $\bigO(d)$ and path length $\bigO(d)$. Note that these are asymptotically optimal among graphs with maximum degree $n^{\bigO(1 / d)}$.\footnote{Consider any graph with this maximum degree. For any vertex in this graph, the number of vertices at distance $o(d)$ from the vertex must be $n^{o(1)}$, thus for a greedily constructed unit demand $D$ built by repeatedly adding $1$ demand between the two furthest vertices not at maximum load yet, we obtain a unit demand where the total amount of demand between vertices at distance $\Omega(d)$ of each other is at least $|E| - n^{\bigO(1 / d)} \cdot n^{o(1)} = \Omega(|E|)$. Any routing of this demand has to have congestion $\Omega(d)$.}


\oblivroutingnoncritical

Even without the noncriticality condition, an efficiently sampleable oblivious routing with length polynomial in the optimal length over graphs of maximum degree $n^{\bigO(1 / d)}$ exists on a $(k, d, \tau)$-semi-hypercube. 



\begin{restatable*}{theorem}{oblivroutinggeneral}\label{thm:obliv-routing-general}
    Let $G = (V, E)$ be a $(k, d, \tau)$-semi-hypercube for $\tau \leq \frac{1}{27d}$ and $k \geq 27 \cdot 2^{d}$. Then, there exists an oblivious routing $R$ on $G$ of length $\bigO(d^3)$ that can route any unit demand on $V$ with congestion $\bigO\left(kd \cdot 10^d\right)$. A path from this oblivious routing can be sampled in expected running time $\bigO(d^3)$.
\end{restatable*}

Unfortunately, for semi-hypercubes defined as they are, the term $k \cdot 2^{\Omega(d)}$ in the congestion in \Cref{thm:obliv-routing-general} is necessary. First, for a lower bound of $\Omega(kd)$, consider just a $(k, d, \tau)$-semi-hypercube formed by removing from a $(k, d)$-semi-hypercube all child clusters of the root cluster except two. Then, the degree of every vertex is still $\Omega(kd)$, but the edge set between the two child clusters of the root cluster is a matching, thus the average congestion of edges in that matching has to be $\Omega(kd)$. The construction for a lower bound of $k \cdot 2^{\Omega(d)}$ is more complex, and is given in \Cref{sec:obliv-routing-lower-bound}.

\begin{restatable*}{lemma}{oblivroutinglowerbound}\label{lem:obliv-routing-lower-bound}
    For $k \geq 2$ and $\tau \geq \frac{1}{k}$, there is a $(k, d, \tau)$-semi-hypercube $G^{\mathrm{bad}}_d = (V, E)$ on which there exists a unit demand, for which any routing has to have congestion $k \cdot 2^{\Omega(d)}$.
\end{restatable*}

\subsection{Non-Critical Case}\label{sec:noncritical-routing}

We warm up by proving that on a \emph{non-critical} semi-hypercube, there exists an asymptotically optimal oblivious routing that can be efficiently sampled from. 

\oblivroutingnoncritical

The oblivious routing of \Cref{lem:obliv-routing-noncritical} is analogous to the Valiant routing on a hypercube \cite{ValiantRouting}: pick a random midpoint, and greedily route from it to both endpoints. A greedy path on a semi-hypercube is defined like on a hypercube as the path formed by iteratively taking an edge to increase the length of the longest common prefix between the destination and the current path endpoint, which is unique if it exists.

The main difference compared to routing on a hypercube is that not all vertices in the graph have a greedy path to a given vertex. There is still a simple and perhaps obvious strategy: keep sampling midpoints, until we find one that has a greedy path to both of the desired endpoints. We show in \Cref{lem:greedy-path-reach-many} that in a non-critical semi-hypercube with the parameters of \Cref{lem:obliv-routing-noncritical}, for any fixed vertex, at least a $\frac{8}{9}$-fraction of the vertices have a greedy path to that vertex, and thus at least a $\frac{7}{9}$-fraction of vertices have greedy paths to both endpoints.

\begin{definition}[Greedy path]
    Let $G = (V, E)$ be a $(k, d, \tau)$-semi-hypercube. For a source, destination pair $s, t \in V$, the \textit{greedy ($s$, $t$)-path} $\greedy(s, t)$, if it exists, is the path constructed by repeatedly performing the following:
    \begin{itemize}
        \item Take the longest common prefix $\sigma$ of $t$ and the current endpoint $u$ of the path being constructed.
        \item If there is an edge $(u, v)$ from $u$ to a vertex $v$ in the child cluster of $\sigma$ containing $t$, append that edge to the path.
        \item Otherwise, declare the greedy path does not exist.
    \end{itemize}
    For a vertex $t \in V$, we denote the set of vertices $s \in V$ for which there exists a $(s, t)$-greedy path by $\calG_t$. For a cluster $\sigma$, we write $\calG_{t, \sigma} := \calG_t \cap V_{\sigma}$. 
\end{definition}
Since each edge in the greedy path increases the length of the longest common prefix of the current path endpoint and the destination vertex by at least one, any greedy path has length at most $d$. The uniqueness of the path, if it exists, holds as the edge set between the two child clusters of $\sigma$ containing $t$ and the current path endpoint respectively being a matching.

We now show in \Cref{lem:greedy-path-reach-many} the claim that in a non-degenerate $(k, d, \tau)$-semi-hypercube, at least a $\frac{8}{9}$-fraction of vertices have a greedy path to any fixed destination vertex. The full Lemma is more general, and will come useful in the next subsection, where we consider routing in a general $(k, d, \tau)$-semi-hypercube (that can contain critical clusters). It shows that even then, a large fraction of the vertices contained in the \textit{home cluster} of the destination vertex have greedy paths to it. Note that in a non-critical semi-hypercube, the home cluster of every vertex is the root cluster.

\begin{definition}[Home Cluster]
    The \textit{home cluster} $\home(v)$ of a vertex $v \in V$ in a $(k, d, \tau)$-semi-hypercube is the greatest-depth critical ancestor cluster of $v$, or the root cluster if no ancestor of $v$ is critical. 
\end{definition}

\begin{lemma} \label{lem:greedy-path-reach-many}
    Let $G = (V, E)$ be a $(k, d, \tau)$-semi-hypercube with $\frac{1}{\tau}, k \geq 27d$. For a vertex $v$, let $\sigma = \home(v)$. Then, for every non-critical ancestor cluster $\sigma'$ of $v$ that is either a descendant of $\sigma$ or $\sigma$ itself,
    \begin{equation*}
        |\mathcal{G}_{v, \sigma'}| \geq \left(1 - \frac{1}{9d} (d - |\sigma'|)\right) |V_{\sigma'}|,
    \end{equation*}
    and even if the cluster $\sigma$ is critical, $|\mathcal{G}_{v, \sigma}| \geq \frac{5}{9} |V_\sigma|$.
\end{lemma}

\begin{proof}
    We prove the claim through induction. The base case $\sigma' = v$ holds trivially as any vertex has a greedy path to itself.
    
    Now, fix an ancestor cluster $\sigma'$ of $v$ that is either a descendant of $\sigma$ or is $\sigma$ itself, and suppose the claim holds for its child cluster $V_{\sigma' + i'}$ containing $v$. Let $V_{\sigma' + i}$ be an arbitrary noncritical child cluster of $V_{\sigma'}$ other than the child containing $v$. Consider some vertex $s$ in that child cluster. A greedy path from $s$ to $v$ exists if and only if
    \begin{enumerate}
        \item There is an edge from $s$ to some vertex $s'$ in the child cluster $V_{\sigma' + i'}$ containing $v$, and
        \item There is a greedy path from that $s'$ to $v$, i.e. $s' \in \calG_{v}$.
    \end{enumerate}
    Since the edge set between the two child clusters forms a matching, the vertex $s'$ if it exists is unique and not shared with any other vertex in $V_{\sigma' + i}$. Thus, the number of vertices in $V_{\sigma' + i}$ with a greedy path to $v$ is at least the number of edges between the two clusters minus the number of vertices in $V_{\sigma' + i'}$ that \emph{do not} have a greedy path to $v$:
    \begin{align*}
        |\mathcal{G}_{v, \sigma' + i}| &\geq |E[V_{\sigma' + i'}, V_{\sigma' + i}]| - \left(|V_{\sigma' + i'}| - |\mathcal{G}_{v, \sigma' + i'}|\right)\\
                                                        &\geq \left((1 - 2\tau) k^{d - |\sigma'| - 1}\right) - \left(\frac{1}{9d} (d - |\sigma'| - 1) k^{d - |\sigma'| - 1}\right)\\
                                                        &\geq \left(1 - \frac{1}{9d} (d - |\sigma'|) + \frac{1}{27d}\right) k^{d - |\sigma'| - 1}\\
                                                        &\geq \left(1 - \frac{1}{9d} (d - |\sigma'|) + \frac{1}{27d}\right) |V_{\sigma' + i}|.
    \end{align*}
    If $\sigma'$ has no critical child cluster, we are done. Otherwise, let $\sigma' + t_\sigma$ be that critical child cluster. Then,
    \begin{equation*}
        |\mathcal{G}_{v, \sigma'}| \geq \left(1 - \frac{1}{9d} (d - |\sigma'|) + \frac{1}{27d}\right) |V_{\sigma'}| - |V_{\sigma' + t_\sigma}| = \left(1 - \frac{1}{9d} (d - |\sigma'|)\right) |V_{\sigma'}| - \left(|V_{\sigma' + t_\sigma}| - \frac{1}{27d} |V_{\sigma'}|\right).
    \end{equation*}
    If the cluster $\sigma'$ in consideration is noncritical, we are then done, as a critical child cluster of a noncritical cluster makes up at most a $\frac{1}{k} \leq \frac{1}{27d}$ of the cluster's vertices: $\sigma'$ has size at least $(1 - \tau) k^{d - |\sigma'|}$ as a noncritical cluster, and $\sigma' + t_\sigma$ has size at most $(1 - \tau) k^{d - |\sigma' + t_\sigma|} = \frac{(1 - \tau)}{k} \cdot k^{d - |\sigma'|}$ as a critical cluster.
    
    It remains to consider the case where $\sigma'$ is critical, which can only occur when $\sigma' = \sigma$ by definition. We branch on whether the cluster $\sigma$ has two or more child clusters. If it has two, the number of edges between the child $V_{\sigma + i'}$ containing $v$ and the critical child $V_{\sigma + t_\sigma}$ (which is a distinct child by definition of $\sigma = \sigma'$) must be at least $\frac{9}{10} \cdot |V_{\sigma + t_\sigma}|$ by the definition of a $(k, d, \tau)$-semi-hypercube, thus
    \begin{equation*}
        |\mathcal{G}_{v, \sigma + t_\sigma}| \geq \frac{9}{10} |V_{\sigma + t_\sigma}| - \left(|V_{\sigma + i'}| - |\mathcal{G}_{v, \sigma + i'}|\right) \geq \frac{9}{10} |V_{\sigma + t_\sigma}| - \frac{1}{9} |V_{\sigma + i'}|,
    \end{equation*}
    and, as desired, $ |\mathcal{G}_{v, \sigma}|   \geq \frac{7}{9} |V_{\sigma + i'}| + \frac{9}{10} |V_{\sigma + t_\sigma}| \geq \frac{7}{9} |V_{\sigma}|$. Finally, suppose $\sigma$ has at least three nonempty child clusters. Then, $|V_{\sigma + t_\sigma}| \leq \frac{1}{3} |V_{\sigma}|$ as the other child clusters of $\sigma'$ are not critical but $\sigma + t_\sigma$ is, and 
    \begin{equation*}
        |\mathcal{G}_{v, \sigma}|   \geq \left(1 - \frac{1}{9}\right) \left(1 - \frac{1}{3}\right) |V_{\sigma}|
                                    = \frac{16}{27} |V_{\sigma}| \geq \frac{5}{9} |V_{\sigma}|.
    \end{equation*}
\end{proof}

Now that we have that at least a $\frac{7}{9}$-fraction of vertices are valid midpoints, it remains to show that a path set that for every vertex $t$ contains a uniformly random greedy path to $t$ has low expected congestion. The congestion of the final oblivious routing on any unit demand is then at most $\frac{9}{7} \cdot kd$ times this congestion, as every vertex in a $(k, d, \tau)$-semi-hypercube has degree at most $kd$.

The following Lemma shows this as a special case, with the full Lemma again being useful in the next subsection. The proof of the Lemma is like the analysis of the congestion of Valiant routing on a hypercube.

\begin{lemma}\label{lem:home-routing}
    Let $G = (V, E)$ be a $(k, d, \tau)$-hypercube for $\frac{1}{\tau}, k \geq 27d$, $\sigma$ a cluster and $H_\sigma := \{v \in V_{\sigma} : \home(v) = \sigma\}$ the set of its residents. Consider a path set that contains for each vertex $t \in H_\sigma$ a greedy path $\greedy(s, t)$ from an uniformly random source $s \in \calG_{t, \sigma}$ in the cluster. Then, for any edge $e$,
    \begin{itemize}
        \item If the edge is not contained in the cluster $\sigma$, or if the edge is contained in a critical descendant cluster of $\sigma$, no sampled path contains the edge.
        \item Otherwise, the expected number of paths containing the edge is
        \begin{itemize}
            \item $\bigO(\frac{1}{k})$ if the cluster $\sigma$ is not critical, and
            \item $\bigO(1)$ otherwise.
        \end{itemize}
    \end{itemize}
\end{lemma}

\begin{proof}
    A greedy path from a vertex $s \in V_\sigma$ to a vertex $t \in V_{\sigma}$ is contained in the cluster $V_\sigma$, thus cannot contain edges not contained in the cluster $\sigma$. If the edge $e$ is contained in some critical descendant cluster $\sigma'$ of $\sigma$, then note that $t$ cannot be contained in the cluster $\sigma'$ as otherwise the cluster $\sigma$ could not be $t$'s home cluster, and thus the lowest common ancestor cluster of $t$ and either of the edge $e$'s endpoints cannot be the cluster $\sigma'$, thus the edge cannot be on a greedy path to a vertex in $H_\sigma$.
    
    Now, fix an edge $e$ between two child clusters $V_{\sigma', i}$ and $V_{\sigma', i'}$ of some cluster $\sigma'$ that is either a descendant cluster of $\sigma$ or the cluster $\sigma$ itself. Then,
    \begin{itemize}
        \item The edge $e$ can only be contained in paths $\greedy(s, t)$ for $t$ in one of the child clusters $V_{\sigma', i}$ or $V_{\sigma', i'}$, which combined have size at most $2k^{d - |\sigma'| - 1}$.
        \item For any such $t \in V_{\sigma', i} \cup V_{\sigma', i'}$, the number of $s \in \calG_{t, \sigma}$ for which $e \in \greedy(s, t)$ is at most $k^{|\sigma'| - |\sigma|}$. By \Cref{lem:greedy-path-reach-many}, $|\calG_{t, \sigma}| \geq \frac{5}{9} |V_{\sigma}|$ for any $t \in H_\sigma$, thus the probability that some $s$ for which $e \in \greedy(s, t)$ is sampled is at most $\frac{5}{9} \cdot k^{|\sigma'| - |\sigma|}\ /\ |V_\sigma|$.
    \end{itemize}
    If $\sigma$ is not critical, then $|V_{\sigma}| \geq (1 - \tau) k^{d - |\sigma|} = \Theta(k^{d - |\sigma|})$, thus the expected number of paths using the edge is at most
    \begin{equation*}
        \frac{10}{9} \cdot k^{d - |\sigma'| - 1} \cdot k^{|\sigma'| - |\sigma|}\ /\ |V_\sigma| = \Theta(k^{d - |\sigma'| - 1}) \cdot \Theta(k^{|\sigma'| - |\sigma|})\ /\ \Theta(k^{d - |\sigma|}) = \Theta(1 / k).   
    \end{equation*}
    Otherwise, $|V_{\sigma}| \geq (1 - \tau) k^{d - |\sigma| - 1} = \frac{1}{k} \Theta\left(k^{d - |\sigma|}\right)$, and the expectation is $\bigO(1)$.
\end{proof}

\Cref{lem:obliv-routing-noncritical} now follows as a corollary of \Cref{lem:greedy-path-reach-many} and \Cref{lem:home-routing}:

\oblivroutingnoncritical*

\begin{proof}
    We define the oblivious routing $R_{s, t}$ as follows: for $s, t \in H_\sigma$ and $v \in \calG_{s, \sigma} \cap \calG_{t, \sigma}$, let $P_{s, v, t}$ be the path formed by concatenating $\reverse(\greedy(v, s))$ with $\greedy(v, t)$. The oblivious routing $R_{s, t}$ is a uniform distribution over such paths: $R_{s, t}(P_{s, v, t}) := 1 / |\calG_{s, \sigma} \cap \calG_{t, \sigma}|$ for all $s, t \in H_0$.
    
    Note that by \Cref{lem:greedy-path-reach-many}, $|\calG_{s, \sigma} \cap \calG_{t, \sigma}| \geq \frac{7}{9} |V_{\sigma}|$. Thus, the congestion this oblivious routing places on any edge is at most the congestion in a path set that contains each greedy path $\greedy(s, t)$ with weight
    \begin{equation*}
        \frac{\left(\sum_{s'} D(s', t) + D(t, s')\right)}{\frac{7}{9} |V|} \leq \frac{9kd}{7} \cdot \frac{1}{|V|} \leq \frac{9kd}{7} \cdot \frac{1}{|\calG_t|}.
    \end{equation*}
    By \Cref{lem:home-routing}, the congestion of this path set is $\bigO(d)$, as since the semi-hypercube is noncritical, the home cluster of every vertex is the root cluster, and the root cluster is not critical.

    Finally, it remains to analyze the length and sampling complexity. Since any greedy path has length at most $d$, the lengths of the paths supported by the oblivious routing are at most $2d$. To sample a path in expected time $\bigO(d)$, repeatedly sample a random vertex $v$, check if $v \in \calG_{s, \sigma} \cap \calG_{t, \sigma}$ by attempting to construct the greedy paths in $\bigO(d)$ time each, and resample if unsuccessful. The expected number of samples required is at most $\frac{9}{7} = \bigO(1)$. For constructing a greedy path in time $\bigO(d)$, it suffices to store the for each vertex an array of size $d \cdot k$ where the edge to sibling $i$ at depth $j$ appears at index $(i - 1) + k \cdot j$ and can thus be looked up in time $\bigO(1)$.
\end{proof}

\subsection{General Case}

We now proceed to proving the main result of the section.

\oblivroutinggeneral

We assume without loss of generality that the root cluster is not degenerate, as if it is, we may consider routing on its only child cluster, which up to vertex labeling is a $(k, d - 1, \tau)$-semi-hypercube (or in general if that cluster too is degenerate, on its \textit{representative cluster}, defined as the first nondegenerate descendant).

\begin{definition}[Representative Cluster]
In a $(k, d, \tau)$-semi-hypercube, the representative cluster $\rep(\sigma)$ of a degenerate cluster $\sigma$ is the least-depth nondegenerate descendant of the cluster, and the representative of a non-degenerate cluster $\sigma$ is the cluster $\sigma$ itself. In other words, the representative cluster $\rep(\sigma)$ can be defined recursively as follows:
\begin{itemize}
    \item If $\sigma$ is not degenerate, $\rep(\sigma) = \sigma$.
    \item Otherwise, let $\sigma + t$ be the only nonempty child of $\sigma$. Then, $\rep(\sigma) = \rep(\sigma + t)$.
\end{itemize}
\end{definition}

From the previous subsection, we know that oblivious routing between two vertices sharing the same home cluster can be done with low congestion over short paths. For routing between the entire vertex set, we first want to route from both endpoints to vertices for which the root cluster is the home cluster. As a measure of progress towards routing to this easily-manageable part of the graph, we define \textit{isolation}.

\begin{definition}[Isolation]
    The \textit{isolation} $\iso(v)$ of a vertex $v \in V$ in a $(k, d, \tau)$-semi-hypercube is
    \begin{equation*}
        \iso(v) := \sum_{\sigma \in \ancs(v)} (d - |\sigma|) \cdot I[\text{the cluster $\sigma$ is critical}].
    \end{equation*}
    For a cluster $\sigma$, the isolation of the cluster is the minimum isolation of a vertex in the cluster, i.e.
    \begin{equation*}
        \iso(\sigma) := \min_{v \in V_\sigma} \iso(v).
    \end{equation*}
\end{definition}

Recall that the home cluster $\sigma$ of a vertex $s$ is the critical ancestor cluster of $s$ of greatest depth (i.e. closest to $s$ in the cluster tree), thus $\iso(s) = \iso(\sigma)$. Additionally, note that the root cluster is always the minimum-isolation cluster in the graph. Combining these two, one can observe that repeatedly routing to a vertex of strictly lower isolation results in a path to a vertex with the the root cluster as the home cluster. As the isolation of a vertex can never exceed $\sum_{j = 0}^{d - 1} d - j$, doing this $\bigO(d^2)$ times suffices.

Since $\iso(s) = \iso(\sigma)$ for the home cluster $\sigma$ of $s$, any path to a vertex of lower isolation from $s$ must leave its home cluster. Thus, our strategy for routing from $s$ to a vertex of lower isolation is rather simple: sample a vertex $u$ in the home cluster of $s$ and a nonempty sibling of the home cluster of $s$, and check
\begin{itemize}
    \item if there is a greedy path from $u$ to $s$ (i.e. $u \in \calG_s$ holds),
    \item if there is an edge $(u, v)$ from $u$ to the sampled sibling cluster, and
    \item if the endpoint $v$ of the edge has strictly lower isolation than the current path endpoint $s$.
\end{itemize}
If the answer to each check is yes, we take the reverse of the greedy path from $u$ to $s$ and then the edge $(u, v)$, ending at the vertex $v$ of strictly lower isolation than $s$. Otherwise, we resample until a success. This sampling ensures that the first found valid edge $(u, v)$ is a uniformly random valid edge.

\begin{algorithm}[h]
    \caption{Oblivious Path Sampling} \label{alg:obliv-routing-general}
    \begin{algorithmic}[1]
        \Function{Greedy}{$s, t$}
            \State Let $P$ be an empty path
            \While{$s \neq t$}
                \State Let $\sigma = \mathrm{lcp}(s, t)$
                \State Let $i$ be such that $t \in V_{\sigma + i}$
                \If{there is an edge $(s, s')$ from $s$ to a vertex $s' \in V_{\sigma + i}$}
                    \State $P \gets \Call{concat}{P, (s, s')}$
                    \State $s \gets s'$
                \Else\ \Return $\bot$ \EndIf
            \EndWhile
            \State \Return $P$
        \EndFunction
        
        \State
        
        \Function{Escape}{$s$}
            \If{$\home(s) = \rep(\rt)$}
                \State \Return $(s)$
            \EndIf
            \State Let $\sigma = \home(s)$ and $\sigma'$ be the nondegenerate ancestor of $\sigma$ of greatest depth
            \State Let $S = \{i \in [k] : \text{the child cluster $\sigma' + i$ of $\sigma'$ is nonempty and noncritical}\}$
            \While{True}
                \State Sample a uniformly random pair $(u, i) \in V_{\sigma} \times S$
                \State Let $v$ be the unique vertex in $V_{\sigma' + i}$ such that $(u, v) \in E$ (or $\bot$ if none exists)
                \State Let $P = \Call{reverse}{\Call{Greedy}{u, s}}$ (or $\bot$ if no greedy path exists)
                \If{$v, P \neq \bot$ and $\iso(v) < \iso(\sigma)$}  \label{line:edge-condition-line}
                    \State \Return $\Call{concat}{P, (u, v), \Call{Escape}{v}}$
                \EndIf
            \EndWhile
        \EndFunction

        \State

        \Function{SamplePath}{$s$, $t$}
            \State Let $P_0 = \Call{Escape}{s}$
            \State Let $P_3 = \Call{reverse}{\Call{Escape}{t}}$
            \State Let $s'$ be the endpoint of $P_0$ and $t'$ the startpoint of $P_3$
            \While{True}
                \State Let $v$ be a uniformly random vertex in $V$
                \State Let $P_1 = \Call{reverse}{\Call{Greedy}{v, s'}}$ (or $\bot$ if none exists)
                \State Let $P_2 = \Call{Greedy}{u, t'}$ (or $\bot$ if none exists)
                \If{$P_1 \neq \bot$ and $P_2 \neq \bot$}
                    \State \Return $\Call{concat}{P_0, P_1, P_2, P_3}$
                \EndIf
            \EndWhile
        \EndFunction
    \end{algorithmic}
\end{algorithm}

The following Lemma guarantees that the probability the answer to each check is yes is at least $\frac{1}{9}$. In the analysis, in addition to isolation, we care about the \textit{count} of critical ancestors. We denote by $\cou(v)$ the critical ancestor count of a vertex $v$, i.e. we let
\begin{equation*}
    \cou(v) := |\{\sigma \in \ancs(v) : \text{the cluster $\sigma$ is critical}\}|.
\end{equation*}

\begin{lemma}\label{lem:out-edge-count}
    Suppose $k \geq 3 \cdot 2^{d} / (1 - \tau)$. Let $t$ be a vertex with a non-root home cluster, $\sigma = \home(t)$ its home cluster, $\sigma'$ the highest-depth non-degenerate ancestor of $\sigma$, and $k'$ the number of noncritical nonempty child clusters of $\sigma'$. Then, the number of edges $(u, v) \in E$ such that
    \begin{itemize}
        \item $u \in \calG_{t, \sigma}$ and $v \in V_{\sigma'} \setminus V_{\sigma}$, and
        \item $\iso(v) < \iso(\sigma)$,
    \end{itemize}
    is at least $\frac{1}{9} \cdot k' \cdot |V_{\sigma}|$.
\end{lemma}

\begin{proof}
    By \Cref{lem:greedy-path-reach-many}, $|\mathcal{G}_{s, \sigma}| \geq \frac{5}{9} |V_{\sigma}|$. Since $\sigma$ is critical, $|E[V_\sigma, V_{\sigma'} \setminus V_\sigma]| \geq \frac{9}{10} \cdot k' \cdot |V_{\sigma}|$ by the definition of a $(k, d, \tau)$-semi-hypercube. Finally, as every vertex $v$ in a sibling cluster of $\sigma$ is the endpoint of at most one edge from the cluster $\sigma$ to its siblings, and each vertex in the cluster $\sigma$ is the endpoint of at most $k'$ of these edges, the number of edges satisfying both requirements is at least
    \begin{align*}
        &\left(\frac{9}{10} \cdot k' \cdot |V_{\sigma}|\right) - \left(\left(1 - \frac{5}{9}\right) \cdot k' \cdot |V_{\sigma}|\right) - \left|\left\{v \in V_{\sigma'} \setminus V_\sigma : \iso(v) \geq \iso(s)\right\}\right|\\
        \geq&\ \frac{4}{9} \cdot k' \cdot |V_{\sigma}| - \left|\left\{v \in V_{\sigma'} \setminus V_\sigma : \iso(v) \geq \iso(s)\right\}\right|. 
    \end{align*}
    We now bound the latter term, i.e. the number of vertices $v$ with isolation at least $\iso(s)$.
    
    First, we show that $\iso(v) \geq \iso(\sigma)$ implies $\cou(v) > \cou(s)$ for $v \in V_{\sigma'} \setminus V_{\sigma}$. 
    Note that sorting the non-shared ancestor clusters of $v$ and $s$ in increasing order of depth, a prefix of the ancestors of $s$ are critical, the rest noncritical, while the least-deep non-shared ancestor of $v$ cannot be critical (as $\sigma'$ can have at most one critical child cluster). As the isolation of a vertex is the sum of $d - |\sigma|$ over its critical ancestor clusters $\sigma$, the vertex $s$ must have strictly larger isolation than $v$ if $\cou(v) \leq \cou(s)$. As $\iso(s) = \iso(\sigma)$, the claim holds.

    Next, for some nonempty child cluster $\sigma''$ of $\sigma'$ other than the child cluster containing $s$, we want to upper bound the number of vertices $v \in V_{\sigma''}$ for which $\cou(v) > \cou(s)$.
    
    Note that every vertex $v$ satisfying $\cou(v) > \cou(s)$ in the cluster $\sigma''$ corresponds to a walk down the cluster tree from $\sigma''$, where we elect to go to a critical child at least $|\sigma| - |\sigma'| + 1$ times. Define $p := |\sigma| - |\sigma'|$, as we will use this quantity a lot. There are $\binom{d - |\sigma''|}{p + 1}$ ways to select the first $p + 1$ depths where the critical child is selected, and for any cluster, there is at most one critical child and at most $k$ noncritical children. Thus, the number of vertices in $V_{\sigma''}$ for which $\cou(v) > \cou(s)$ is at most
    \begin{equation*}
        \binom{d - |\sigma''|}{p + 1} \cdot k^{d - |\sigma''| - (p + 1)} \leq 2^{d} \cdot k^{d - |\sigma''| - (p + 1)} = \frac{2^{d}}{k} \cdot k^{d - |\sigma'| - (p + 1)},
    \end{equation*}
    and counting each child cluster $\sigma''$, the total number of such vertices is at most $k'$ times that many.
    
    Now, we need to lower bound $|V_\sigma|$. Since the cluster is the home cluster of a vertex, it has at least one noncritical child cluster, and thus has size at least $|V_\sigma| \geq (1 - \tau) k^{d - |\sigma| - 1} = (1 - \tau) k^{d - |\sigma'| - (p + 1)}$. Thus,
    \begin{align*}
        \left|\left\{v \in V_{\sigma'} \setminus V_\sigma : \iso(v) \geq \iso(s)\right\}\right| &\leq \frac{2^{d}}{k} \cdot k' \cdot k^{d - |\sigma'| - (p + 1)}\\
            &\leq \frac{2^{d}}{(1 - \tau) k} \cdot k' \cdot |V_\sigma|\\
            &\leq \frac{1}{3} \cdot k' \cdot  |V_\sigma|.
    \end{align*}
\end{proof}

We are ready to prove \Cref{thm:obliv-routing-general}.

\oblivroutinggeneral*

\begin{proof} 

The oblivious routing is the one defined by the distribution of paths returned by $\Call{SamplePath}{s, t}$ in \Cref{alg:obliv-routing-general}.

We first analyze the function escape. Let $\sigma$ be some non-root cluster, and consider a call to $\Call{Escape}{t}$ on some vertex $t$ for which $\home(t) = \rep(\sigma)$. We first want to bound the probability that a specific fixed edge $(u, v)$ is the first edge for which the conditions of \Cref{line:edge-condition-line} are satisfied. The conditions can only be satisfied when
\begin{itemize}
    \item $u \in \calG_{s}$ (as the greedy path from $u$ to $t$ has to exist)
    \item $u \in V_\sigma$ (as this is the set $u$ is sampled from in the function)
    \item $v \in V_{\sigma'} \setminus V_\sigma$ (as the cluster $\rep(\sigma)$ is a non-root home cluster, and thus $\sigma$ is critical and does not occur in the set $S$, which only contains nonempty, noncritical child clusters of $\sigma'$)
    \item $\iso(v) < \iso(\sigma)$
\end{itemize}
Note that these are exactly the conditions on the edges counted by \Cref{lem:out-edge-count}, thus the number of eligible edges is at least $\frac{1}{9} \cdot k' \cdot |V_{\sigma}|$, where $k'$ is the number of noncritical nonempty child clusters of $\sigma'$. The first edge for which the conditions of \Cref{line:edge-condition-line} are satisfied is uniformly random over eligible edges, as the selected edge is determined uniquely by a pair $(u, i)$, and both $u$ and $i$ are sampled uniformly at random. Thus, the probability that a specific eligible edge gets selected is at most $\frac{9}{k' \cdot |V_\sigma|} \leq \frac{9}{|V_\sigma|}$, and as any vertex $v$ is the endpoint of at most one eligible edge, as otherwise it would be incident to two vertices in $V_\sigma$, the probability a specific vertex $v$ is the endpoint of that edge is similarly at most $\frac{9}{|V_\sigma|}$.

We can now already argue the claimed time complexity and path lengths. We have shown that a valid $(u, v)$ exists for a sampled pair $(u, i)$ with probability at least $\frac{1}{9}$, thus the expected number of samples required to find a valid edge is $\bigO(1)$. Checking a sample only requires constructing one greedy path, which can be done in time $\bigO(d)$. The actual sampling of a vertex from a cluster can be done in time $\bigO(d)$ by storing the vertex set as a trie and storing on every vertex the size of its subtree. As the isolation decreases with each recursive call and is initially at most $\bigO(d^2)$, the total expected running time of escape and the length of the path returned are $\bigO(d^3)$ as desired. Finally, by \Cref{lem:obliv-routing-noncritical}, the oblivious routing between $s'$ and $t'$ takes $\bigO(d)$ time and returns a path of length $\bigO(d)$.

Now, we want to bound the contribution of the escape-parts of the oblivious routing to the congestion of routing a unit demand. For a cluster $\sigma$ and vertex $s \in V \setminus V_\sigma$, denote by $W_s(\sigma)$ the probability that a call to $\Call{Escape}{s}$ makes at some point a recursive call to $\Call{Escape}{v}$ for some $v \in V_\sigma$. For $s \in V_\sigma$, define $W_s(\sigma) = 1$. Then, for any vertex $s$ and cluster $\sigma''$, we claim the following
\begin{equation*}
    W_s(\sigma'') \leq \mathbb{I}[s \in V_{\sigma''}] + \sum_{\sigma' \in \ancs(\sigma'')} \mathbb{I}[\sigma' \text{ has a critical child } \sigma] \cdot \mathbb{I}[\iso(\sigma) > \iso(\sigma'')] \cdot \left(\frac{9}{|V_{\sigma}|} \cdot |V_{\sigma''}| \right) \cdot W_s(\sigma)
\end{equation*}
The claim clearly holds when $s \in V_{\sigma''}$ due to the first term. Otherwise, consider the first edge $(u, v)$ that results in a recursive call to $\Call{Escape}{v}$ on some vertex $v \in V_{\sigma''}$. Let $s'$ be the vertex the prior call to escape was made on, $\sigma'$ be the lowest common ancestor of $u$ and $v$, and $\sigma$ the child cluster of $\sigma'$ containing $u$. Then
\begin{itemize}
    \item $\sigma'$ must be an ancestor of $\sigma''$, as the vertex $v$ must be in the cluster $\sigma''$, but the vertex $u$ cannot be, as $u$ and $s'$ are contained in the same child cluster of $\sigma'$, and $s' \not\in V_{\sigma''}$,
    \item $\sigma$ must be the critical child of $\sigma'$, as it is not the root cluster and $\rep(\sigma) = \home(s')$,
    \item $\iso(\sigma) > \iso(\sigma'')$ must hold, as $\iso(v) \geq \iso(\sigma'')$ and $\iso(\sigma) > \iso(v)$ since the edge is valid,
    \item The valid edge is selected with probability at most $\frac{9}{|V_\sigma|}$, and there are $|V_{\sigma''}|$ options for the specific vertex $v$ in the cluster $\sigma''$, and
    \item The previous call to escape was made on $s' \in \sigma$,
\end{itemize}
which combined give the upper bound. Now, define $W(\sigma) := \sum_{s \in V} W_s(\sigma)$. Then, by summing over $s$ in the inequality (and using $\mathbb{I}[\iso(\sigma) > \iso(\sigma'')] \leq 1$), we get
\begin{align*}
    W(\sigma'')   &\leq |V_{\sigma''}| + \sum_{\sigma' \in \ancs(\sigma'')} \mathbb{I}[\sigma' \text{ has a critical child } \sigma] \cdot \mathbb{I}[\iso(\sigma) > \iso(\sigma'')] \cdot \left(\frac{9}{|V_{\sigma}|} \cdot |V_{\sigma''}| \right) \cdot W(\sigma)\\
                    &\leq |V_{\sigma''}| \left(1 + 9 \sum_{\sigma' \in \ancs(\sigma'')} \mathbb{I}[\sigma' \text{ has a critical child } \sigma] \cdot \frac{W(\sigma)}{|V_{\sigma}|}\right).
\end{align*}
Suppose that the cluster $\sigma''$ is either critical, or a vertex. Then, since any cluster $\sigma'$ can have at most one critical child, for any ancestor $\sigma'$ of $\sigma''$ and its critical child $\sigma$, we have $|\sigma| > |\sigma''|$. We can thus show by induction on $|\sigma''|$ that $W(\sigma'') \leq |V_{\sigma''}| \cdot 10^{|\sigma''|}$:
\begin{align*}
    W(\sigma'') &\leq |V_{\sigma''}| \left(1 + 9 \sum_{\sigma' \in \ancs(\sigma'')} \mathbb{I}[\sigma' \text{ has a critical child } \sigma] \cdot \frac{W(\sigma)}{|V_{\sigma}|}\right)\\
                &\leq |V_{\sigma''}| \left(1 + 9 \sum_{\sigma' \in \ancs(\sigma'')} \mathbb{I}[\sigma' \text{ has a critical child } \sigma] \cdot 10^{|\sigma|}\right) \leq |V_{\sigma''}| \cdot 10^{|\sigma''|}.
\end{align*}
Thus, in particular, we have obtained that when calling $\Call{Escape}{s}$ exactly once on each vertex $s$, the expected total number of calls to $\Call{Escape}{v}$ that get made for some fixed vertex $v$ is at most $W(v) \leq 10^d$.

It only remains to relate this to the congestion of obliviously routing a unit demand. We first bound the congestion caused by the greedy subpaths of the routing paths, then bound the congestion caused by the remaining edges $(u, v)$ concatenated to the paths in escape.

Consider the greedy subpaths of the half of a routing path for some pair $(s, t)$ until the midpoint. There is exactly one greedy subpath for every call to escape made, each ending on the corresponding vertex escape was called on: one for every call but the last recursive call in escape in the path returned by escape, and then one on the endpoint $s'$ of the path escape returned, for which $\home(s') = \rep(\rt)$.

Conditioning on one of these calls to escape being on some vertex $q$, the greedy subpath ending at $q$ with a home cluster $\sigma = \home(q)$ is from a vertex $u \in \calG_{q, \sigma}$, and the probability that the path comes from some particular $u \in V_{\sigma}$ is at most $\frac{1}{\frac{1}{9} |V|} = \frac{9}{|V_\sigma|}$ by \Cref{lem:greedy-path-reach-many} if $q$ is the first vertex for which $\sigma = \home(q) = \rep(\rt)$, i.e. $q = s'$, and otherwise at most $9 / |V_\sigma|$ by \Cref{lem:out-edge-count} as argued earlier.

Thus, the congestion the greedy subpaths in an oblivious routing of a unit demand place on any edge $e$ is at most the congestion of the edge in a path set that contains each greedy path $\greedy(s, t)$ for $t \in V$ and $s \in \calG_{t, \sigma}$ where $\sigma = \home(t)$ with weight
\begin{equation*}
    kd \cdot W(t) \cdot \frac{9}{|V_\sigma|} \leq 9kd \cdot 10^d \cdot \frac{1}{|\calG_t|}.
\end{equation*}
By \Cref{lem:home-routing}, the congestion of the edge in this path set is $\bigO(kd \cdot 10^d)$, as for the edge there is at most one cluster $\sigma$ that can be a home cluster (which can only occur if it is the root cluster or is critical), which contains the edge, and which has no critical descendant cluster containing the edge.

Finally, we analyze the congestion caused by the remaining edges on the routing paths. Fix again some edge $e = (u, v)$. Let $\sigma'$ be the lowest ancestor cluster of $u$ and $v$, i.e. the unique cluster for which the edge connects two vertices in two distinct child clusters $V_{\sigma' + i}$ and $V_{\sigma' + i'}$, assume WLOG that the latter child cluster is noncritical, and let $\sigma = \rep(\sigma' + i)$. Then, the edge can only appear as the selected edge in calls to escape on a vertex $t$ for which $\home(t) = \sigma$, and the probability the edge is selected by a call $\Call{Escape}{t}$ is at most $9 / |V_\sigma|$. Thus, the congestion caused to the edge by the oblivious routing is at most
\begin{equation*}
    kd \cdot \left(\sum_{t \in H_\sigma} W(t)\right) \cdot \frac{9}{|V_\sigma|} \leq 9kd \cdot 10^d
\end{equation*}
and we are done.

\end{proof}

\subsection{Lower Bound on Semi-Hypercube Routing Congestion} \label{sec:obliv-routing-lower-bound}

\oblivroutinglowerbound

\begin{proof}
For some starting depth $d_0 \geq 2$, we define $(k, d, \tau)$-semi-hypercubes $G^{\mathrm{bad}}_d$ for $d$ of the same parity as $d_0$ and $(k, d, \tau)$-semi-hypercubes $G^{\mathrm{spread}}_d$ for the other parity. Each graph has an associated hard subset $\mathrm{hard}(G^{\mathrm{bad}}_d)$ or $\mathrm{hard}(G^{\mathrm{spread}}_d)$ with a small size cut to the rest of the graph. We define $s_d$ as the size of the hard subset $\mathrm{hard}(G^{\mathrm{bad}}_d)$ for $d$ of the same parity as $d_0$ and $\mathrm{hard}(G^{\mathrm{spread}}_d)$ for the other parity, and by $c_d$ the size of the cut between the hard subset and the rest of the graph in $G^{\mathrm{bad}}_d$ for $d$ of the same parity as $d_0$ and in $G^{\mathrm{spread}}_d$ for the other parity.

For some starting depth $d_0$ selected later, the root cluster of the graph $G^{\mathrm{bad}}_{d_0}$ has two child clusters, the first of which and the hard vertex subset of the graph is a $k$-clique, and the second of which is a $(k, d_0 - 1)$-semi-hypercube. The matching between these child clusters is simply an arbitrary matching of size $k$. Thus, $s_1 = k$ and $c_1 = k$.

For $d > d_0$ of the same parity as $d_0$, the root cluster of the graph $G^{\mathrm{bad}}_d$ consists of two children, one of which is an arbitrary $(k, d - 1, \tau)$-semi-hypercube of size $s_{d - 1}$ and average degree at least $k - 1$, and the other of which is a copy of $G^{\mathrm{spread}}_{d - 1}$. The edge set between those two children simply matches vertices in the first child with the vertices in the hard subset of the second child, which by definition are equal size. The hard subset is the union of the whole first child with the hard subset of the second child, thus $s_d = 2s_{d - 1}$. This does not add any edges to the cut, thus $c_d = c_{d - 1}$.

For $d > d_0$ of the other parity, the root cluster of the graph $G^{\mathrm{bad}}_d$ has $G^{\mathrm{bad}}_{d - 1}$ as one child cluster, and $k - 1$ children which are simply $(k, d - 1)$-semi-hypercubes. The edge sets between those $k - 1$ children are arbitrary, and the edge set between the child $G^{\mathrm{bad}}_{d - 1}$ and any of the other $k - 1$ children matches the vertices not in the hard subset of $G^{\mathrm{bad}}_{d - 1}$ with arbitrary vertices in the other child, and does not match any vertex in the hard subset. The hard subset of the graph is simply the hard subset of the child $G^{\mathrm{bad}}_{d - 1}$, thus $s_d = s_{d - 1}$. This does not add any edges to the cut, thus $c_d = c_{d - 1}$.

We thus have for all $d' \geq d_0$ that $s_{d'} = k \cdot 2^{\lfloor (d' - d_0) / 2 \rfloor}$ and $c_{d'} = k$. The average degree of vertices in the hard vertex set of $G^{\mathrm{bad}}_d$ is $\Omega(k)$, and the hard vertex set is smaller than the set of rest of the vertices, thus there is a unit demand on $G^{\mathrm{bad}}_d$ that cannot be routed with congestion $o(k \cdot 2^{-(d' - d_0) / 2})$.

It remains to show that each graph $G^{\mathrm{bad}}_d$ and $G^{\mathrm{spread}}_d$ is a $(k, d, \tau)$-semi-hypercube for some $d_0 \geq d / 2$. First, note that the number of vertices in the graph $G^{\mathrm{spread}}_d$ is at least $(k - 1) \cdot k^{d - 1}$. Now, consider the graph $G^{\mathrm{bad}}_d$:
\begin{itemize}
    \item The number of vertices in the graph $G^{\mathrm{bad}}_d$ is at least $(k - 1) k^{d - 2}$, as it either has size $k + k^{d_0 - 1}$ if $d = d_0$, or it has $G^{\mathrm{spread}}_{d - 1}$ as a child cluster.
    \item The root cluster of the graph has at most one $\tau$-critical child cluster, as $G^{\mathrm{spread}}_{d - 1}$ has size $(k - 1) k^{d - 2} \geq (1 - \tau) k^{d - 1}$.
    \item That child has a matching to the only other child cluster of size equal to that child cluster's size.
    \item Finally, a $(k, d - 1, \tau)$-semi-hypercube of size $s_{d - 1}$ and average degree at least $k - 1$ exists, as each $s_{d'}$ is a multiple of $k$ for all $d' \geq d_0$, and $s_{d - 1} = k \cdot 2^{\lfloor (d - 1 - d_0) / 2 \rfloor} \leq k^{d - 1}$.
\end{itemize}
Finally, consider the graph $G^{\mathrm{spread}}_d$:
\begin{itemize}
    \item The root cluster of the graph has at most one $\tau$-critical child cluster.
    \item That child has a matching to each other child of size at least
    \begin{equation*}
        |V^{\mathrm{bad}}_d| - s_{d - 1} = |V^{\mathrm{bad}}_d| \left(1 - \frac{s_{d - 1}}{|V^{\mathrm{bad}}_d|}\right) \geq |V^{\mathrm{bad}}_d| \left(1 - \frac{s_{d - 1}}{(k - 1) k^{d - 3}}\right) \geq |V^{\mathrm{bad}}_d| \left(1 - \frac{1}{k^{d_0 - 4}}\right).
    \end{equation*}
    Thus, the matching is large enough for $d_0$ large enough, and in fact the counterexample would still exist even if the constant $\frac{9}{10}$ in the definition of a $(k, d, \tau)$-semi-hypercube was changed to $1 - n^{-(1 - \epsilon)}$.
\end{itemize}
\end{proof}

%

\section{Deterministic Dynamic Routing on a Semi-Hypercube}\label{sec:det-routing}



In the previous section, we showed the existence of easy-to-sample-from oblivious routings on semi-hypercubes, which enables routing of demands on the semi-hypercube with low congestion over short paths. However, in some applications, it is important to \emph{maintain} with low recourse a routing for a demand while the graph is undergoing updates. In this section, we give a deterministic algorithm for maintaining such a routing, at the cost of losing a factor of $2^{\bigO(d)}$ in path length to the oblivious routing. 

We specifically consider maintaining an integral routing of an integral demand $D$ consisting of demand pairs $(s, t, \id)$, where $s$ and $t$ are the endpoint vertices of the demand, and $\id$ is a unique identifier for the demand pair. The routing should assign a single path $P(\id)$ to each demand pair. The recourse of an update to the graph and/or the demand is the number of demand pairs with changed routing paths.

The routing in question is constructed recursively: first, contract each child cluster of the root cluster into a single vertex and pick a path for each demand pair on this graph, which is essentially a clique with capacitated edges, the capacities being the size of the matching between the two child clusters. Then, after selecting the root-cluster edges taken by each routing path, simply recursively route the induced demands inside the child clusters. 

Suppose first now that the graph is in fact a \emph{noncritical} semi-hypercube. Then, we can simply select for each demand pair a direct edge between the child clusters of the root cluster containing the demand endpoints\footnote{Of course, this is not the best way to route on a clique, but the factor $\Theta(k)$ in congestion that this loses is unavoidable on a general $(k, d, \tau)$-semi-hypercube regardless.}. Balancing the congestion of edges in any individual matching between child clusters, this at worst causes congestion $L / (1 - 2 \tau) = \Theta(L)$ on each matching edge when routing a demand of load $L$.

We must also consider the loads of the recursive demands: the load on a vertex will be the total congestion of root-cluster edges incident to it plus possibly up to $L$ from the portion of the original demand that was already contained in the same child cluster of the root cluster as the vertex. Thus, the load is at most $kL$, but this is too weak of a bound. Instead, note that the total congestion of root-cluster edges incident to the child cluster of the vertex is at most $L$ times the size of the child cluster. If the congestion of edges in any individual matching is balanced, the total congestion of edges incident to the vertex is thus at most $\bigO(L + k)$, where the additive $k$ appears as in the worst case, the congestion on all of the up-to-$(k - 1)$ root cluster edges incident to the vertex round up in congestion. Thus, as long as the load of the initial demand was $L = \Omega(k)$, we obtain a routing of congestion $L \cdot 2^{\bigO(d)}$.

To dynamically maintain this routing under changes to the graph and demand, it suffices to have a basic load balancer data structure for maintaining that the congestion of edges in any particular matching differ by at most one. When a demand pair is added to the data structure, we can simply assign it to the lowest congestion edge, and when a demand pair is removed, we only have to change the assigned edge of at most one other demand pair to maintain the property. When an edge is added to or removed from the data structure, we might have to change the assigned edges of up to congestion-many demand pairs. Since every edge appears in exactly one matching, we obtain an update recourse of
\begin{itemize}
    \item $1$ for adding a demand pair to the maintained demand, as the only change to the recursive demands is the addition of two recursive demands with the same $\id$.
    \item $2^{\bigO(d)}$ for removing a demand pair from the maintained demand, as when we remove the demand pair from a matching, the assigned edge of up to one other demand pair changes, which adds a removal of one demand pair and an addition of one demand pair in each of the child clusters the matching is connecting. Thus, the number of demand pairs removed grows by a constant factor at every depth processed.
    \item $L \cdot 2^{\bigO(d)}$ for adding or removing an edge from the graph, as the removed edge appears in exactly one matching, and removing it from that matching causes the removals and additions of congestion-many recursive demands.
\end{itemize}

Now, we drop the assumption that the graph is noncritical, and consider maintaining the routing on a general $(k, d, \tau)$-semi-hypercube. This introduces a complication in that when one of the demand endpoints is in a critical child cluster of the root cluster, we cannot simply select a direct edge connecting the child clusters containing the endpoints, as the matching between the critical child cluster and the other child could be very small or even empty.

The guarantee we do have is that the total edge set between the critical child cluster and the other child clusters is large. Thus, a natural approach is to first take an edge out of the critical child cluster, then take a direct edge between the non-critical child cluster containing the endpoint of that edge and the child cluster containing the other endpoint of the demand as before. Balancing congestion over the overall edge set between the critical child and the other child clusters, this would indeed result in a routing with congestion $L \cdot 2^{\bigO(d)}$.

However, this routing cannot be maintained with low recourse. Consider a graph update in which a child cluster of a cluster with no critical child clusters prior to the update becomes critical: one would need to change the routing of all demand pairs in the cluster with an endpoint in the critical child. There needs to be a "smooth transition" to a child cluster becoming critical.

For this, we set a congestion limit to the load balancer data structure of each matching, with any demand pairs we try to assign to the matching that would result in an average congestion higher than the limit instead being moved to the \textit{overflow demand} on the cluster. Setting this congestion limit higher than the maximum congestion that could result between two noncritical child clusters, we guarantee that all overflow demand pairs must have the critical child cluster as one endpoint. Note that we still cannot allow the congestions of edges in the matching to vary arbitrarily as long as they are below the limit, as even in a non-critical semi-hypercube, this would result in recursive demands on child clusters of load $k$ times larger than the load of the recursive demand on this cluster, as explained earlier.

Finally, to route the overflow demand inside the cluster, to avoid having to maintain which child cluster is the critical child cluster, which could again lead to the same recourse issue, for each demand pair in the overflow, we simply take a load-balanced edge out of \emph{both} of the child clusters containing the two endpoints, and finally take a load-balanced direct edge between the two child clusters those edges connected to. Thus, the routing path of a demand pair that is part of the overflow in a cluster $\sigma$ contains three cluster-$\sigma$-edges.

The specific result obtained is the following:

\begin{theorem}\label{thm:det-routing}
    Let $G = (V, E)$ be a $(k, d, \tau)$-semi-hypercube for $\tau \leq \frac{1}{4}$ and $D$ an integral $L$-load demand on $G$. There is a deterministic algorithm that maintains a routing of $D$ of congestion $\max(k, L) \cdot 2^{\bigO(d)}$ and length $2^{\bigO(d)}$ under changes to the demand and vertex/edge additions and deletions, subject to the requirement that $G$ remains a $(k, d, \tau)$-semi-hypercube and $D$ a $L$-load demand on $G$.

    In one update, let $\Delta_G$ be the total number of vertex and edge insertions and deletions, $\Delta_{D+}$ and $\Delta_{D-}$ the total number of added and deleted demand pairs respectively and $\Delta_{D} := \Delta_{D+} + \Delta_{D-}$. Then,
    \begin{itemize}
        \item Recourse: at most $(\max(k, L) \Delta_G + \Delta_{D-}) \cdot 2^{\bigO(d)} + \Delta_{D+}$ paths in the routing are added, removed or changed.
        \item Work: processing the batch takes $(\Delta_G + \Delta_D) \cdot 2^{\bigO(d)} \cdot L \poly(k)$ work. 
    \end{itemize}
\end{theorem}
Note that both guarantees are worst-case, and there is no limit to the number of batches processed. When a vertex is removed, the same update must remove all demand with it as an endpoint and every edge incident to it.

Since the recourse when adding a single demand pair is $1$ (i.e. the only routing path that changed is the routing path for the new demand pair), this algorithm can additionally be viewed as a deterministic streaming algorithm for routing on a $(k, d, \tau)$-semi-hypercube.

\subsection{Static Routing Algorithm}

In this subsection, we show that the (static) initialization of the dynamic deterministic routing data structure (\Cref{alg:det-routing-data}) produces a routing with congestion $\max(k, L) \cdot \bigO(1)^d$ and length $\bigO(1)^d$ for an integral $L$-load demand on a $(k, d, \tau)$-semi-hypercube. Then, in \Cref{sec:det-dyn-routing}, we show how to dynamically maintain such a routing under graph and demand updates.

Crucial to the dynamic algorithm is the following basic load-balancer data structure.

\begin{restatable}{lemma}{loadbalancerlemma}\label{lem:load-balancer}
There exists a data structure $\mathrm{LoadBalancer}$ that maintains a set of \textit{clients} $I$ and \textit{buckets} $E$ along with an assignment $B$ of clients to buckets, such that for every $e \in E$, the number of clients $\id \in I$ assigned to $e$ satisfies
\begin{equation*}
    \lfloor |I| / |E| \rfloor \leq |\{id \in I : B(id) = e\}| \leq \lceil |I| / |E| \rceil
\end{equation*}
under updates $\textsc{UpdateLoadBalancer}(B, I^+, I^-, E^+, E^-)$, where $I \gets (I \setminus I^-) \cup I^+$ and $E \gets (E \setminus E^-) \cup E^+ $, returning $I^{A+}$, $I^{A-}$ and $B^{\mathrm{rem}}$, where
\begin{itemize}
    \item $I^{A+}$ contains clients that were either added or had a changed assigned bucket ($I^+ \subseteq I^{A+}$).
    \item $I^{A-}$ contains clients that were either removed or had a changed assigned bucket ($I^- \subseteq I^{A-}$).
    \item $B^{\mathrm{rem}}$ is a mapping from $I^{A-}$ to the buckets those clients were assigned to before the update.
\end{itemize}

The load balancer data structure has 
\begin{itemize}
    \item Recourse $|I^{A+}| + |I^{A-}| \leq |I^+| + 3|I^-| + 2\lceil |I| / |E| \rceil \cdot (|E^-| + |E^+|)$ (for pre-update $|I| / |E|$).
    \item Work $\tilde{\bigO}(|I^{A+}| + |I^{A-}| + |E^+|)$ (where $\tilde{O}$ hides $\poly\log(|I| + |I^+| + |E| + |E^+|)$-factors). 
\end{itemize} 
\end{restatable}

For a proof of \Cref{lem:load-balancer}, see \Cref{sec:appendix-loadbalancer}.

\begin{lemma}\label{lem:det-routing-init}
    Let $G = (V, E)$ be a $(k, d, \tau)$-semi-hypercube for $\tau \leq \frac{1}{4}$ and $D$ an $L$-load demand on $G$. Then, \Cref{alg:det-routing-data} initialized with $(V, E, D, k, d, \tau, L)$ produces a routing for $D$ of congestion $\max(k, L) \cdot 20^d$ and length $4^d$.
\end{lemma}

\begin{proof}
    \Cref{alg:det-static-routing} constructs a routing for the demand in a top-down, recursive fashion, selecting for every demand pair $(a, b, \id) \in D$ an ordered list of up to three edges between child clusters of the root cluster, and creating "recursive" demands contained in child clusters of the root cluster to connect the endpoints of the demand and the endpoints of the root-cluster edges selected for the demand pair in order. For example, if edges $(a', c'), (e', f'), (d', b')$ are selected for the demand pair $(a, b, \id) \in D$ (with $a$ and $a'$, $c'$ and $e'$, $f'$ and $d'$, and $b'$ and $b$ sharing the same child clusters), the demands $(a, a', 2\id + 1)$, $(c', e', 2\id)$, $(f', d', 2\id)$, $(b', b, 2\id + 1)$ are added to the recursive demands, and the assigned path for the demand pair will equal the concatenation of the paths assigned to the recursive demands interspersed with the edges $(a', c')$, $(e', f')$ and $(d', b')$. It immediately follows that paths in the routing have length at most $4^d$.
    
    To analyze congestion, we bound the loads of the recursive demands at every level. Denote by $D_\sigma$ the recursive demand at cluster $\sigma$, with $D_{\rt} = D$. As the load of a vertex in the recursive demand $D_{\sigma + i}$ of a child cluster $\sigma + i$ of a cluster $\sigma$ equals the load of the vertex in the demand $D_{\sigma}$ of $\sigma$ plus the congestion of cluster-$\sigma$ edges incident on $v$, the congestion of any cluster-$\sigma$ edge is at most the maximum load of a child clusters' recursive demand $D_{\sigma + i}$.

    We will show by induction on $|\sigma|$ that the load of $D_\sigma$ is at most $L_{|\sigma|} := \max(k, L) \cdot 20^{|\sigma|}$. This holds for the root cluster $\sigma = \rt$ as the initial demand $D$ has load $L$. 

    \begin{algorithm}[H]
    \caption{$\textsc{DynDetRouter}$: data, initialization, path queries} \label{alg:det-routing-data}
    \begin{algorithmic}[1]
        \Class{DynDetRouter}
            \Data
                \LeftComment{Constants}
                \State Width $k \in \mathbb{N}$
                \State Dimension $d \in \mathbb{N}$
                \State Threshold $\tau \in [0, \frac{1}{4}]$
                \State Max load $L \in \mathbb{N}$
                \LeftComment{Variables}
                \State Demands $D^{\mathrm{base}}_{\sigma, i, j}, D^{\mathrm{ovf}}_{\sigma, i, j}$
                \State Load balancers $B_{\sigma, i}^{\mathrm{ovf}}$
                \State Load balancers $B_{\sigma, i, j}^{\mathrm{dir}}$
            \EndData
            
            \LeftComment{Initializes the data structure with $(k, d, \tau)$-semi-hypercube $G = (V, E)$ and $L$-load demand $D$}
            \LeftComment{Afterwards, $\textsc{GetPath}(\epsilon, (a, b, \id))$ returns path assigned to $(a, b, \id) \in D$}
            \Function{Initialize}{$V$, $E$, $D$, $k$, $d$, $\tau$, $L$}
                \For{$\sigma \in [k]^{< d}$ }
                    \For{$i \in [k]$ }
                        \State $\Call{UpdateLoadBalancer}{B_{\sigma, i}^{\mathrm{ovf}}, \emptyset, \emptyset, E_{\sigma, i, -i}, \emptyset}$
                        \For{$j \in [k]$, $i < j$ }
                            \State $\Call{UpdateLoadBalancer}{B_{\sigma, i, j}^{\mathrm{dir}}, \emptyset, \emptyset, E_{\sigma, i, j}, \emptyset}$
                        \EndFor
                    \EndFor
                \EndFor
                \State $\Call{StaticRecRoute}{\epsilon, D}$.
            \EndFunction
            \State
            \LeftComment{Returns path within cluster $\sigma$ assigned to demand pair $(a, b, \id)$}
            \Function{GetPath}{$\sigma$, $(a, b, \id)$}
                \LeftComment{If the demand pair is part of overflow, get paths out of endpoint clusters}
                \If{$\id \in D^{\mathrm{ovf}}_{\sigma}$}
                    \State Let $i, j$ be the child clusters $V_{\sigma + i} \ni a$, $V_{\sigma + j} \ni b$ containing $a, b$
                    \State Let $(a', c') \gets B^{\mathrm{ovf}}_i(\id)$
                    \State Let $(b', d') \gets B^{\mathrm{ovf}}_j(\id)$
                    \State Let $P_{\mathrm{tail}} = \Call{concat}{\Call{GetPath}{\sigma + i, (a, a', 2\id + 1)}, (a', c')}$
                    \State Let $P_{\mathrm{head}} = \Call{concat}{(d', b'), \Call{GetPath}{\sigma + j, (b', b, 2\id + 1)}}$
                    \State $a \gets c'$, $b \gets d'$
                \Else
                    \State Let $P_{\mathrm{tail}}$ and $P_{\mathrm{head}}$ be empty paths
                \EndIf
                \State
                \LeftComment{Compute directly routed part}
                \State Let $i, j$ be the child clusters $V_{\sigma + i} \ni a$, $V_{\sigma + j} \ni b$ containing $a, b$
                \If{$i = j$}
                    \State \Return $\Call{concat}{P_{\mathrm{tail}}, \Call{GetPath}{\sigma + i, (a, b, 2\id)}, P_{\mathrm{head}}}$
                \Else
                    \State Let $(a', b') \gets B^{\mathrm{dir}}_{i, j}(\id)$
                    \State Let $P_{\mathrm{tail}} \gets \Call{concat}{P_{\mathrm{tail}}, \Call{GetPath}{\sigma + i, (a, a', 2\id)}}$
                    \State Let $P_{\mathrm{head}} \gets \Call{concat}{\Call{GetPath}{\sigma + j, (b', b, 2\id)}, P_{\mathrm{head}}}$
                    \State \Return $\Call{concat}{P_\mathrm{tail}, (a', b'), P_{\mathrm{head}}}$
                \EndIf
            \EndFunction
            \algstore{DynDetRouter}
    \end{algorithmic}
\end{algorithm}

\addtocounter{algorithm}{-1}
\begin{algorithm}[H]
    \caption{$\textsc{DynDetRouter}$: Static Routing} \label{alg:det-static-routing}
    \begin{algorithmic}[1]
            \algrestore{DynDetRouter}
            \Function{StaticRecRoute}{$\sigma$, $D$}
                \LeftComment{Step 1: split the demand into direct demand and overflow}
                \For{$i, j \in [k]$, $i < j$ }
                    \State Let $L_{|\sigma|} := \max(k, L) \cdot 20^{|\sigma|}$
                    \State Let $T^{\mathrm{ovf}} \gets \max(0, 4L_{|\sigma|} \cdot (|E_{\sigma, i, j}| - \frac{1}{4} k^{d - |\sigma| - 1}))$
                    \State Let $D_{i, j}$ be the subdemand of $D$ from $V_{\sigma + i}$ to $V_{\sigma + j}$
                    \State $D^{\mathrm{base}}_{\sigma, i, j}$ $\gets$ a subdemand of $D_{i, j}$ of size $\min(|D_{i, j}|, T^{\mathrm{ovf}})$
                    \State $D^{\mathrm{ovf}}_{\sigma, i, j} \gets D_{i, j} \setminus D^{\mathrm{base}}_{\sigma, i, j}$
                \EndFor
                \State Let $D^{\mathrm{dir}} \gets D \setminus \bigcup_{i < j} D^{\mathrm{ovf}}_{\sigma, i, j}$
                \State
                \LeftComment{Step 2: route overflow out of its endpoint clusters}
                \For{$i \in [k]$ }
                    \State Let $D^{\mathrm{ovf}}_{i, -i} \gets \bigcup_{j \neq i} D^{\mathrm{ovf}}_{i, j}$
                    \State $\Call{UpdateLoadBalancer}{B^{\mathrm{ovf}}_{\sigma, i}, D^{\mathrm{ovf}}_{i, -i}, \emptyset, \emptyset, \emptyset}$
                \EndFor
                \State
                \State Let $D_{\sigma + i}$ be an empty demand for every $i \in [k]$
                \For{$i, j \in [k]$, $i < j$ }
                    \For{$(a, b, \id) \in D^{\mathrm{ovf}}_{i, j}$ }
                        \State Let $(a', c') \gets B^{\mathrm{ovf}}_{\sigma, i}(\id)$
                        \State Let $(b', d') \gets B^{\mathrm{ovf}}_{\sigma, j}(\id)$
                        \State Add $(a, a', 2\id + 1)$ to $D_{\sigma + i}$
                        \State Add $(b', b, 2\id + 1)$ to $D_{\sigma + j}$
                        \State Add $(c', d', \id)$ to $D^{\mathrm{dir}}$
                    \EndFor
                \EndFor
                \State
                \LeftComment{Step 3: route direct demand}
                \For{$i, j \in [k]$, $i < j$ }
                    \State Let $D^{\mathrm{dir}}_{i, j}$ be the subdemand of $D^{\mathrm{dir}}$ from $V_{\sigma + i}$ to $V_{\sigma + j}$
                    \State $\Call{UpdateLoadBalancer}{B^{\mathrm{dir}}_{\sigma, i, j}, D^{\mathrm{dir}}_{i, j}, \emptyset, \emptyset, \emptyset}$
                    \For{$(a, b, \id) \in D^{\mathrm{dir}}_{i, j}$ }
                        \State Let $(a', b') \gets B^{\mathrm{dir}}_{\sigma, i, j}(\id)$
                        \State Add $(a, a', 2\id)$ to $D_{\sigma + i}$
                        \State Add $(b', b, 2\id)$ to $D_{\sigma + j}$
                    \EndFor
                \EndFor
                \State
                \LeftComment{Recursively route demand contained within clusters}
                \If{$|\sigma| + 1 < d$}
                    \For{$i \in [k]$ }
                        \State Let $D^{\mathrm{dir}}_i$ be the subdemand of $D^{\mathrm{dir}}$ contained in $V_{\sigma + i}$
                        \For{$(a, b, \id) \in D^{\mathrm{dir}}_i$ }
                            \State Add $(a, b, 2\id)$ to $D_{\sigma + i}$
                        \EndFor
                        \State \Call{StaticRecRoute}{$\sigma + i, D_{\sigma + i}$}
                    \EndFor
                \EndIf
            \EndFunction
            \algstore{DynDetRouter}
    \end{algorithmic}
\end{algorithm}

    
    Consider the recursive call to cluster $\sigma$ in \Cref{alg:det-static-routing}. The edges between child clusters of $V_\sigma$ are chosen as follows:
    \begin{enumerate}
        \item[Step $1$.] First, for every pair $\sigma + i$, $\sigma + j$ of child clusters, if the amount $|D_{\sigma, i, j}|$ of demand pairs in $D_\sigma$ between the child clusters exceeds the overflow threshold $T^{\mathrm{ovf}}_{\sigma, i, j} := \max(0, 4L_{|\sigma|} \cdot (|E_{\sigma, i, j}| - \frac{1}{4} k^{d - |\sigma| - 1}))$, the excess demand is designated as \textit{overflow} demand: $D^{\mathrm{ovf}}_{\sigma, i, j}$ is a subdemand of $D_{\sigma, i, j}$ of size $\max(0, |D_{\sigma, i, j}| - T^{\mathrm{ovf}}_{\sigma, i, j})$. The demand $D^{\mathrm{dir}}_\sigma$ contains the demand pairs in $D_\sigma$ that are not part of the overflow.
        \item[Step $2$.] Then, for every child cluster $\sigma + i$, the load balancer data structure $B^{\mathrm{ovf}}_{\sigma, i}$ is used with clients $D^{\mathrm{ovf}}_{\sigma, i, -i} := \bigcup_{j \neq i} D^{\mathrm{ovf}}_{\sigma, i, j}$ and buckets $E_{\sigma, i, -i}$ to assign to every overflow demand pair $(a, b, \id) \in D^{\mathrm{ovf}}_{\sigma, i, -i}$ an edge $B^{\mathrm{ovf}}_{\sigma, i}(\id) \in E_{\sigma, i, -i}$, such that every edge in $E_{\sigma, i, -i}$ has either $\lfloor |D^{\mathrm{ovf}}_{\sigma, i, -i}| / |E_{\sigma, i, -i}| \rfloor$ or $\lceil |D^{\mathrm{ovf}}_{\sigma, i, -i}| / |E_{\sigma, i, -i}| \rceil$ demand pairs assigned to it.

        Then, for each demand pair $(a, b, \id) \in D^{\mathrm{ovf}}_{\sigma, i, j}$, these edges $(a', c') = B^{\mathrm{ovf}}_{\sigma, i}$ and $(b', d') = B^{\mathrm{ovf}}_{\sigma, j}$ are set as the first and last cluster-$\sigma$ edges of that demand pair, and $(c', d', \id)$ is added to $D^{\mathrm{dir}}_{\sigma}$.
        \item[Step $3$.] Finally, the load balancer data structure $B^{\mathrm{dir}}_{\sigma, i, j}$ is used with clients $D^{\mathrm{dir}}_{\sigma, i, j}$ and buckets $E_{\sigma, i, j}$ to assign to every demand pair $(a, b, \id) \in D^{\mathrm{dir}}_{\sigma, i, j}$ an edge $B^{\mathrm{dir}}_{\sigma, i, j}(\id) \in E_{\sigma, i, j}$, such that every edge in $E_{\sigma, i, j}$ has either $\lfloor |D^{\mathrm{dir}}_{\sigma, i, j}| / |E_{\sigma, i, j}| \rfloor$ or $\lceil |D^{\mathrm{dir}}_{\sigma, i, j}| / |E_{\sigma, i, j}| \rceil$ demand pairs assigned to it. For every demand pair $(a, b, \id) \in D^{\mathrm{dir}}_{\sigma, i, j}$, the edge $(a', b')$ is set as either the only (if the demand pair was not part of the overflow) or the second (if it was) cluster-$\sigma$ edge of the demand pair.
    \end{enumerate}
    
    Now, we start analyzing congestion and load. First, consider the congestion caused to edges in the second step: for a cluster-$\sigma$ edge $e \in E_{\sigma, i, -i}$, let $C^{\mathrm{ovf}}_{e, i} := |\{(a, b, \id) \in D^{\mathrm{ovf}}_{\sigma, i, -i} : e = B^{\mathrm{ovf}}_{\sigma, i, -i}(\id)\}|$, and for $e \in E_{\sigma, i, j}$, let $C^{\mathrm{ovf}}_e := C^{\mathrm{ovf}}_{e, i} + C^{\mathrm{ovf}}_{e, j}$. Let $k' + 1$ be the number of nonempty child clusters of $\sigma$. For every child cluster $\sigma + i$, including the critical child cluster if one exists, we have
    \begin{equation*}
        |E_{\sigma, i, -i}| \geq \frac{1}{2} \cdot k' \cdot |V_{\sigma + i}| \geq \frac{k'}{2L_{|\sigma|}} \cdot \sum_{j \neq i} |D_{\sigma, i, j}|,
    \end{equation*}
    thus for every edge $e \in E_{\sigma, i, -i}$, we have
    \begin{equation*}
        C^{\mathrm{ovf}}_{e, i} \leq |\{(a, b, \id) \in D^{\mathrm{ovf}}_{\sigma, i, -i} : e = B^{\mathrm{ovf}}_{\sigma, i, -i}(\id)\}| \leq \left\lceil \frac{2L_{|\sigma|}}{k'} \right\rceil,
    \end{equation*}
    and for every vertex $v \in V_{\sigma + i}$, we have
    \begin{equation*}
        \sum_{\substack{e\ \in\ E_{\sigma, i, -i}\\\text{$e$ incident on $v$}}} C^{\mathrm{ovf}}_{e} \leq 2k' \cdot \left\lceil \frac{2L_{|\sigma|}}{k'} \right\rceil \leq 2k' + 4L_{|\sigma|} \leq 6L_{|\sigma|}.
    \end{equation*}
    
    Next, consider the congestion caused to edges in the third step: for a cluster-$\sigma$ edge $e \in E_{\sigma, i, j}$, let $C^{\mathrm{dir}}_{e} := |\{a, b, \id) \in D^{\mathrm{dir}}_{\sigma, i, j} : e = B^{\mathrm{dir}}_{\sigma, i, j}(\id)\}|$, and let $D^{\mathrm{base}}_{\sigma, i, j}$ be the subdemand of $D^{\mathrm{dir}}_{\sigma, i, j}$ containing the demand pairs that were not part of overflow (i.e. were added to $D^{\mathrm{dir}}_{\sigma, i, j}$ in step 1).
    
    Note that for any demand pair $(a, b, \id) \in D^{\mathrm{ovf}}_\sigma$ in the overflow, either $a$ or $b$ must be in the critical child cluster $\sigma + t$ of $\sigma$ (and if the cluster has no critical child cluster, there cannot be any overflow), as $|E_{\sigma, i, j}| \geq \frac{1}{2} k^{d - |\sigma| - 1}$ for any pair of noncritical child clusters $\sigma + i$ and $\sigma + j$. Consider any demand pair $(a, b, \id) \in D^{\mathrm{ovf}}_{\sigma, t, j}$, and let $(a', c') := B^{\mathrm{ovf}}_{\sigma, t, -t}(\id)$ and $(b', d') := B^{\mathrm{ovf}}_{\sigma, j, -j}(\id)$ be the cluster-$\sigma$ edges selected for the demand pair in step $2$. Now, either $c'$ and $d'$ are both in noncritical child clusters of $\sigma$, or we can \textit{charge} the demand pair to the edge $(a', c')$. Thus, for every noncritical child cluster $\sigma + i$, we have
    \begin{equation*}
        |D^{\mathrm{dir}}_{\sigma, t, i} \setminus D^{\mathrm{base}}_{\sigma, t, i}| \leq \sum_{e \in E_{\sigma, t, i}} C^{\mathrm{ovf}}_{e, i} \leq |E_{\sigma, t, i}| \cdot \left\lceil\frac{2L_{|\sigma|}}{k'} \right\rceil,
    \end{equation*}
    and, combined with the definition of the overflow threshold,
    \begin{itemize}
        \item if $|E_{\sigma, t, i}| \geq \frac{1}{4} |V_{\sigma + t}|$, then
        \begin{equation*}
            |D^{\mathrm{dir}}_{\sigma, t, i}| \leq |E_{\sigma, t, i}| \cdot \left(\frac{|D_{\sigma, t, i}|}{|E_{\sigma, t, i}|} + \left\lceil\frac{2L_{|\sigma|}}{k'}\right\rceil\right) \leq |E_{\sigma, t, i}| \cdot \left(\frac{4|D_{\sigma, t, i}|}{|V_{\sigma + t}|} + \left\lceil\frac{2L_{|\sigma|}}{k'}\right\rceil\right),
        \end{equation*}
        \item and if $|E_{\sigma, t, i}| \leq \frac{1}{4} |V_{\sigma + t}|$, then $T^{\mathrm{ovf}}_{\sigma, t, i} = 0$, and
        \begin{equation*}
            |D^{\mathrm{dir}}_{\sigma, t, i}| \leq |E_{\sigma, t, i}| \cdot \left\lceil\frac{2L_{|\sigma|}}{k'}\right\rceil \leq |E_{\sigma, t, i}| \cdot \left(\frac{4|D_{\sigma, t, i}|}{|V_{\sigma, t}|} + \left\lceil\frac{2L_{|\sigma|}}{k'}\right\rceil\right),
        \end{equation*}
    \end{itemize}
    and for $e \in E_{\sigma, t, i}$ we have $C^{\mathrm{dir}}_{e} \leq 4|D_{\sigma, t, i}| / |V_{\sigma + t}| + 2L_{|\sigma|} / k' + 2$. Thus, the sum of $C^{\mathrm{dir}}_e$ over the up-to-$k'$ cluster-$\sigma$ edges incident to any $v \in V_{\sigma + t}$ is at most
    \begin{equation*}
        \sum_{i \neq t} \frac{4|D_{\sigma, t, i}|}{|V_{\sigma + t}|} + \frac{2L_{|\sigma|}}{k'} + 2 \leq 4L_{|\sigma|} + 2L_{|\sigma|} + 2k' \leq 8L_{|\sigma|}.
    \end{equation*}
    Finally, for any noncritical child clusters $\sigma + i$, we have $\sum_{j \neq i, t} |D^{\mathrm{dir}}_{\sigma, i, j}| \leq L_{|\sigma|} |V_{\sigma + i}| + |D^{\mathrm{ovf}}_\sigma| \leq 2L_{|\sigma|} |V_{\sigma + i}|$ as $|V_{\sigma + i}| \geq |V_{\sigma + t}|$. Since $|E_{\sigma, i, j}| \geq \frac{1}{2} |V_{\sigma + i}|$ for all $j \neq t$, the sum of $C^{\mathrm{dir}}_e$ over cluster-$\sigma$ edges incident on some vertex $v \in V_{\sigma, i}$ without an endpoint in $V_{\sigma + t}$ is at most $4L_{|\sigma|} + k' \leq 5L_{|\sigma|}$. The edge to the critical cluster, if it exists, has $C^{\mathrm{dir}}_e \leq 8L_{|\sigma|}$, thus the sum of $C^{\mathrm{dir}}_{e}$ over the cluster-$\sigma$ edges incident to $V_{\sigma, i}$ is at most $13L_{|\sigma|}$.

    Thus, the congestion of any cluster-$\sigma$ edge is at most $19L_{|\sigma|}$, and the loads of the recursive demands $D_{\sigma + i}$ are at most $L_{|\sigma|} + 6L_{|\sigma|} + 13L_{|\sigma|} \leq 20L_{|\sigma|} = L_{|\sigma + i|}$ as desired.
\end{proof}

\subsection{Deterministically Maintaining a Routing}\label{sec:det-dyn-routing}

Using the load balancer data structure, the routing can be made dynamic in a straight-forward way. Updates are performed in the same top-down, recursive fashion as with static initialization. Each cluster receives sets of demand changes $D^+_\sigma$ and $D^-_\sigma$ along with the edge updates $E^+_\sigma$, $E^-_\sigma$, reassigns cluster-$\sigma$-edges for some demand pairs (always including those assigned to a newly-removed edge or a added/removed demand pair), and produces recursive demand changes $D^+_{\sigma + i}$, $D^-_{\sigma + i}$ for child clusters. After the recursive updates, the update function returns the set $I^{\Delta}$ of identifiers $\id$ of demand pairs for which the assigned path changed.

By the properties of the load balancer data structure, it is straightforward to check that the dynamic update maintains a routing that by \Cref{lem:det-routing-init} routes the demand with congestion $\max(k, L) \cdot 20^d$ and length $4^d$. The main content of this section is dedicated to showing the following bound on the recourse of the dynamic update:

\begin{lemma}\label{lem:det-routing-recourse}
    Each call $\textsc{DynamicRecUpdate}(\sigma, E^+, E^-, D^+, D^-)$ satisfies
    \begin{equation*}
        \sum_{i \in [k]} |D^+_{\sigma + i}| + |D^-_{\sigma + i}| \leq 72(|D^+_\sigma| + |D^-_\sigma|) + 320 \max(k, L) \cdot 20^{|\sigma|} \cdot (|E^+_\sigma| + |E^-_\sigma|)
    \end{equation*}
    Further, when the only change is the addition of some demand pairs (i.e. $E^+, E^-$ and $D^-$ are all empty), $\sum_{i \in [k]} |D^-_{\sigma + i}| = 0$ and for all $\id' \in \bigcup_{i \in [k]} D^+_{\sigma + i}$ we have $\lfloor \id' / 2 \rfloor \in D^+_\sigma$; the only demand pairs for which the assigned path changes are those in $D^+_{\sigma + i}$.
\end{lemma}

\begin{proof}
    We analyze the steps one by one.

    \addtocounter{algorithm}{-1}
\begin{algorithm}[H]
    \caption{$\textsc{DynDetRouter}$: Dynamic Routing Update (step 1/3)} \label{alg:det-dyn-routing-step1}
    \begin{algorithmic}[1]
            \algrestore{DynDetRouter}
            \LeftComment{Updates the routing with a change to the vertex and edge sets and demand}
            \LeftComment{Done in two calls for easy proof of streaming property}
            \Function{DynamicUpdate}{$V^+$, $V^-$, $E^+$, $E^-$, $D^+$, $D^-$}
                \State Let $I^{\Delta} \gets \Call{DynamicRecUpdate}{\epsilon, E^+, E^-, \emptyset, D^-}$
                \State \Call{DynamicRecUpdate}{$\epsilon, \emptyset, \emptyset, D^+, \emptyset$}
                \State \Return $I^{\Delta} \cup \{\id : (a, b, \id) \in D^+\}$
            \EndFunction
            \State
            \Function{DynamicRecUpdate}{$\sigma$, $E^+$, $E^-$, $D^+$, $D^-$}
                \LeftComment{Step 1: split the demand into direct demand and overflow}
                \State Let $D^{\mathrm{base+}}_{i, j}$, $D^{\mathrm{base-}}_{i, j}$, $D^{\mathrm{ovf+}}_{i, j}$ and $D^{\mathrm{ovf-}}_{i, j}$ be empty demands
                \For{$i, j \in [k]$, $i < j$ }
                    \LeftComment{Process removed demand}
                    \State Let $D^+_{i, j}$ and $D^-_{i, j}$ be the subdemands of $D^+$ and $D^-$ from $V_{\sigma + i}$ to $V_{\sigma + j}$
                    \State Let $D^{\mathrm{ovf}-}_{i, j} \gets D^{\mathrm{ovf}}_{\sigma, i, j} \cap D^-_{i, j}$ and $D^{\mathrm{base}-}_{i, j} \gets D^{\mathrm{base}}_{\sigma, i, j} \cap D^-_{i, j}$
                    \State Let $D^{\mathrm{ovf}}_{\sigma, i, j} \gets D^{\mathrm{ovf}}_{\sigma, i, j} \setminus D^{\mathrm{ovf}-}_{\sigma, i, j}$ and $D^{\mathrm{base}}_{\sigma, i, j} \gets D^{\mathrm{base}}_{\sigma, i, j} \setminus D^{\mathrm{base}-}_{\sigma, i, j}$
                    \State
                    \LeftComment{Update threshold and process added demand}
                    \State Let $L_{|\sigma|} := \max(k, L) \cdot 20^{|\sigma|}$
                    \State Let $T^{\mathrm{ovf}} \gets \max(0, 4L_{|\sigma|} \cdot (|E_{\sigma, i, j} \cup E^+_{\sigma, i, j} \setminus E^-_{\sigma, i, j}| - \frac{1}{4} k^{d - |\sigma| - 1}))$
                    \State Let $\mathrm{rf} \gets \min(T^{\mathrm{ovf}} - |D^{\mathrm{base}}_{\sigma, i, j}|, |D^{\mathrm{ovf}}_{\sigma, i, j}| + |D^+_{i, j}|)$
                    \If{$\mathrm{rf} \geq 0$}
                        \State Let $D^{\mathrm{base+}}_{i, j}$ be a subdemand of $D^{\mathrm{ovf}}_{\sigma, i, j} \cup D^+_{i, j}$ of size $\mathrm{rf}$
                        \State Let $D^{\mathrm{ovf-}}_{i, j} \gets D^{\mathrm{ovf-}}_{i, j} \cup (D^{\mathrm{ovf}}_{\sigma, i, j} \cap D^{\mathrm{base+}}_{i, j})$
                        \State Let $D^{\mathrm{ovf+}}_{i, j} \gets D^+_{i, j} \setminus D^{\mathrm{base}+}_{i, j}$
                    \Else
                        \State Let $D^{\mathrm{ovf}+}_{i, j}$ be a subdemand of $D^{\mathrm{base}}_{\sigma, i, j}$ of size $-\mathrm{rf}$
                        \State Let $D^{\mathrm{base}-}_{i, j} \gets D^{\mathrm{base}-}_{i, j} \cup D^{\mathrm{ovf}+}_{i, j}$
                        \State Let $D^{\mathrm{ovf}+}_{i, j} \gets D^{\mathrm{ovf}+}_{i, j} \cup D^+_{i, j}$
                    \EndIf
                    \State $D^{\mathrm{ovf}}_{\sigma, i, j} \gets D^{\mathrm{ovf}}_{\sigma, i, j} \setminus D^{\mathrm{ovf}-}_{\sigma, i, j}$
                    \State $D^{\mathrm{base}}_{\sigma, i, j} \gets D^{\mathrm{base}}_{\sigma, i, j} \setminus D^{\mathrm{base}-}_{\sigma, i, j}$
                \EndFor
                \LeftComment{Update direct demand}
                \State Let $D^{\mathrm{dir}+} \gets \left(\bigcup_{i} D^{+}_{i, i}\right) \cup (\bigcup_{i < j} D^{\mathrm{base}+}_{i, j})$
                \State Let $D^{\mathrm{dir}-} \gets \left(\bigcup_{i} D^{-}_{i, i}\right) \cup (\bigcup_{i < j} D^{\mathrm{base}-}_{i, j})$
            \algstore{DynDetRouter}
    \end{algorithmic}
\end{algorithm}

    \begin{enumerate}
        \item[Step 1.] This step updates the split of the demand into base and overflow demand.
        
        Consider one pair of child clusters $i, j \in [k]$. The threshold $T^{\mathrm{ovf}} := \max(0, 4L_{|\sigma|} \cdot (|E_{\sigma, i, j}| - \frac{1}{4}k^{d - |\sigma| - 1})$ changes by at most $|\Delta(T^{\mathrm{ovf}})| \leq 4L_{|\sigma|} \cdot (|E^+_{\sigma, i, j}| + |E^-_{\sigma, i, j}|)$, and the required fill $\mathrm{rf}$ is bounded by $|\mathrm{rf}| \leq |D^+_{\sigma, i, j}| + |D^-_{\sigma, i, j}| + |\Delta(T^{\mathrm{ovf}})|$, as before the update either $T^{\mathrm{ovf}} = |D^{\mathrm{base}}_{\sigma, i, j}|$ or $|D^{\mathrm{ovf}}_{\sigma, i, j}| = 0$ held, and $T^{\mathrm{ovf}} \geq |D^{\mathrm{base}}_{\sigma, i, j}|$ always holds. Now, if the required fill $\mathrm{rf}$ is nonnegative, we have
        \begin{itemize}
            \item $|D^{\mathrm{base}-}_{i, j}| \leq |D^-_{i, j}|$
            \item $|D^{\mathrm{base}+}_{i, j}| \leq |\mathrm{rf}|$
            \item $|D^{\mathrm{ovf}-}_{i, j}| \leq |D^-_{i, j}| + |\mathrm{rf}|$
            \item $|D^{\mathrm{ovf}+}_{i, j}| \leq |D^+_{i, j}|$
        \end{itemize}
        Otherwise, if the required fill is negative,
        \begin{itemize}
            \item $|D^{\mathrm{base}-}_{i, j}| \leq |D^-_{i, j}| + |\mathrm{rf}|$
            \item $|D^{\mathrm{base}+}_{i, j}| \leq 0$
            \item $|D^{\mathrm{ovf}-}_{i, j}| \leq |D^-_{i, j}|$
            \item $|D^{\mathrm{ovf}+}_{i, j}| \leq |D^+_{i, j}| + |\mathrm{rf}|$
        \end{itemize}
        thus, in either case,
        \begin{align*}
            |D^{\mathrm{base}+}_{i, j}| + |D^{\mathrm{base}-}_{i, j}| + |D^{\mathrm{ovf}+}_{i, j}| + |D^{\mathrm{ovf}-}_{i, j}| &\leq 2(|D^+_{i, j}| + |D^-_{i, j}| + |\mathrm{rf}|)\\
            &\leq 3(|D^+_{i, j}| + |D^-_{i, j}|) + 8L_{|\sigma|} \cdot (|E^+_{\sigma, i, j}| + |E^-_{\sigma, i, j}|)
        \end{align*}
        and in particular, summing over $i, j \in [k]$ and noting that demand pairs in $D^+_{i, i}$ and $D^-_{i, i}$ do not contribute to the terms $ |D^{\mathrm{base}+}_{i, j}| + |D^{\mathrm{base}-}_{i, j}| + |D^{\mathrm{ovf}+}_{i, j}| + |D^{\mathrm{ovf}-}_{i, j}|$, we get
        \begin{equation} \label{eq:detroute-recourse-bound-1}
            |D^{\mathrm{dir}+}| + |D^{\mathrm{dir}-}| + \sum_{i, j} |D^{\mathrm{ovf}+}_{i, j}| + |D^{\mathrm{ovf}-}_{i, j}| \leq 3(|D^+| + |D^-|) + 8L_{|\sigma|} \cdot (|E^+_{\sigma}| + |E^-_{\sigma}|).
        \end{equation}
        Finally, if $E^+$, $E^-$ and $|D^-|$ are all empty, the thresholds cannot change, and each demand pair in $D^+$ can appear in at most one of $D^{\mathrm{base}+}_{i, j}$, $D^{\mathrm{ovf}+}_{i, j}$ and not at all in $D^{\mathrm{base}-}_{i, j}$ and $D^{\mathrm{ovf}-}_{i, j}$, thus after the step, $D^{\mathrm{dir}-}$ and each $D^{\mathrm{base}-}_{i, j}, D^{\mathrm{ovf}-}_{i, j}$ are empty, and
        \begin{equation*}
            |D^{\mathrm{dir}+}| + \sum_{i, j} |D^{\mathrm{ovf}+}_{i, j}| \leq |D^+|.
        \end{equation*}

    \addtocounter{algorithm}{-1}
\begin{algorithm}[H]
    \caption{$\textsc{DynDetRouter}$: Dynamic Routing Update (step 2/3)} \label{alg:det-dyn-routing-step2}
    \begin{algorithmic}[1]
            \algrestore{DynDetRouter}
                \LeftComment{Step 2: route overflow out of its endpoint clusters}
                \State Let $I^+$ and $I^-$ be empty sets and $E^{\mathrm{rem}}_i$ empty maps.
                \For{$i \in [k]$ }
                    \State Let $D^{\mathrm{ovf}+}_{i, -i} \gets \bigcup_{j \neq i} D^{\mathrm{ovf}+}_{\min(i, j), \max(i, j)}$
                    \State Let $D^{\mathrm{ovf}-}_{i, -i} \gets \bigcup_{j \neq i} D^{\mathrm{ovf}-}_{\min(i, j), \max(i, j)}$
                    \State $(I^+_i, I^-_i, B^{\mathrm{rem}}_i) \gets \Call{UpdateLoadBalancer}{B^{\mathrm{ovf}}_{\sigma, i}, D^{\mathrm{ovf}+}_{i, -i}, D^{\mathrm{ovf}-}_{i, -i}, E^+_{\sigma, i, -i}, E^-_{\sigma, i, -i}}$
                    \State $I^+ \gets I^+ \cup I^+_i$
                    \State $I^- \gets I^- \cup I^-_i$
                \EndFor
                \State
                \State Let $D^+_{\sigma + i}$ and $D^-_{\sigma + i}$ be empty demands for every $i \in [k]$
                \For{$(a, b, \id) \in I^-$ }
                    \State Let $V_{\sigma, i} \ni a$ and $V_{\sigma, j} \ni b$ be the child clusters containing $a$ and $b$
                    \State Let $(a', c') \gets B^{\mathrm{rem}}_i(\id)$
                    \State Let $(b', d') \gets B^{\mathrm{rem}}_j(\id)$
                    \State Add $(a, a', 2\id + 1)$ to $D^-_{\sigma + i}$
                    \State Add $(b', b, 2\id + 1)$ to $D^-_{\sigma + j}$
                    \State Add $(c', d', \id)$ to $D^{\mathrm{dir}-}$
                \EndFor
                \For{$(a, b, \id) \in I^+$ }
                    \State Let $V_{\sigma, i} \ni a$ and $V_{\sigma, j} \ni b$ be the child clusters containing $a$ and $b$
                    \State Let $(a', c') \gets B^{\mathrm{ovf}}_{\sigma, i}(\id)$
                    \State Let $(b', d') \gets B^{\mathrm{ovf}}_{\sigma, j}(\id)$
                    \State Add $(a, a', 2\id + 1)$ to $D^+_{\sigma + i}$
                    \State Add $(b', b, 2\id + 1)$ to $D^+_{\sigma + j}$
                    \State Add $(c', d', \id)$ to $D^{\mathrm{dir}+}$
                \EndFor
            \algstore{DynDetRouter}
    \end{algorithmic}
\end{algorithm}

        \item[Step 2.] This step updates the routing of overflow out of its endpoint clusters.

        Consider one cluster $V_{\sigma + i}$. By \Cref{lem:det-routing-init}, the maximum congestion $C^{\mathrm{ovf}}_e$ caused to any edge by routing the overflow demand is at most $6L_{|\sigma|}$, thus by the recourse of the load balancer data structure, 
        \begin{equation} \label{eq:detroute-recourse-bound-2}
            |I^+_i| + |I^-_i| \leq |D^{\mathrm{ovf}+}_{i, -i}| + 3|D^{\mathrm{ovf}-}_{i, -i}| + 12L_{|\sigma|}(|E^{+}_{\sigma, i, -i}| + |E^{-}_{\sigma, i, -i}|).
        \end{equation}
        After updating the load balancers and producing the sets $I^+$ and $I^-$, demand pairs in $I^+$ and $I^-$ are added to $D^{\mathrm{dir}-}$ or $D^{\mathrm{dir}+}$ respectively, and for each such demand pair, a corresponding demand is added to two recursive demands. Thus, as every edge $e \in E_\sigma$ appears in at most two sets $E_{\sigma, i, -i}$, after the second step, distinguishing the state of the update to $D^{\mathrm{dir}}$ after steps $1$ and $2$ by the subscripts $1$ and $2$, using \Cref{eq:detroute-recourse-bound-2} and \Cref{eq:detroute-recourse-bound-1}, we have 
        \begin{align} \label{eq:detroute-recourse-bound-3}
        \begin{split}
            |D^{\mathrm{dir}+}_2| + |D^{\mathrm{dir}-}_2|
                &\leq |I^+| + |I^-| + |D^{\mathrm{dir}+}_1| + |D^{\mathrm{dir}-}_1|\\
                &\leq 2 \left(\sum_{i, j} |D^{\mathrm{ovf}+}_{i, j}| + 3|D^{\mathrm{ovf}-}_{i, j}|\right) + 24L_{|\sigma|}(|E^{+}_{\sigma}| + |E^{-}_{\sigma}|) + |D^{\mathrm{dir}+}_1| + |D^{\mathrm{dir}-}_1|\\
                &\leq 18(|D^+| + |D^-|) + 72L_{|\sigma|}(|E^{+}_{\sigma}| + |E^{-}_{\sigma}|).
        \end{split}
        \end{align}
        When the update consists solely of added demand pairs, we have $|I^+| = |D^{\mathrm{ovf}+}|$ and $\{\id \in I^+\} = \{\id \in D^{\mathrm{ovf}+}\}$. After the second step, $D^{\mathrm{dir}-}$ is empty, $|D^{\mathrm{dir}+}| \leq |D^+|$ and $\{\id \in D^{\mathrm{dir}+}\} = \{\id \in D^+\}$.

\addtocounter{algorithm}{-1}
\begin{algorithm}[H]
    \caption{$\textsc{DynDetRouter}$: Dynamic Routing Update (step 3/3)} \label{alg:det-dyn-routing-step3}
    \begin{algorithmic}[1]
            \algrestore{DynDetRouter}
                \LeftComment{Step 3: route direct demand}
                \For{$i, j \in [k]$, $i < j$ }
                    \State Let $D^{\mathrm{dir}+}_{i, j}$ and $D^{\mathrm{dir}-}_{i, j}$ be the subdemands of $D^{\mathrm{dir}+}$ and $D^{\mathrm{dir}-}$ from $V_{\sigma + i}$ to $V_{\sigma + j}$
                    \State $(I^+_{i, j}, I^-_{i, j}, B^{\mathrm{rem}}_{i, j}) \gets \Call{UpdateLoadBalancer}{B^{\mathrm{dir}}_{\sigma, i, j}, D^{\mathrm{dir}+}_{i, j}, D^{\mathrm{dir}-}_{i, j}, E_{\sigma, i, j}^+, E_{\sigma, i, j}^-}$
                    \For{$(a, b, \id) \in I^-_{i, j}$}
                        \State Let $(a', b') \gets B^{\mathrm{rem}}_{i, j}(\id)$
                        \State Add $(a, a', 2\id)$ to $D^-_{\sigma + i}$
                        \State Add $(b', b, 2\id)$ to $D^-_{\sigma + j}$
                    \EndFor
                    \For{$(a, b, \id) \in I^+_{i, j}$}
                        \State Let $(a', b') \gets B^{\mathrm{dir}}_{\sigma, i, j}(\id)$
                        \State Add $(a, a', 2\id)$ to $D^+_{\sigma + i}$
                        \State Add $(b', b, 2\id)$ to $D^+_{\sigma + j}$
                    \EndFor
                \EndFor
                \State
                \LeftComment{Recursively route demand contained within clusters}
                \LeftComment{Return indices of demand pairs for which assigned path changed}
                \State Let $I^{\Delta} \gets \{\id : (a, b, \id) \in D^{\mathrm{dir}+} \cup D^{\mathrm{dir}-} \cup \bigcup_{i, j} I^+_{i, j} \cup I^-_{i, j}\}$
                \If{$|\sigma| + 1 < d$}
                    \For{$i \in [k]$ }
                        \State Let $D^{\mathrm{dir}+}_i$ and $D^{\mathrm{dir}-}_i$ be the subdemands of $D^{\mathrm{dir}+}$ and $D^{\mathrm{dir}-}$ contained in $V_{\sigma + i}$
                        \For{$(a, b, \id) \in D^{\mathrm{dir}-}_i$ }
                            \State Add $(a, b, 2\id)$ to $D^-_{\sigma + i}$
                        \EndFor
                        \For{$(a, b, \id) \in D^{\mathrm{dir}+}_i$ }
                            \State Add $(a, b, 2\id)$ to $D^+_{\sigma + i}$
                        \EndFor
                        \State Let $I^{\Delta}_i \gets \Call{DynamicRecUpdate}{\sigma + i, E^+, E^-, D^+_{\sigma + i}, D^-_{\sigma + i}}$
                        \For{$\id' \in I^{\Delta}_i$}
                            \State Add $\lfloor \id' / 2 \rfloor$ to $I^{\Delta}$
                        \EndFor
                    \EndFor
                \EndIf
                \State \Return $I^{\Delta}$
            \EndFunction
        \EndClass
    \end{algorithmic}
\end{algorithm}

        \item[Step 3.] By \Cref{lem:det-routing-init}, the maximum congestion $C^{\mathrm{dir}}_{e}$ caused to any edge $e$ by the third step is at most $8L_{|\sigma|}$. Thus, for every pair $\sigma + i$, $\sigma + j$ of child clusters, we have 
        \begin{equation} \label{eq:detroute-recourse-bound-4}
            |I^+_{i, j}| + |I^-_{i, j}|
                \leq |D^{\mathrm{dir}+}_{i, j}| + 3|D^{\mathrm{dir}-}_{i, j}| + 16L_{|\sigma|} (|E^+_{\sigma, i, j}| + |E^-_{\sigma, i, j}|).
        \end{equation}
        Then, for every demand pair in $I^+_{i, j}$ or $I^-_{i, j}$, a corresponding demand pair is added to the two recursive demands $D^+_{\sigma + i}$ and $D^-_{\sigma + j}$. Finally, for any demand pair in $D^{\mathrm{dir}+}$ or $D^{\mathrm{dir}-}$ contained in a child cluster, a corresponding demand pair is added to the recursive demand. Thus,
        \begin{align*}
            \sum_{i \in [k]} |D^+_{\sigma + i}| + |D^-_{\sigma + i}| &\leq 2(|I^+| + |I^-|) + 2\sum_{i, j} |I^+_{i, j}| + |I^-_{i, j}|\\
            &\leq 2(|I^+| + |I^-|) + 2 (|D^{\mathrm{dir+}}| + 3|D^{\mathrm{dir-}}|) + 32 L_\sigma (|E^+_\sigma| + |E^-_\sigma|)\\
            &\leq 4 (18(|D^+| + |D^-|) + 72L_{|\sigma|}(|E^{+}_{\sigma}| + |E^{-}_{\sigma}|)) + 32 L_\sigma (|E^+_\sigma| + |E^-_\sigma|)\\
            &= 72 (|D^+| + |D^-|) + 320 L_{|\sigma|}(|E^{+}_{\sigma}| + |E^{-}_{\sigma}|))\\
        \end{align*}
        where we first use \Cref{eq:detroute-recourse-bound-4} in the second inequality and then \Cref{eq:detroute-recourse-bound-3} in the third.
        
        Finally, if the update consists only of added demand pairs, it immediately follows combined with steps 1 and 2 that only $\id \in D^{\mathrm{dir}+} \cup D^{\mathrm{dir}-}$ can appear in the positive recursive demand updates $D^+_{\sigma + i}$, and that the negative recursive demand updates $D^-_{\sigma + i}$ are all empty. 
    \end{enumerate}
\end{proof}

Now, to obtain \Cref{thm:det-routing}, we simply apply \Cref{lem:det-routing-recourse} inductively on $|\sigma|$.

\begin{proof}(Of \Cref{thm:det-routing}.)
    By \Cref{lem:det-routing-init}, we have $\sum_{i \in [k]} |D^+_{\sigma + i}| + |D^-_{\sigma + i}| \leq 72 (|D^+_\sigma| + |D^-_\sigma|) + 320 \cdot \max(k, L) \cdot 20^{|\sigma|}  \cdot (|E^+_\sigma| + |E^-_\sigma|)$. Thus, for all $d' \in [d]$,
    \begin{equation*}
        \sum_{\sigma \in [k]^{d'}} |D^+_{\sigma}| + |D^-_{\sigma}| \leq 72^{d'} (|D^+| + |D^-|) + 320 \cdot \max(k, L) \cdot 72^{d'} \sum_{\sigma \in [k]^{\leq d'}} |E^+_\sigma| + |E^-_\sigma|
    \end{equation*}
    and in particular for $d' = d$,
    \begin{equation*}
        \sum_{\sigma \in [k]^{d}} |D^+_\sigma| + |D^-_\sigma| \leq 72^d (|D^+| + |D^-|) + 320 \cdot \max(k, L) \cdot 72^d \cdot (|E^+| + |E^-|)
    \end{equation*}
    as these (trivial) demands on single vertices determine all top-level demand pairs for which the assigned path changed, combining the above with the addition-only specific bound, the recourse is as desired.
    
    The total work of the algorithm has only $\poly(k, d)$ overhead on the sizes of these sets. 
\end{proof}

\section{Pruning in General Graphs Through Embedding} \label{sec:pruning-through-embedding}

In this section, we show how to use the self-pruning data structure to worst-case prune a general graph while maintaining an efficiently-sampleable oblivious routing on the graph, and a second result for maintaining an explicit routing for a dynamic demand. The former result is the following:

\begin{theorem}\label{thm:embedded-oblivious-routing}
Let $G = (V, E)$ be a graph, $H = (V'', E_H)$ a $(k, d)$-semi-hypercube on $V'' \subseteq V$ for $k \geq 16 d$ with an embedding $\Pi$ of the edge set $E_H$ into $G$ of congestion $\kappa$ and path length $h$, such that every vertex in $V$ is on some embedding path. Then, there is a deterministic data structure for maintaining a remaining vertex set $V' \subseteq V$ (initially $V' = V$) and an oblivious routing $R$ on $G[V']$ with
\begin{itemize}
    \item congestion $\kappa h \cdot k \cdot 2^{\bigO(d)}$ on $1$-load demands on $V'$,
    \item path length $\bigO(h \cdot d^3)$, and
    \item efficiently sampleable paths: a path can be sampled from $R(s, t)$ with expected work $\bigO(h \cdot d^3)$ 
\end{itemize}
under pruning updates: the query to the data structure gives an edge $e^{\mathrm{del}} \in E[V']$ to be deleted. The algorithm responds by selecting vertices $V_{\mathrm{trim}} \subseteq V'$ to \textit{prune}. The remaining vertex set is updated to $V' \gets V' \setminus (V_{\mathrm{trim}} \cup \{v\})$ and the edge set to $E \gets E - e^{\mathrm{del}}$.

The data structure has the following guarantees:
\begin{itemize}
    \item Worst-case pruning ratio: $|V_{\mathrm{trim}}| \leq \kappa h \cdot k \cdot \bigO(\log n)^{2d}$.
    \item Work per update: $\kappa h \cdot \poly(k \log^d(n))$.
\end{itemize}
\end{theorem}


Suppose that $V = V''$. To obtain the result, directly attempting to maintain that $H[V']$ is a semi-hypercube unfortunately does not work: when an edge $e^{\mathrm{del}}$ is deleted from the graph $G$, we need to remove from $H$ every edge $e'$ whose embedding path contains $e^{\mathrm{del}}$ by pruning an endpoint of $e'$. However, this would result in more embedding paths not being contained in $G[V']$ anymore.

The solution is to maintain $V''$ as a subset of the remaining vertex set $V'$, guaranteeing that $H[V'']$ is a semi-hypercube, and that the embedding path of every edge in $H[V'']$ is contained in $G[V']$.

Routing on the semi-hypercube is easy, using \Cref{sec:rand-routing} for the oblivious routing or \Cref{sec:det-routing} for the explicit routing for a dynamic demand. To route in the whole remaining graph $G[V']$, we want to be able to route from vertices in $V' \setminus V''$ to $V''$ with low congestion. This can be done easily if the vertex $v \in V' \setminus V''$ is on some embedding path of an edge in $H[V'']$, as the endpoints of the embedding path are in $V''$, and the set of embedding paths have low congestion and are short (thus only a few vertices need to use any particular embedding path).

We indeed maintain the set $V'$ to be exactly the set of vertices on some embedding path of an edge in $H[V'']$. By this definition, when a vertex leaves $V'$, no embedding path in $H[V'']$ is no longer contained in $G[V']$. Thus, with the deletion of some edge $e^{\mathrm{del}}$, we can simply first update $V''$, then remove from $V'$ the vertices no longer on an embedding path, of which there are at most $kd \cdot h$ times the number of vertices that left $V''$, as each vertex in a $(k, d, \tau)$-semi-hypercube has degree at most $kd$. Since self-pruning has a worst-case pruning ratio guarantee, the pruning guarantee of $V'$ is worst-case as well.

In this section, we first show \Cref{thm:embedded-oblivious-routing} in \Cref{sec:pruning-obliv-maintain}, then its dynamic demand counterpart \Cref{thm:embedded-deterministic-routing} in \Cref{sec:pruning-explicit-maintain}. In \Cref{sec:appendix-general-graph-embedding}, we obtain the corollaries of the two theorems in general graphs without the structure of a given embedded semi-hypercube by combining the two results with multi-commodity flow to embed the initial semi-hypercube.

\subsection{Worst-Case Pruning With Oblivious Routing} \label{sec:pruning-obliv-maintain}

\begin{proof}(of \Cref{thm:embedded-oblivious-routing}). The claimed data structure is given in \Cref{alg:general-obliv-pruning-p1}. In this data structure, we maintain an additional vertex set $V'' \subseteq V'$ for which we guarantee that $H[V'']$ always remains a $(k, d, \tau)$-semi-hypercube for $\tau = \frac{1}{4350d}$. 
    The vertex sets $V'$ and $V''$ will be related as follows:
    \begin{itemize}
        \item For every edge $e' \in E_H[V'']$, the embedding path $\Pi(e')$ is contained in $G[V']$
        \item For every vertex $v \in V'$, there is some edge $e' \in E_H[V'']$ such that $v$ is on the embedding path $\Pi(e')$.
    \end{itemize}
    Since the embedding path $\Pi(e')$ of an edge $e'$ always has the same endpoints as $e'$ by definition, there thus is in $G[V']$ for each vertex $v \in V'$ a path to a vertex $v' \in V''$ that is the prefix of some embedding path $\Pi(e')$. We maintain for each vertex $v \in V'$ one arbitrary such path $P''_v$ that we call the \textit{connecting path} of $v$, as it connects $v$ to the easy-to-route-in $V''$. We call the path endpoint $v' \in V''$ the \textit{representative} of $v$.
    
    The routing path for a demand pair $(u, v, \id)$ then consist of three parts: the connecting path of $u$, the projection $\Pi(P')$ of an oblivious routing path $P'$ on the semi-hypercube sampled as in \Cref{sec:rand-routing}, and then the reverse of the connecting path of $v$.

\begin{algorithm}[H]
    \caption{$\textsc{PruneOblivRouter}$: Initialization and Path queries} \label{alg:general-obliv-pruning-p1}
    \begin{algorithmic}[1]
        \Class{PruneOblivRouter}
            \Data
                \State Constants $V, E, \Pi$
                \State Variables $V'$, $V''$, $E''$, $U_v$, $U_e$, $\conn_v$
                \State $\textsc{SelfPruningSHC}\ \mathrm{pruner}$
            \EndData

            \Function{Initialize}{$V, V'', E, E_H, \Pi, k, d$}
                \State $V' \gets V$
                \State $E'' \gets E_H$
                \For{$e'' \in E''$}
                    \For{vertices $v$ on $\Pi(e'')$} add $e''$ to $U_v$ \EndFor
                    \For{$e \in \Pi(e'')$} add $e''$ to $U_e$ \EndFor
                \EndFor
                \For{$v \in V$} $\conn_v \gets \text{ arbitrary member of } U_v$ \EndFor
                \State $\Call{pruner.initialize}{E'', k, d, \frac{1}{4350 d}}$
            \EndFunction

            \State
            \LeftComment{Returns endpoint in $V''$ of connecting path $P''(v)$ of vertex $v \in V'$}
            \Function{getRep}{$v$}
                \If{$v \in V''$} \Return $v$
                \Else\ \Return first endpoint $u'$ of the edge $(u', v') = \conn_v$ \EndIf
            \EndFunction
            
            \State
            \LeftComment{Returns connecting path $P''(v)$ of vertex $v \in V'$}
            \Function{getConnection}{$v$}
                \If{$v \in V''$} \Return empty path from $v$ to $v$
                \Else\ \Return the reverse of the prefix of $\Pi(\conn_v)$ ending at $v$ \EndIf
            \EndFunction

            \State
            \LeftComment{Samples a path between $u$ and $v$ from the oblivious routing}
            \Function{SamplePath}{$u, v$}
                \State Let $P''_u := \Call{getConnection}{u}$
                \State Let $P''_v := \Call{reverse}{\Call{getConnection}{v}}$
                \State Let $P' := \Call{pruner.SamplePath}{\Call{getRep}{u}, \Call{getRep}{v}}$ 
                \State \Return $\Call{concat}{P''_u, \Pi(P'), P''_v}$
            \EndFunction

            \algstore{PruneRouter}
\end{algorithmic}
\end{algorithm}

We can now already analyze the congestion and path length of the routing. Consider some $1$-load demand $D$ on $V'$. This demand has a natural projection $D'$ to a demand supported on $V''$, by replacing each demand endpoint with its representative. As the demand $D$ has load $1$, the demand $D'$ is at most $h$ times a unit demand on $H[V'']$, as each embedding path has length $h$, thus there are at most $h - 1$ vertices in $V' \setminus V''$ on it. Thus, by \Cref{thm:obliv-routing-general}, the oblivious routing of $D'$ on $H[V'']$ has congestion $h \cdot k \cdot 2^{\bigO(d)}$ and length $\bigO(d^3)$, where the $2^{\bigO(d)}$-term absorbs the multiplicative $d$-term in congestion. Projecting this routing of $D'$ on $H[V'']$ back to $G[V']$, the congestion increases by a $\kappa$-factor and the length by a $h$-factor, thus the resulting congestion is $\kappa h \cdot k \cdot 2^{\bigO(d)}$ and the length $\bigO(h \cdot d^3)$. Finally, each of the connecting paths has length at most $h$, and the connecting paths cause at most $\kappa h$ congestion to any edge, which does not affect the asymptotic congestion and length of the final routing.

\addtocounter{algorithm}{-1}
\begin{algorithm}[H]
\caption{$\textsc{PruneOblivRouter}$: Pruning} \label{alg:general-obliv-pruning-p2}
\begin{algorithmic}[1]
            \algrestore{PruneRouter}
            
            \Function{PruneStep}{$e^{\mathrm{del}}$}
                \LeftComment{Delete vertices incident to an edge in $H[V'']$ whose embedding path contains the deleted edge}
                \LeftComment{Prune to maintain $H[V'']$, determining vertices and edges $V^-$ and $E^-$ to remove}
                \State Let $V^{\mathrm{aff}} := \Call{endpoints}{U_{e^{\mathrm{del}}}}$
                \State Let $V^- \gets \emptyset$
                \For{$v^{\mathrm{aff}} \in V^{\mathrm{aff}}$}
                    \State $V^- \gets V^- \cup \Call{pruner.Delete}{v^{\mathrm{aff}}}$
                \EndFor
                \State Let $E^- \gets$\ the set of edges in $E''$ incident to some $v \in V^-$
                \State $V'' \gets V'' \setminus V^-$
                \State $E'' \gets E'' \setminus E^-$
                \State
                \LeftComment{Build maps of edges to remove from $U_e$'s and $U_v$'s respectively, then update $U_e$, $U_v$}
                \State Let $\mathrm{edge\_dels} := \bigcup_{e' \in E^-} \{(e \rightarrow e') : e \in \Pi(e')\}$
                \State Let $\mathrm{vertex\_dels} := \bigcup_{e' \in E^-} \{(v \rightarrow e') : \text{$v$ is on $\Pi(e')$}\}$
                \For{$e \in \Call{keys}{\mathrm{edge\_dels}}$} $U_e \gets U_e \setminus \mathrm{edge\_dels}(e)$ \EndFor
                \For{$v \in \Call{keys}{\mathrm{vertex\_dels}}$} $U_{v} \gets U_{v} \setminus \mathrm{vertex\_dels}(v)$ \EndFor
                \State
                \LeftComment{Prune from $V'$ the vertices no longer on an embedding path}
                \State Let $V^{\mathrm{trim}} := \{\text{key } v \in \mathrm{vertex\_dels} : U_v = \emptyset\}$
                \State $V' \gets V' \setminus V^{\mathrm{trim}}$
                \State
                \LeftComment{Update connecting paths}
                \For{$v \in \Call{keys}{\mathrm{vertex\_dels}}$}
                    \If{$v \not\in V^{\mathrm{trim}}$ and $\conn_v \not\in U_v$} $\conn_v \gets \text{ arbitrary member of } U_v$ \EndIf
                \EndFor
                \State
                \LeftComment{Return vertices pruned from $V'$}
                \State \Return $V^{\mathrm{trim}}$
            \EndFunction
        \EndClass
    \end{algorithmic}
\end{algorithm}
    
    We now consider the maintenance of the sets $V'$ and $V''$ during pruning steps. When an edge $e^{\mathrm{del}}$ is deleted, the embedding paths containing $e^{\mathrm{del}}$ are no longer contained in $G[V']$. Let $V^{\mathrm{aff}}$ be the set of vertices in $V''$ incident to some edge $e' \in E_H[V'']$ with an embedding path $\Pi(e')$ containing $e^{\mathrm{del}}$. This vertex set has size at most $2 \kappa$, as the congestion of the embedding is $\kappa$.

    We want to remove the vertices $V^{\mathrm{aff}}$ from $V''$. To maintain that the set $V''$ remains a $(k, d, \tau)$-semi-hypercube, we maintain the worst-case pruning data structure of \Cref{thm:shc-strong-pruning} on $H$. When we remove $V^{\mathrm{aff}}$, this data structure additionally gives some additional vertices to be pruned; let $V^-$ be the total resulting set of vertices to remove. By \Cref{thm:shc-strong-pruning}, the size of this set is at most
    \begin{equation*}
        |V^-| = \bigO\left(\frac{\log(n)}{\tau}\right)^d |V^{\mathrm{aff}}| = \kappa \cdot \bigO(d \log n)^d \leq \kappa \cdot \bigO(\log n)^{2d}.
    \end{equation*}
    
    Let $E^-$ be the set of edges in $E_H[V'']$ incident to some vertex in $V^-$, this is the set of edges that leave $H[V'']$ when the vertices $V^-$ are removed from $V''$. Note that every edge in $E_H[V'']$ whose embedding path contained $e^{\mathrm{del}}$ is part of these leaving edges $E^-$. The set $E^-$ has size at most $kd \cdot |V^-|$.

    As a result of the edges $E^-$ leaving $E_H[V'']$, some vertices $v \in V'$ no longer appear on any embedding path $\Pi(e')$ of an edge $e' \in E_H[V''] \setminus E^-$. We let the set of pruned vertices $V^{\mathrm{prune}}$ be exactly the set of these vertices. By definition, applying $V' \gets V' \setminus V^{\mathrm{prune}}$ does not result in any embedding path of an edge in $E_H[V''] \setminus E^-$ no longer being contained in $G[V']$, and for every remaining vertex $v \in V'$, there will be some edge in $E_H[V'']$ whose embedding path $v$ appears on. Since there are at most $h + 1$ vertices on any embedding path, the size of $V^{\mathrm{prune}}$ is at most
    \begin{equation*}
        |V^{\mathrm{prune}}| \leq (h + 1) \cdot |E^-| \leq \kappa h \cdot k \cdot \bigO(\log n)^{2d},
    \end{equation*}
    where the $\bigO(1)^d$-term absorbs the multiplicative $d$.

    To bound work, note that the pruning data structure is called $|V^{\mathrm{aff}}| \leq 2\kappa$ in an update, thus the work taken by those calls is $\kappa \cdot \poly(k \log^d(n))$. The work outside those calls to the pruning data structure is dominated by updating the sets $U_v$ and $U_e$, which takes work $|\mathrm{vertex\_dels}| \poly \log(n) \leq (\kappa h \cdot k \cdot \bigO(\log n)^{2d}) \poly(\log n)$. Thus, the update takes work $\kappa h \cdot \poly(k \log^d(n))$. 
\end{proof}

\subsection{Worst-Case Pruning With Explicit Routing of Dynamic Demand} \label{sec:pruning-explicit-maintain}

In this section, instead of an oblivious routing, we consider maintaining on $G[V']$ the routing of a demand of a bounded load $L$ undergoing dynamic updates. For this, we additionally maintain on the graph $G[V'']$ an instance of the dynamic deterministic routing data structure of \Cref{sec:det-dyn-routing}. 

\begin{theorem}\label{thm:embedded-deterministic-routing}
Let $G = (V, E)$ be a graph, $H = (V'', E_H)$ a $(k, d)$-semi-hypercube on $V'' \subseteq V$ for $k \geq 16 d$ with an embedding $\Pi$ of the edge set $E_H$ into $G$ of congestion $\kappa$ and path length $h$, such that every vertex in $V$ is on some embedding path. Let $D$ be a $L$-load demand on $G$. Then, there is a deterministic algorithm for maintaining 
\begin{itemize}
    \item a remaining vertex set $V' \subseteq V$ (initially $V' = V$), such that $G[V']$ is a $\kappa h \cdot k \cdot 2^{\bigO(d)}$-congestion $h \cdot \bigO(d^3)$-length router for the all-one node weighting $\mathbbm{1}_{V'}$ on the remaining vertex set, and
    \item an explicit routing of $D$ over paths in $G[V']$ of congestion $L \cdot \kappa h \cdot k \cdot 2^{\bigO(d)}$ and length $h \cdot 2^{\bigO(d)}$, supporting queries for the path $P$ assigned to a demand pair with work $|P| \poly \log(n)$,
\end{itemize}
under updates each consisting of two steps:
\begin{itemize}
    \item Pruning step: Here, the query to the data structure gives an edge $e^{\mathrm{del}} \in E[V']$ to be deleted. The algorithm responds by selecting vertices $V_{\mathrm{trim}} \subseteq V'$ to \textit{prune}. The remaining vertex set is updated to $V' \gets V' \setminus (V_{\mathrm{trim}} \cup \{v\})$ and the edge set to $E \gets E - e^{\mathrm{del}}$. 
    \item Rerouting step: Here, the query to the data structure gives sets of demand pairs $D^+$ and $D^-$ to respectively add and remove from the demand, such that $D \gets D \cup D^+ \setminus D^-$ is supported with load $L$ on the post-pruning-step remaining vertex set $V'$. The algorithm responds by updating the routing, and returning the set $I^{\Delta}$ of identifiers of demand pairs for which the assigned path changed.
\end{itemize}
The algorithm has the following worst-case guarantees:
\begin{itemize}
    \item Pruning step recourse: $|V_{\mathrm{trim}}| \leq \kappa h \cdot k \cdot \bigO(\log n)^{2d}$.
    \item Rerouting step recourse: $|I^{\Delta}| \leq L \cdot \kappa h \cdot k^2 \cdot \bigO(\log n)^{2d} + |D^-| \cdot \bigO(1)^d + |D^+|$.
    \item Work per update: $(L \cdot \poly(\kappa h \cdot k \log^d(n))) \cdot (L + |D^+| + |D^-|)$.
\end{itemize}
\end{theorem}


\begin{proof}
    Exactly as in \Cref{alg:general-obliv-pruning-p1}, we maintain an additional vertex set $V'' \subseteq V'$ for which we guarantee that $H[V'']$ always remains a $(k, d, \tau)$-semi-hypercube for $\tau = \frac{1}{4350d}$. We also again maintain for each vertex $v \in V'$ a connecting path that is the prefix of some embedding path, and again define the representative $v' \in V''$ of $v$ as the endpoint of that path. These connecting paths give a natural projection of the demand $D$ on $V'$ to a demand $D'$ on $V''$, with a demand pair $(u, v, \id)$ being projected to $(u', v', \id)$ where $u'$ is the representative of $u$ and $v'$ the representative of $v$. We maintain a semi-hypercube dynamic routing data structure of \Cref{sec:det-dyn-routing} to maintain a routing of the demand $D'$ on $H[V'']$.
    
    The routing path for a demand pair $(u, v, \id)$ then consist of three parts: the connecting path of $u$, the projection $\Pi(P')$ of the routing path $P'$ for the projected demand pair $(u', v', \id)$ in the semi-hypercube routing data structure, and then the reverse of the connecting path of $v$.  
    
    The routing properties of the remaining vertex set $V'$ follow from \Cref{thm:embedded-oblivious-routing}, as the set is maintained exactly the same way between the two algorithms.

    \begin{algorithm}[H]
    \caption{$\textsc{PruneRouter}$: Initialization and Path queries} \label{alg:general-explicit-pruning-p1}
    \begin{algorithmic}[1]
        \Class{PruneRouter}
            \Data
                \State Constants $V, E, \Pi$
                \State Variables $V'$, $V''$, $E''$, $U_v$, $U_e$, $\conn_v$, $V^-$, $E^-$, $D$, ${D'}^+$, ${D'}^-$
                \State $\textsc{SelfPruningSHC}\ \mathrm{pruner}$
                \State $\textsc{DynDetRouter}\ \mathrm{router}$
            \EndData

            \Function{Initialize}{$V, V'', E, E_H, \Pi, D, k, d$}
                \State $V' \gets V$
                \State $E'' \gets E_H$
                \State $V^-, E^-, {D'}^+, {D'}^- \gets \emptyset$
                \For{$e'' \in E''$}
                    \For{vertices $v$ on $\Pi(e'')$} add $e''$ to $U_v$ \EndFor
                    \For{$e \in \Pi(e'')$} add $e''$ to $U_e$ \EndFor
                \EndFor
                \For{$v \in V$} $\conn_v \gets \text{ arbitrary member of } U_v$ \EndFor
                \State Let $\tau := \frac{1}{4350 d}$
                \State $\Call{pruner.initialize}{E'', k, d, \tau}$
                \State $\Call{router.initialize}{V'', E'', D, k, d, \tau}$
            \EndFunction

            \State
            \LeftComment{Returns endpoint in $V''$ of connecting path $P''(v)$ of vertex $v \in V'$}
            \Function{getRep}{$v$}
                \If{$v \in V''$} \Return $v$
                \Else\ \Return first endpoint $u'$ of the edge $(u', v') = \conn_v$ \EndIf
            \EndFunction
            
            \State
            \LeftComment{Returns connecting path $P''(v)$ of vertex $v \in V'$}
            \Function{getConnection}{$v$}
                \If{$v \in V''$} \Return empty path from $v$ to $v$
                \Else\ \Return the reverse of the prefix of $\Pi(\conn_v)$ ending at $v$ \EndIf
            \EndFunction

            \State
            \LeftComment{Returns path $P$ assigned to demand pair $(u, v, \id) \in D$}
            \Function{GetPath}{$u, v, \id$}
                \State Let $P''_u := \Call{getConnection}{u}$
                \State Let $P''_v := \Call{reverse}{\Call{getConnection}{v}}$
                \State Let $P' := \Call{router.GetPath}{\epsilon, (\Call{getRep}{u}, \Call{getRep}{v}, \id)}$
                \State \Return $\Call{concat}{P''_u, \Pi(P'), P''_v}$
            \EndFunction

            \algstore{PruneRouter}
    \end{algorithmic}
    \end{algorithm}

    \addtocounter{algorithm}{-1}
    \begin{algorithm}[H]
    \caption{$\textsc{PruneRouter}$: Pruning and Rerouting Steps} \label{general-explicit-pruning-p2}
    \begin{algorithmic}[1]
            \algrestore{PruneRouter}
            
            \Function{PruneStep}{$e^{\mathrm{del}}$}
                \LeftComment{Delete vertices incident to an edge in $H[V'']$ whose embedding path contains the deleted edge}
                \LeftComment{Prune to maintain $H[V'']$, determining vertices and edges $V^-$ and $E^-$ to remove}
                \State Let $V^{\mathrm{aff}} := \Call{endpoints}{U_{e^{\mathrm{del}}}}$
                \State $V^- \gets \Call{pruner.Delete}{V^{\mathrm{aff}}} \cup V^{\mathrm{aff}}$
                \State $E^- \gets$\ the set of edges in $E''$ incident to some $v \in V^-$
                \State $V'' \gets V'' \setminus V^-$
                \State $E'' \gets E'' \setminus E^-$
                \State
                \LeftComment{Build maps of edges to remove from $U_e$'s and $U_v$'s respectively, then update $U_e$, $U_v$}
                \State Let $\mathrm{edge\_dels} := \bigcup_{e' \in E^-} \{(e \rightarrow e') : e \in \Pi(e')\}$
                \State Let $\mathrm{vertex\_dels} := \bigcup_{e' \in E^-} \{(v \rightarrow e') : \text{$v$ is on $\Pi(e')$}\}$
                \For{$e \in \Call{keys}{\mathrm{edge\_dels}}$} $U_e \gets U_e \setminus \mathrm{edge\_dels}(e)$ \EndFor
                \For{$v \in \Call{keys}{\mathrm{vertex\_dels}}$} $U_{v} \gets U_{v} \setminus \mathrm{vertex\_dels}(v)$ \EndFor
                \State
                \LeftComment{Prune from $V'$ the vertices no longer on an embedding path}
                \State Let $V^{\mathrm{trim}} := \{\text{key } v \in \mathrm{vertex\_dels} : U_v = \emptyset\}$
                \State $V' \gets V' \setminus V^{\mathrm{trim}}$
                \State
                \LeftComment{Update connecting paths and get change to $D'$ from doing so}
                \State Let $D^\Delta_u := \{(u, v, \id) \in D : u \in \Call{keys}{\mathrm{vertex\_dels}} \text{ and } \conn_u \not\in U_u\}$
                \State Let $D^\Delta_v := \{(u, v, \id) \in D : v \in \Call{keys}{\mathrm{vertex\_dels}} \text{ and } \conn_v \not\in U_v\}$
                \State Let $D^\Delta := D^\Delta_u \cup D^\Delta_v$
                \State ${D'}^- \gets \{(\Call{getRep}{u}, \Call{getRep}{v}, \id) : (u, v, \id) \in D^\Delta\}$
                \For{$v \in \Call{keys}{\mathrm{vertex\_dels}}$}
                    \If{$v \not\in V^{\mathrm{trim}}$ and $\conn_v \not\in U_v$} $\conn_v \gets \text{ arbitrary member of } U_v$ \EndIf
                \EndFor
                \State ${D'}^+ \gets \{(\Call{getRep}{u}, \Call{getRep}{v}, \id) : (u, v, \id) \in D^\Delta\}$
                \State
                \LeftComment{Return vertices pruned from $V'$}
                \State \Return $V^{\mathrm{trim}}$
            \EndFunction

            \State

            \Function{RerouteStep}{$D^+, D^-$}
                \LeftComment{Push changes to the demand $D$ to changes in the demand $D'$}
                \LeftComment{The pushed removed demand ${D'}^-_{\mathrm{tmp}}$ should be removed either from $D' \setminus {D'}^-$ or ${D'}^+$}
                \State Let ${D'}^+_{\mathrm{tmp}} := \{(\Call{getRep}{u}, \Call{getRep}{v}, \id) : (u, v, \id) \in D^+\}$
                \State Let ${D'}^-_{\mathrm{tmp}} := \{(\Call{getRep}{u}, \Call{getRep}{v}, \id) : (u, v, \id) \in D^-\}$
                \State Let ${D'}_\cap = {D'}^+ \cap {D'}^-_{\mathrm{tmp}}$
                \State ${D'}^+ \gets {D'}^+ \cup {D'}^+_{\mathrm{tmp}} \setminus {D'}_\cap$
                \State ${D'}^- \gets {D'}^- \cup {D'}^-_{\mathrm{tmp}} \setminus {D'}_\cap$
                \State $D \gets D \cup D^+ \setminus D^-$
                \State
                \LeftComment{Update routing of the demand $D'$, return indices $R$ of demand pairs for which routing changed}
                \State Let $R \gets \Call{router.DynamicUpdate}{\emptyset, V^-, \emptyset, E^-, {D'}^+, {D'}^-}$
                \State $V^-, E^-, {D'}^+, {D'}^- \gets \emptyset$
                \State \Return $R$
            \EndFunction
        \EndClass
    \end{algorithmic}
    \end{algorithm}

    We can now already analyze the congestion and path length of the routing, as in the proof of \Cref{thm:embedded-oblivious-routing}. First, as the demand $D$ has load $L$, the demand $D'$ has load at most $L \cdot kd \cdot h$, as each vertex in $V''$ is incident to at most $kd$ edges in $E_H[V'']$ as $H$ as a $(k, d)$-semi-hypercube has regular degree $kd$, and the embedding path of each of those edges contains at most $h - 1$ vertices in $V'' \setminus V$. The dynamic semi-hypercube routing data structure thus routes this demand on $H[V'']$ with congestion $L \cdot h \cdot k \cdot 2^{\bigO(d)}$ and length $2^{\bigO(d)}$, where the $2^{\bigO(d)}$-term absorbs the multiplicative $d$-term in congestion. Projecting this routing of $D'$ on $H[V'']$ back to $G[V']$, the congestion increases by a $\kappa$-factor and the length by a $h$-factor, thus the resulting congestion is $L \cdot \kappa h \cdot k \cdot 2^{\bigO(d)}$ and length $h \cdot 2^{\bigO(d)}$. Finally, each of the connecting paths has length at most $h$, and the connecting paths cause at most $\kappa L$ congestion to any edge, which does not affect the asymptotic congestion and length of the final routing.

    The maintenance of the sets $V'$ and $V''$ during pruning steps is done exactly as in \Cref{alg:general-obliv-pruning-p2}, thus the same worst-case pruning guarantee $|V^{\mathrm{prune}}| \leq \kappa h \cdot k \cdot \bigO(\log n)^{2d}$ holds.
    
    Consider now the rerouting step. The path for a demand pair $(u, v, \id)$ changes if either the connecting path of $u$ or $v$ changes or the routing of $(u', v', \id')$ in the semi-hypercube routing data structure changes (which always occurs when the demand pair is added or removed). The number of vertices $u$ for which the connecting path changes during a pruning step is at most $(h + 1) \cdot |E^-|$, as we only change the connecting path if the embedding path the previous connecting path was a prefix of no longer exists due to the edge it is embedding leaving $E_H[V'']$, thus equivalently being part of $E^-$. Each changing connecting path results in at most $L$ added and removed demand pairs to $D'$. Denoting the demand pairs added to $D'$ by ${D'}^+$ and the demand pairs removed by ${D'}^-$, we thus have
    \begin{align*}
        |{D'}^+| &\leq |D^+| + L (s + 1) \cdot |E^-|,\\
        |{D'}^-| &\leq |D^-| + L (s + 1) \cdot |E^-|.
    \end{align*}
    Consider now the update $\Call{DynamicUpdate}{\emptyset, V^-, \emptyset, E^-, {D'}^+, {D'}^-}$ to the semi-hypercube routing data structure. Denoting by $L' = \bigO(L \cdot h \cdot kd)$ the load bound on $D'$, by \Cref{thm:det-routing}, the number of demand pairs in $D'$ whose routing path this affects is at most
    \begin{align*}
        |I^{\Delta}| &\leq (L' (|V^-| + |E^-|) + |{D'}^-|) \cdot \bigO(1)^d + |{D'}^+|\\
            &= (L' |E^-| + |D^-|) \cdot \bigO(1)^d + |D^+|\\
            &\leq L' \cdot \kappa \cdot k \cdot \bigO(\log n)^{2d} + |D^-| \cdot \bigO(1)^d + |D^+|\\
            &\leq L \cdot \kappa h \cdot k^2 \cdot \bigO(\log n)^{2d} + |D^-| \cdot \bigO(1)^d + |D^+|.
    \end{align*}

    It remains to bound the work of each update. The work of the update is dominated by the work of the calls to the pruning and routing data structures. The pruning data structure is called up to $|V^{\mathrm{aff}}| \leq 2\kappa$ times, thus the total work taken by those calls is at most $\kappa \cdot \poly(k \log^d(n))$. The work taken by the dynamic routing data structure update is
    \begin{align*}
        &(|E^-| + |V^-| + |{D'}^+| + |{D'}^-|) \cdot 2^{\bigO(d)} \cdot L' \poly(k)\\
        \leq& (L \cdot \kappa h \cdot \poly(k) \cdot \bigO(\log n)^{2d} + |D^+| + |D^-|) \cdot 2^{\bigO(d)} \cdot L' \poly(k)\\
        \leq& (L \cdot h \cdot \poly(k)) \cdot \left(L \cdot \kappa h \cdot \bigO(\log n)^{2d} + (|D^+| + |D^-|) \cdot 2^{\bigO(d)}\right)\\
        \leq& (L \cdot \poly(\kappa h \cdot k \log^d(n))) \cdot (L + |D^+| + |D^-|).
    \end{align*}
    This last term thus bounds the total work of an update.
\end{proof}
%

\section{Formal Self-Pruning} \label{sec:pruning-formal-proof}


In this section, we formally prove the properties of \Cref{alg:efficient-self-pruning}, a worst-case work-bounded version of \Cref{alg:shc-strong-self-pruning}. Specifically, we show the following result:

\selfpruning

The pruning strategy of \Cref{alg:efficient-self-pruning} is the same as that of \Cref{alg:shc-strong-self-pruning}, an overview of which is given in \Cref{sec:selfpruning-overview}. To obtain non-amortized work for each update, only one notable change is necessary: \textit{shadow mark sets}.

In \Cref{alg:shc-strong-self-pruning}, when a cluster $\sigma$ initially becomes the target child cluster of its parent child cluster $\sigma'$ and the set of marked vertices $M_\sigma$ are initialized inside $\Call{Retarget}{\sigma'}$, the construction of the mark set takes work proportional to the size of the cluster $\sigma$. Thus, a straightforward implementation of the algorithm would only give asymptotic, not worst-case, bounds on the work per update.

To solve this, we maintain every possible mark set that could be assigned to the cluster, from the start of the algorithm. Recall that $M_\sigma$ is either the set of vertices with cluster-$\sigma'$-degree strictly less than some value (as the degree is always an integer, specifically strictly less than $\lceil \frac{19}{20} k' \rceil$, where $k'$ is the number of sibling clusters the cluster $\sigma$ has), or the marked vertex set $M_{\sigma'}$ of the parent cluster in case the parent cluster is degenerate. Thus, we maintain sets $M^{\mathrm{shadow}}_{\sigma, \sigma', j}$, containing the subset of vertices in the cluster $\sigma$, with cluster-$\sigma'$-degree strictly less than $j$. We additionally store for each cluster its mark source information $\mathrm{minfo}_\sigma$, which just stores the ancestor cluster $\sigma'$ and the $j$ that determine the mark set $M_\sigma$. We initialize $\mathrm{minfo}_\sigma \gets \bot$ for all $\sigma$ as the mark set has not been initialized, and define $M^{\mathrm{shadow}}_{\sigma, \bot} = \emptyset$. 

These shadow mark sets also allow querying how many of the marked vertices in a cluster $\sigma$ are in a particular child cluster $\sigma + i$, which $\Call{Retarget}{\sigma}$ needs to know to select the next target child based on marks. Specifically, we have $|M_{\sigma + i, \mathrm{minfo}_\sigma}^{\mathrm{shadow}}| = |M_\sigma \cap V_{\sigma + i}|$. We additionally maintain the sizes $\mathrm{size}_\sigma := |V_\sigma|$ of clusters $\sigma$, which are used to check if a cluster is empty and for retarget to find the smallest nonempty child cluster. Finally, the change to \textsc{Retarget} to maintain the set of child clusters $S_\sigma$ that have not become target child clusters yet is just for simpler analysis, specifically for clusters $\sigma$ with leaf children, i.e. with $|\sigma| = d - 1$, when processing a delete-operation on a non-target child on that cluster, as in this case, when retargeting, the smallest child that has not been a target child yet might in fact be empty.


We next present in \Cref{sec:formal-self-pruning-algo} the algorithm, along with a proof of \Cref{thm:shc-strong-pruning} assuming \Cref{lem:child-sizes} showing that the algorithm maintains that each cluster has at most one critical child and and \Cref{lem:simple-mark-ratio} showing that the algorithm maintains that each cluster has a low ratio of marked vertices. We prove \Cref{lem:child-sizes} in \Cref{sec:pruning-proof-sizes}, and finally prove \Cref{lem:simple-mark-ratio} in \Cref{sec:pruning-proof-marks}. 

\begin{restatable*}{lemma}{childsizes}\label{lem:child-sizes}
    In \Cref{alg:efficient-self-pruning}, for any cluster $\sigma$, after any sequence of \textsc{Delete}-operations,
    \begin{itemize}
        \item No cluster $\sigma + i$ for $i \in S_\sigma$ is empty or $\tau$-critical.
        \item When the current target child $\sigma + t_\sigma$ became the target child, it was not $\tau$-critical.
    \end{itemize}
\end{restatable*}

\begin{restatable*}{lemma}{simplemarkratio}\label{lem:simple-mark-ratio}
In \Cref{alg:efficient-self-pruning}, initially and after every Delete-operation, for any nonempty cluster $\sigma$,
\begin{equation*}
    \frac{|M^{\mathrm{shadow}}_{\sigma, \mathrm{minfo}_\sigma}|}{|V_\sigma|} \leq \frac{1}{20}.
\end{equation*}
\end{restatable*}

\subsection{The Algorithm} \label{sec:formal-self-pruning-algo}

\begin{algorithm}[H]
    \caption{Efficient semi-hypercube self-pruning (part 1/2)}
    \label{alg:efficient-self-pruning}
    \begin{algorithmic}[1]
        \Class{EfficientSelfPruningSHC}
            \Data
                \State Constants $k, d, \rho, \mathrm{rcycle} \in \mathbb{N}$, $\tau \in (0, \frac{1}{4350 d})$, $E \subseteq \binom{[k]^d}{2}$
                \State Vertex set $V \subseteq [k]^d$
                \State Cluster sizes $\mathrm{size}_\sigma \in \{0, 1, \dots, k^{d - |\sigma|}\}$, for $\sigma \in [k]^{\leq d}$
                \State Target child clusters $t_\sigma \in [k]$, for $\sigma \in [k]^{< d}$
                \State Non-target child indices $S_\sigma \subseteq [k]$, for $\sigma \in [k]^{< d}$
                \State Shadow mark sets $M^{\mathrm{shadow}}_{\sigma, \sigma', j} \subseteq V_\sigma$, for $\sigma \in [k]^{\leq d}$, $\sigma' \in \ancs(\sigma)$, $j \in [k]$
                \State Mark set parameters $\mathrm{minfo}_\sigma \in \{\bot\} \cup (\ancs(\sigma) \times [k])$, for $\sigma \in [k]^{\leq d}$
            \EndData

            \Function{Initialize}{$k, d, \tau, E$}
                \State $\mathrm{rcycle} := 2d + 1$
                \State $\rho := \left\lceil 2H_k \cdot \mathrm{rcycle} / \tau \right\rceil$
                \State $V \gets [k]^d$
                \State $\mathrm{size}_\sigma \gets k^{d - |\sigma|}$ for all $\sigma \in [k]^{\leq d}$
                \State $t_\sigma \gets 1$ for all $\sigma \in [k]^{< d}$
                \State $S_\sigma \gets [k] \setminus \{1\}$ for all $\sigma \in [k]^{< d}$
                \State $M^{\mathrm{shadow}}_{\sigma, \sigma', j} \gets \emptyset$ for all $\sigma \in [k]^{\leq d}$, $\sigma' \in \ancs(\sigma)$, $j \in [k]$
                \State $\mathrm{minfo}_\sigma \gets \bot$ for all $\sigma \in [k]^{\leq d}$
            \EndFunction

            \State

            \LeftComment{Checks if the current target child has become empty, and selects new target if it has}
            \LeftComment{Maintaining the set $S_\sigma$ rather than taking all nonempty children is done for a simpler proof}
            \Function{Retarget}{$\sigma$}
                \While{$S_\sigma \neq \emptyset$ and $|\mathrm{size}_{\sigma + t_\sigma}| = 0$}
                    \If{$S_\sigma = \{i\}$}
                        \State $t_{\sigma} \gets i$
                        \State $\mathrm{minfo}_{\sigma + i} \gets \mathrm{minfo}_{\sigma}$
                    \ElsIf{$S_\sigma \neq \emptyset$}
                        \If{$|S_\sigma| = 1 \text{ (mod } \mathrm{rcycle}\text{)}$}
                            \State $t_{\sigma} \gets \arg\min_{i \in S_\sigma} \mathrm{size}_{\sigma + i}$
                        \Else
                            \State $t_{\sigma} \gets \arg\max_{i \in S_\sigma} |M_{\sigma + i, \mathrm{minfo}_\sigma}^{\mathrm{shadow}}|$
                        \EndIf
                        \State $\mathrm{minfo}_{\sigma + t_\sigma} \gets (\sigma, \lceil \frac{19}{20} (|S| - 1) \rceil)$
                    \EndIf
                    \State $S_\sigma \gets S_\sigma \setminus \{t_\sigma\}$
                \EndWhile
            \EndFunction
        \algstore{ParallelPruning}
    \end{algorithmic}
\end{algorithm}

\addtocounter{algorithm}{-1}
\begin{algorithm}[H]
    \caption{Efficient semi-hypercube self-pruning (part 2/2)}
    \label{alg:efficient-self-pruning-p2}
    \begin{algorithmic}[1]
        \algrestore{ParallelPruning}

            \Function{Trim}{$\sigma$}
                \If{$\sigma$ is a leaf}
                    \State \Call{ProcessRemoval}{$\sigma$}
                    \State \Return $\sigma$
                \Else
                    \State Let $v_{\mathrm{trim}} := \Call{Trim}{\sigma + t_\sigma}$
                    \State $\Call{Retarget}{\sigma}$
                    \State \Return $v_{\mathrm{trim}}$
                \EndIf
            \EndFunction

            \State
            
            \Function{Delete}{$v$}
                    \State \Call{ProcessRemoval}{$v$}
                \State Let $\sigma \gets v$, $V^\mathrm{rem} \gets \{v\}$
                \While{$|\sigma| > 0$}
                    \State $\sigma \gets \Call{Parent}{\sigma}$
                    \State $\Call{Retarget}{\sigma}$
                    \State Let $\mathrm{trim\_cou} := \Call{min}{\rho \cdot |V^\mathrm{rem}|, \mathrm{size}_{\sigma}}$
                    \For{$\mathrm{trim\_cou}$ times} $V^\mathrm{rem} \gets V^\mathrm{rem} \cup \{\Call{Trim}{\sigma}\}$ \EndFor
                \EndWhile
                \State \Return $V^\mathrm{rem} \setminus \{v\}$
            \EndFunction

            \State
            
            \Function{ProcessRemoval}{$v^{\mathrm{rem}}$}
                \State Let $V^{\mathrm{aff}} := \{v \in V : (v^{\mathrm{rem}}, v) \in E\}$
                \For{$v \in V^{\mathrm{aff}}$}
                    \State Let $\sigma$ be the lowest common ancestor of $v$ and $v^{\mathrm{rem}}$
                    \State Let $\mathrm{deg} := |\{u \in V_\sigma : (v, u) \in E\}|$
                    \For{descendant cluster $\sigma''$ of $\sigma$ containing $v$} $M^{\mathrm{shadow}}_{\sigma'', \sigma, \mathrm{deg}} \gets M^{\mathrm{shadow}}_{\sigma'', \sigma, \mathrm{deg}} \cup \{v\}$ \EndFor
                \EndFor
                \State $V \gets V \setminus v^{\mathrm{rem}}$
                \For{$\sigma \in (\{v^{\mathrm{rem}}\} \cup \ancs(v^{\mathrm{rem}}))$} $\mathrm{size}_\sigma \gets \mathrm{size}_\sigma - 1$ \EndFor
            \EndFunction
        \EndClass
    \end{algorithmic}
\end{algorithm}





\begin{lemma}\label{lem:selfpruning-work}
    The function $\Call{Delete}{V^{\mathrm{del}}}$ in \Cref{alg:efficient-self-pruning} has the following guarantees:
    \begin{itemize}
        \item Worst-case pruning ratio: the set of pruned vertices has size $|V^{\mathrm{prune}}| \leq \mathcal{O}\left(\frac{\log(n)}{\tau}\right)^d$.
        \item Worst-case work: $\Call{Delete}{V^{\mathrm{del}}}$ takes work $\poly(k \cdot (\log(n) / \tau)^d)$ and has depth $\poly \log(n)$.
    \end{itemize}
\end{lemma}

\begin{proof}
    For the pruning ratio, the size of the removed vertex set $V^{\mathrm{rem}}$ is multiplied by at most $(\rho + 1)$ at every ancestor $\sigma$ of $v$ handled, thus the total size of the removed vertex set is at most
    \begin{equation*}
        |V^{\mathrm{prune}}| = |V^{\mathrm{rem}}| - 1 \leq (\rho + 1)^{d} - 1 = \bigO(d H_k / \tau)^d = \bigO(\log(n) / \tau)^d
    \end{equation*}
    as $\log(n) = \log(k^d) = d \log(k) = \bigO(d H_k)$. For work,
    \begin{itemize}
        \item \textsc{Retarget} takes $\bigO(k)$ work, as it iterates over $i \in [k]$. 
        \item \textsc{ProcessRemoval} takes $\bigO((kd)^2 \log n)$ work, as
        \begin{itemize}
            \item the set $V^{\mathrm{aff}}$ can be constructed by iterating the $kd$ edges $(v^{\mathrm{rem}}, v) \in E$ and checking if $v \in V$ with work $\bigO(\log n)$ per $v$.
            \item Computing the cluster-$\sigma$-degree of a vertex $v \in V^{\mathrm{aff}}$, of which there are $\bigO(kd)$, takes $\bigO(kd \log n)$ work similarly. Adding $v$ to $\bigO(d)$ shadow mark sets afterwards takes $\bigO(d \log n)$ work.
            \item Updating sizes of clusters containing the removed vertex takes $\bigO(d)$ work.
        \end{itemize}
        \item \textsc{Trim} calls \textsc{Retarget} on up to $d$ clusters, and calls \textsc{ProcessRemoval} on one vertex, thus it takes $\bigO((kd)^2 \log n)$ work
        \item \textsc{Delete} calls \textsc{ProcessRemoval} directly once, and \textsc{Trim} $|V^{\mathrm{prune}}|$ times, thus it takes work $k^2 d^2 \log(n) \cdot \bigO(\log(n) / \tau)^d \leq \poly(k \cdot \log(n) / \tau)^d$.
    \end{itemize}
\end{proof}

\begin{lemma}\label{lem:selfpruning-invariants}
    For any cluster $\sigma$, initially and after every call to \textsc{ProcessRemoval},
    \begin{itemize}
        \item Maintained sizes are correct: $\mathrm{size}_\sigma = |V_\sigma|$.
        \item For all $\sigma' \in \ancs(\sigma)$ and $j \in [k]$, the set $M^{\mathrm{shadow}}_{\sigma, \sigma', j}$ consists of exactly the vertices $v \in V_\sigma$ of cluster-$\sigma'$-degree strictly less than $j$.
    \end{itemize}
\end{lemma}

\begin{proof}
    The only modifications to the vertex set $V$, cluster sizes $\mathrm{size}_\sigma$ and the shadow mark sets are made by the function $\Call{ProcessRemoval}{v^{\mathrm{rem}}}$. This function removes $v^{\mathrm{rem}}$ from $V$ and subtracts the sizes $\mathrm{size}_\sigma$ of all clusters containing it by one, thus sizes are maintained correctly.
    
    Before removing the vertex, for every vertex $v$ incident to $v^{\mathrm{rem}}$, the function adds $v$ to sets $M^{\mathrm{shadow}}_{\sigma'', \sigma, \mathrm{deg}}$ for $\sigma$ the lowest common ancestor of $v$ and $v^{\mathrm{rem}}$ and $\mathrm{deg}$ the cluster-$\sigma$-degree of $v$ before the removal of $v^{\mathrm{rem}}$. Since the edge between $v$ and $v^{\mathrm{rem}}$ is a cluster-$\sigma$-edge, the cluster-$\sigma$-degree of the vertex $v$ drops from $\mathrm{deg}$ to $\mathrm{deg} - 1$, thus the only integer it is now strictly smaller than that it was not strictly smaller than before is $\mathrm{deg}$. The degrees of vertices not incident to $v^{\mathrm{rem}}$ are obviously not affected, thus the shadow mark sets are maintained correctly.
\end{proof}

The following Lemma shows that the mark sets are maintained correctly. For it, we assume \Cref{lem:child-sizes}. The proof of \Cref{lem:child-sizes} does not need to consider the mark sets at all, as it only cares about cluster sizes.

\begin{lemma}
    Initially and after every call to \textsc{Delete}, for any non-empty non-leaf cluster $\sigma$, 
    \begin{itemize}
        \item If the cluster $\sigma$ is not degenerate, it's target child cluster $\sigma + t_\sigma$ has $\mathrm{minfo}_{\sigma + t_\sigma} = (\sigma, \lceil \frac{19}{20} k' \rceil)$, where $k'$ is the number of nonempty child clusters other than $\sigma + t_\sigma$ the cluster $\sigma$ has.
        \item If the cluster is degenerate, it's only child and target child cluster $\sigma + t_\sigma$ has $\mathrm{minfo}_{\sigma + t_\sigma} = \mathrm{minfo}_{\sigma}$.
    \end{itemize}
\end{lemma}

\begin{proof}
    First, suppose that when the child cluster $\sigma + t_\sigma$ initially became the target child cluster of the cluster $\sigma$, the cluster $\sigma$ was already degenerate. Then, the second claim corresponds to showing that $\mathrm{minfo}_\sigma$ does not change, which could only occur if the cluster $\sigma$ became the target child cluster of its parent cluster while degenerate. However, \Cref{lem:child-sizes} states that whenever $\Call{Retarget}{\sigma'}$ selects a target child cluster, either no nonempty child cluster of $\sigma'$ is $\tau$-critical which a degenerate cluster must be, or the same delete-operation will result in the new target child becoming empty. In either case, we are done.

    Suppose next that the cluster $\sigma$ was not degenerate when the child cluster $\sigma + t_\sigma$ initially became its target child cluster. Then, the claims correspond to the claim that from the moment $\Call{Retarget}{\sigma}$ selects $\sigma + t_\sigma$ as the target child until the delete-operation that causes the target child cluster $\sigma + t_\sigma$ to become empty, the number of nonempty child clusters of $\sigma$ does not change.
    
    To see this, note that during the delete-operation during which $\sigma + t_\sigma$ became the target child, from that point on during the \textsc{Delete}-operation, the only changes to the cluster vertex set $V_\sigma$ can come as the consequence of a call $\Call{Trim}{\sigma}$. Thus, after that delete-operation is complete, it still holds that no child cluster other than $\sigma + t_\sigma$ that was nonempty when $\sigma + t_\sigma$ became the target child is empty or $\tau$-critical, unless $\sigma + t_\sigma$ is empty.

    During any further \textsc{Delete}-operations, whenever a vertex is removed from a child cluster other than $\sigma + t_\sigma$, at least $\rho$ times that many vertices are removed from the target child. Thus, if a non-target child becomes empty during the operation, the target child becomes empty as well.
\end{proof}

We are now ready to prove \Cref{thm:shc-strong-pruning} assuming \Cref{lem:child-sizes} and \Cref{lem:simple-mark-ratio}.

\selfpruning*

\begin{proof}
    The algorithm is deterministic, and by \Cref{lem:selfpruning-work}, the pruning ratio, work and depth claims hold.

    Initially and after every \textsc{Delete}-operation, for a cluster $\sigma$,
    \begin{itemize}
        \item By \Cref{lem:child-sizes}, the cluster $\sigma$ has no $\tau$-critical child clusters except possibly the target child $\sigma + t_\sigma$, noting that each nonempty non-target child cluster must have its child identifier in $S_\sigma$, as \textsc{Retarget} only ever changes the target child cluster when the current one is empty.
        \item If the cluster $\sigma$ is not degenerate, by \Cref{lem:selfpruning-invariants}, that child cluster has $\mathrm{minfo}_{\sigma + t_\sigma} = (\sigma, \lceil \frac{19}{20} k' \rceil)$, where $k'$ is the number of non-target nonempty child clusters the cluster $\sigma$ has.
        \item Thus, by \Cref{lem:selfpruning-invariants}, a vertex $v \in V_{\sigma + t_\sigma}$ is in the set $M^{\mathrm{shadow}}_{\sigma + t_\sigma, \mathrm{minfo}_{\sigma + t_\sigma}}$ if and only if it has cluster-$\sigma$-degree strictly less than $\frac{19}{20} k'$.
        \item By \Cref{lem:simple-mark-ratio}, we have $|M^{\mathrm{shadow}}_{\sigma + t_\sigma, \mathrm{minfo}_{\sigma + t_\sigma}}| \leq \frac{1}{20} |V_{\sigma + t_\sigma}|$.
        \item Thus, the average cluster-$\sigma$-degree of vertices in $\sigma + t_\sigma$ is at least $(1 - \frac{1}{20}) \frac{19}{20} k' \geq \frac{9}{10} k'$.
    \end{itemize}
    Thus, the graph $G[V]$ is a $(k, d, \tau)$-semi-hypercube.
\end{proof}

\subsection{Proof of \Cref{lem:child-sizes}}\label{sec:pruning-proof-sizes}

In this subsection, we prove \Cref{lem:child-sizes} on cluster child sizes.

\childsizes

\begin{proof}
    Consider what changes occur to the vertex set $V_\sigma$ of a cluster $\sigma$ during a Delete-operation after which the cluster remains nonempty:
    \begin{enumerate}
        \item[1.] Some possibly empty set of vertices $V^{\mathrm{rem}} \subseteq V_{\sigma}$ get removed from the cluster. Then $\Call{Retarget}{\sigma}$ is called to make sure the target child cluster of $\sigma$ is not empty.
        \item[2.] At least $\rho \cdot |V^{\mathrm{rem}}|$ calls to $\Call{Trim}{\sigma}$ are made.
    \end{enumerate}
    Call a pair of the steps 1. and 2. an update. We will show that the claim holds after any sequence of updates, even possibly updates that could not be caused by a Delete-operation. This only makes the claim stronger, thus clearly suffices.

    Denote the initial size of each of the children of the cluster $\sigma$ by $n_0 := k^{d - |\sigma| - 1}$.   For $k' \in [k]$ and $j \in [k']$, for a sequence of updates being considered, let $S^{\mathrm{size}}_{k', j}$ be the set of the $j$ indices $i \in S_\sigma$ of minimum size $|V_{\sigma + i}|$ at the moment $\Call{Retarget}{\sigma}$ was selecting from a set of size $|S_\sigma| = k'$. As we "retarget" during initialization by letting $t_\sigma \gets 1$, let $S^{\mathrm{size}}_{k, j}$ be a subset of $[k]$ of size $j$ chosen with the same tie-breaking rule as what is used by \Cref{alg:efficient-self-pruning}.
    
    Let $B_{k', j}$ be the maximum possible total number of missing vertices among clusters $\sigma + i$ for $i \in S^{\mathrm{size}}_{k', j}$, i.e. the value $\sum_{i \in S^{\mathrm{size}}_{k', j}} (n_0 - |V_{\sigma + i}|)$, at the moment $\Call{Retarget}{\sigma}$ was selecting from a set of size $k'$, in any sequence of updates in which between some two updates $|S_{\sigma}| = k' - 1$ held, i.e. the selected target child cluster out of $S^{\mathrm{size}}_{k', j}$ did not just become empty in the same update it became a target child cluster.

    Proving \Cref{lem:child-sizes} now corresponds to showing that $B_{k', 1} \leq \tau n_0$ for all $k' \in [k]$, noting that the first claim follows immediately from the second claim: suppose that after some sequence of updates, a non-target child cluster of $\sigma$ is $\tau$-critical. Then, the adversary could simply call $\Call{Trim}{\sigma}$ until a retarget is triggered without affecting that cluster's size.
    
    Now, we bound the values $B_{k', j}$. Let $x := \frac{2n_0}{\rho}$ be a bound on the size of $|V^{\mathrm{rem}}|$ during an update that triggers at most one retarget. 
    We claim the following inequality holds:
    \begin{equation} \label{eq:size-bound-ineq}
        B_{k', j} \leq x + \left\{ \begin{array}{ll}
            (B_{k' + 1, j + 1}) \cdot \frac{j}{j + 1} & \text{when $k' + 1\ \mathrm{mod}\ \mathrm{rcycle} = 1$}\\
            B_{k', j} &\text{otherwise}
        \end{array}\right..
    \end{equation}
    To show the inequality holds, first suppose that the adversary only ever performs updates that result in exactly one change of target child. Then, we certainly always have $B_{k', j} \leq x + B_{k' + 1, j}$ by the definition of $x$ as the number of removals that force two retargets. When $k' + 1\ \mathrm{mod}\ \mathrm{rcycle} = 1$, the target child prior to the current one was chosen to be the smallest nonempty one. Thus, in particular, the total number of missing vertices among children with indices in $S^{\mathrm{size}}_{k', j}$ at that point must have been at most $\frac{j}{j + 1}$ times the number of missing vertices among children with indices in $S^{\mathrm{size}}_{k' + 1, j + 1}$, as otherwise one of $S^{\mathrm{size}}_{k', j}$ would have been chosen. From that point until the latest retarget, at most $x$ vertices could have been removed from those clusters.

    Suppose next that the adversary performs an update that causes multiple retargets, specifically before which the last time $\Call{Retarget}{\sigma}$ selected a target child $|S_\sigma| = k''$ held, and after which the last time it happened $|S_\sigma| = k' < k''$ held. Then, let $j_{\mathrm{add}}$ be the number of integers in $(k', k'']$ that are $1$ modulo $\mathrm{rcycle}$, and $j'_{\mathrm{add}}$ the number of integers in $(k', k'')$ that are $1$ modulo $\mathrm{rcycle}$. Denote by $B'_{k', j}$ the total number of missing vertices among clusters $\sigma + i$ for $i \in S^{k', j}_\sigma$ when the target child cluster changed for the last time during that update. Then, we claim that
    \begin{equation} \label{eq:multi-size-bound-ineq}
        B'_{k', j} \leq \frac{j}{j + j_{\mathrm{add}}'} \left(\left((k'' - k' + 1) \frac{x}{2}\right) + \frac{j + j_{\mathrm{add}}'}{j + j_{\mathrm{add}}} \left(B_{k'', j + j_{\mathrm{add}}} + \frac{x}{2}\right)\right)
    \end{equation}
    as
    \begin{itemize}
        \item before the update, for any $j'$, the maximum amount of missing vertices among any $j'$ child clusters with identifiers in $S_\sigma$ was at most $B_{k'', j'} + \frac{x}{2}$
        \item the update removes at most and in fact strictly less than $(r' - r + 1) \frac{x}{2}$ vertices, as otherwise the target child would change more than $k'' - k'$ times,
        \item $j_{\mathrm{add}}'$ retargets after the removal of the vertices $V^{\mathrm{rem}}$ in the update but before the last retarget during the update targeted the smallest child cluster, and
        \item $j_{\mathrm{add}}$ retargets before the last retarget during the update targeted the smallest child cluster.
        \item Thus applying the same ratio argument gets the inequality.
    \end{itemize}
    
    One can observe that the bound \Cref{eq:multi-size-bound-ineq} is stronger than applying \Cref{eq:size-bound-ineq} multiple times. Thus, the inequality \Cref{eq:size-bound-ineq} holds even without making assumptions on the adversary.

    Combining \Cref{eq:size-bound-ineq} with $B_{0, j} = 0$, we get that $B_{r, j} \leq x \cdot (j \cdot \mathrm{rcycle} \cdot H_k)$ where $H_k$ is the $k$\textsuperscript{th} harmonic number. As $\rho = \left\lceil 2H_k \cdot \mathrm{rcycle} / \tau \right\rceil$ and $x = \frac{2 n_0}{\rho} \leq \tau n_0 / (\mathrm{rcycle} \cdot H_k)$, we have $B_{r, 1} \leq x \cdot \mathrm{rcycle} \cdot H_k \leq \tau n_0$.
\end{proof}

\subsection{Proof of \Cref{lem:simple-mark-ratio}}\label{sec:pruning-proof-marks}

We move to obtaining the desired bound on the maximum mark ratio. To do this, we prove a stronger property in \Cref{lem:strong-mark-ratio} by induction on $\sigma$, obtaining the desired \Cref{lem:simple-mark-ratio} as an immediate corollary. Throughout the subsection, we write $M_\sigma := M^{\mathrm{shadow}}_{\sigma, \mathrm{minfo}_\sigma}$ for the marked vertex set of the cluster $\sigma$.

\begin{lemma}\label{lem:strong-mark-ratio}
    In \Cref{alg:shc-strong-self-pruning}, for any nonempty, non-leaf cluster $\sigma$, after any sequence of Delete-operations,
    \begin{itemize}
        \item Unless the cluster $\sigma$ is degenerate, the ratio of marked vertices in the cluster satisfies
        \begin{equation*}
            \frac{|M_\sigma \cap V_\sigma|}{|V_\sigma|} \leq \frac{(|\sigma| + 1) e^{(|\sigma| + 1) / d}}{20e \cdot d}.
        \end{equation*}
        \item When the latest call to $\Call{Retarget}{\sigma}$ was made, the ratio of marked vertices in the cluster satisfied
        \begin{equation*}
            \frac{|M_\sigma \cap V_\sigma|}{|V_\sigma|} \leq \frac{1}{2} \cdot \frac{(|\sigma| + 1) e^{(|\sigma| + 1) / d}}{20e \cdot d}.
        \end{equation*}
    \end{itemize}
\end{lemma}

\begin{proof}
We now prove \Cref{lem:strong-mark-ratio} by induction on increasing $|\sigma|$. Note that $M_\epsilon = \emptyset$ at all times, as the root cluster has no parent cluster to mark vertices in it. Fix now a non-root non-leaf cluster $\sigma$. Let $n_0 := k^{d - |\sigma| - 1}$ be the initial size of each of the child clusters of $\sigma$. The cluster $\sigma$ has initial size $k n_0$.

Consider what changes occur to the vertex set $V_\sigma$ and the mark set $M_\sigma$ of a cluster $\sigma$ during a Delete-operation after which the cluster remains nonempty:
\begin{enumerate}
    \item[1.] Some possibly empty set of vertices $V^{\mathrm{rem}} \subseteq V_\sigma$ get removed from the cluster. If this results in the current target child cluster of $\sigma$ becoming empty, $\Call{Retarget}{\sigma}$ is called.
    \item[2.] At least $\rho \cdot |V^{\mathrm{rem}}|$ calls to $\Call{Trim}{\sigma}$ are made.
    \item[3.] The initialization of marks in the cluster might occur, letting $M_\sigma \gets M^{\mathrm{ini\_mark}}$. This can only occur once (when the cluster $\sigma$ becomes its parent's target child), can only occur if $M_\sigma = \emptyset$ beforehand, and as we will show later, the following hold:
    \begin{itemize}
        \item The size of the cluster (after steps 1. and 2.) is at least $|V_\sigma| \geq (1 - \tau) k n_0$.
        \item The number of marked vertices satisfies
        \begin{equation*}
            |M^{\mathrm{ini\_mark}}| \leq \left(\frac{1}{2} \cdot \frac{|\sigma| e^{|\sigma| / d}}{20 e \cdot d}\right) \cdot k n_0.
        \end{equation*}
    \end{itemize}
    \item[4.] Some possibly empty set of vertices $M^{\mathrm{mark}} \subseteq V_\sigma \setminus M_\sigma$ enter the mark set $M_\sigma$. This can only occur if the mark set has been initialized before.
    \item[5.] $\Call{Trim}{\sigma}$ is called at least $\rho \cdot |M^{\mathrm{mark}}|$ times.
\end{enumerate}
The size of the cluster must be at least $|V_\sigma| \geq (1 - \tau) k n_0$ when the mark set is initialized, as by \Cref{lem:child-sizes}, the size of the cluster $\sigma$ must be at least a $(1 - \tau)$-fraction of its original size when it becomes the target child cluster of its parent (as it remains nonempty after the delete-operation, the retarget that made it the target child must have been the last of the delete-operation). The bound on the number of initially marked vertices $|M^{\mathrm{ini\_mark}}|$ follows from induction, specifically the second statement of \Cref{lem:strong-mark-ratio}, if the parent cluster of $\sigma$ is degenerate. Otherwise, as by \Cref{lem:child-sizes} every child cluster of the parent cluster was noncritical when that retarget occurred, the matching between $\sigma$ and each of its siblings had size at least a $(1 - 2\tau)$-fraction of a perfect matching, thus we have
\begin{equation*}
    |M^{\mathrm{ini\_mark}}| \leq \frac{2 \tau \cdot |V_\sigma|}{\frac{1}{20}} \leq 40 \tau \cdot k n_0 \leq \left(\frac{1}{2} \cdot \frac{1}{20e \cdot d}\right) \cdot kn_0 \leq \left(\frac{1}{2} \cdot \frac{|\sigma| e^{|\sigma| / d}}{20 e \cdot d}\right) \cdot k n_0
\end{equation*}
where we use $\tau \leq \frac{1}{4350 d} \leq \frac{1}{1600 e \cdot d}$. Denote this upper bound on $|M^{\mathrm{ini\_mark}}|$ by $x_0$.

To see why at least $\rho \cdot |M^{\mathrm{mark}}|$ trims occur in step 5. after the newly marked vertices in step 4., let $(\sigma', j) \gets \mathrm{minfo}_\sigma$ be the source of the mark set of $\sigma$. Then, it must be that the target child cluster of $\sigma'$ is an ancestor of $\sigma$, and $\sigma$ is the representative of that ancestor, i.e. all ancestors of $\sigma$ before $\sigma'$ are degenerate. Thus, a call to $\Call{Trim}{\sigma'}$ eventually recurses to $\Call{Trim}{\sigma}$. Also, a vertex enters the mark set $M_\sigma$ when its degree in the cluster $\sigma'$ drops to strictly below $j$, which can only occur when a vertex it is incident to in another child cluster of $\sigma'$ is removed, which triggers $\rho$ calls to $\Call{Trim}{\sigma'}$.

Now, we will show that after any sequence of updates, with each update consisting of the above steps 1-5, at most one update containing the mark initialization step 3., and the conditions for that step occurring being met, the claims of \Cref{lem:strong-mark-ratio} hold.

Note that by \Cref{lem:child-sizes}, we have that each non-target child $\sigma + i$ with $i \in S_\sigma$ is nonempty initially and after every update, and the target child cluster is never empty.

For $k' \in [k]$ and $j \in [k']$, after any sequence of updates after which the cluster $\sigma$ has $k'$ nonempty child clusters, let $C_{k', j}$ be the maximum total number of marked vertices (i.e. $|M_\sigma \cap V_{\sigma + i}|$) among $j$ specific child clusters $\sigma + i$, $i \in S_\sigma$ at the moment \textsc{Retarget} selected the current target child.

Showing that the second claim of \Cref{lem:strong-mark-ratio} holds now amounts to showing that
\begin{equation*}
    C_{k', k'} \leq \left(\frac{1}{2} \cdot \frac{(|\sigma| + 1) e^{(|\sigma| + 1) / d}}{20e \cdot d}\right) \cdot (1 - \tau) k' n_0.
\end{equation*}
To show that the first claim of \Cref{lem:strong-mark-ratio} holds, note that after any sequence of updates after which the cluster $\sigma$ has $k' \geq 2$ remaining nonempty children, the size of the cluster $\sigma$ is at least $(1 - \tau) (k' - 1) n_0$, as no non-target child can be $\tau$-critical, and further, the number of vertices that have entered the marked vertex set $M_\sigma$ since the latest call to $\Call{Retarget}{\sigma}$ is at most $n_0 / \rho$, as otherwise the current target child would have been trimmed empty. Thus, showing the first claim of \Cref{lem:strong-mark-ratio} amounts to showing that
\begin{equation*}
    \frac{C_{k', k'} + n_0 / \rho}{(1 - \tau) (k' - 1) n_0} \leq \left(\frac{(|\sigma| + 1) e^{(|\sigma| + 1) / d}}{20e \cdot d}\right).
\end{equation*}
As $k' - 1 \geq k' / 2$ for $k' \geq 2$, to prove both claims, it thus suffices to show that for all $k' \in [k]$,
\begin{equation*}
    C_{k', k'} + \frac{n_0}{\rho} \leq \left(\frac{1}{2} \cdot \frac{(|\sigma| + 1) e^{(|\sigma| + 1) / d}}{20e \cdot d}\right) \cdot (1 - \tau) k' n_0.
\end{equation*}

Having defined the values $C_{k', j}$ and shown that proving \Cref{lem:strong-mark-ratio} amounts to bounding each $C_{k', k'}$, we now consider an adversary attempting to maximize $C_{k', j}$ for some $k', j$ by performing a sequence of updates. We additionally allow this adversary to break the rules of $\Call{Retarget}{\sigma}$, and select the next target child arbitrarily whenever $\Call{Retarget}{\sigma}$ would select the next target child cluster based on the current sizes of the child clusters.

We can make the following simple observations about an optimal strategy for such an adversary: 
\begin{itemize}
        \item We may assume the adversary never removes vertices, i.e. $a$ is always zero in step 1 of each update. This holds as the adversary can arbitrarily select the next target child cluster whenever retarget would select the next target child based on child sizes, as retarget does not consider which specific vertices remain in nonempty child clusters while selecting the next target child based on marks, and as the number of missing vertices in nonempty child clusters does not influence the values $C_{k', j}$ in any way. Further, an adversary can simulate removing vertices by calling $\Call{Trim}{\sigma}$ additional times while the target cluster would contain removed vertices.
        \item We may assume that the adversary never marks vertices in the current target child cluster, as the current target child cluster will be empty the next time a value $C_{k', j}$ is determined.
        \item We may assume that every update performed by the adversary causes the target child of $\sigma$ to change at least once, as performing an update that does not trigger a retarget followed by an update that does has the same equivalent effect as performing both updates at once (marking the union of the marked vertex sets, and calling $\Call{Trim}{\sigma}$ the total number of times those updates would have).
\end{itemize}

Suppose that an update changes the target child of the cluster $\sigma$ some number $j$ of times. Then, in that update, at most $j \cdot \frac{2n_0}{\rho}$ vertices are marked during step 3. We let $x := 2n_0 / \rho$.

Now, we can bound $C_{k', j}$ through simple inequalities. Initially, we have $C_{k', j} = 0$ for all $j \in [k]$. Suppose that the adversary only performed updates that cause the target child of $\sigma$ to change exactly once. Then, for any $k' \in [k - 1]$ and $j \in [k']$ such that $k' + 1 < (1 - \tau) k$, we have
\begin{equation*}
    C_{k', j} \leq x + \left\{ \begin{array}{ll}
        (C_{k' + 1, j + 1}) \cdot \frac{j}{j + 1} & \text{when $k'\ \mathrm{mod}\ \mathrm{rcycle} \neq 0$}\\
        C_{k' + 1, j} &\text{otherwise}
    \end{array}\right.,
\end{equation*}
as step 2 can only occur while the cluster $\sigma$ has at least $(1 - \tau) k$ nonempty child clusters, and $\Call{Retarget}{\sigma}$ selects the new target child cluster to be the one with the largest number of vertices in $M_\sigma$ whenever $k'\ \mathrm{mod}\ \mathrm{rcycle} \neq 1$. When that occurs, the maximum number of marked vertices among $j$ nonempty non-target child clusters is at most $\frac{j}{j + 1}$ times the maximum number of marked vertices among $j + 1$ nonempty child clusters, as the new target child contains the most marked vertices among nonempty child clusters.

Similarly, for any $k' \in [k - 1]$ and $j \in [k']$ with no restriction on how large $k'$ can be, we have
\begin{equation*}
    C_{k', j} \leq x + \max\left(x_0, \left\{ \begin{array}{ll}
        (C_{k' + 1, j + 1}) \cdot \frac{j}{j + 1} & \text{when $k'\ \mathrm{mod}\ (2d + 1) \neq 0$}\\
        C_{k' + 1, j} &\text{otherwise}
    \end{array}\right\}\right).
\end{equation*}


Now, define $\mathrm{mult}_{k', j, 0} := 1$ for all $k' \in [k]$ and $j \in [k]$, and for integer $l \geq 1$, define
\begin{equation*}
    \mathrm{mult}_{k', j, l} := \left\{ \begin{array}{ll}
        (\mathrm{mult}_{k' + 1, j + 1, l - 1}) \cdot \frac{j}{j + 1} & \text{when $k'\ \mathrm{mod}\ \mathrm{rcycle} \neq 0$}\\
        \mathrm{mult}_{k' + 1, j, l - 1} &\text{otherwise}
    \end{array}\right..
\end{equation*}
Let $k'' = \lceil (1 - \tau) k \rceil - 1$ be the maximum integer strictly less than $(1 - \tau) k$. Applying the above inequalities $l \geq 1$ times, we get for some $j' = j'(k', j, l)$ that
\begin{equation}\label{eq:rkj-upper-bound}
    C_{k', j} \leq \left(\sum_{l' = 0}^{l - 1} x \cdot \mathrm{mult}_{k', j, l'}\right) + \max\left(\mathrm{mult}_{k', j, l} \cdot C_{k' + l, j'}, \mathrm{mult}_{k', j, \max(0, k'' - k')} \cdot x_0\right).
\end{equation}
Now, discarding the assumption that the adversary can only perform updates that cause $\Call{Retarget}{\sigma}$ to be called exactly once, we have
\begin{equation*}
    C_{k', j} \leq \max_{l \geq 1} \left(l \cdot x \cdot \mathrm{mult}_{k', j, l - 1} + \max\left(\mathrm{mult}_{k', j, l} \cdot C_{k' + l, j'}, \mathrm{mult}_{k', j, \max(0, k'' - k')} \cdot x_0\right)\right),
\end{equation*}
where the maximum is over the number $l$ of calls to $\Call{Retarget}{\sigma}$ caused by the latest update performed by the adversary. Since $\mathrm{mult}_{k', j, l}$ is decreasing in $l$, we thus obtain that the maximum is always attained when $l = 1$.

Now, applying \Cref{eq:rkj-upper-bound} with $l = k - k'$ on some $k' \in [k - 1]$, we thus get
\begin{equation*}
    C_{k', k'}  \leq x \cdot \left(\sum_{l' = 0}^{k - k' - 1} \mathrm{mult}_{k', j, l'}\right) + x_0 \cdot \mathrm{mult}_{k', j, \max(0, k'' - k')}.
\end{equation*}
It is sufficient for us to lose constants when bounding the first term. We have
\begin{equation*}
    \sum_{l' = 0}^{k - k' - 1} \mathrm{mult}_{k', j, l'} \leq \sum_{l' = 0}^{k} \frac{2 k'}{k' + l} \leq k' \cdot 2 H_k.
\end{equation*}
For the second term, we need a tighter bound. We have
\begin{align*}
    \mathrm{mult}_{k', j, \max(0, k'' - k')}    &\leq \frac{k'}{k' + \lfloor (k'' - k') \frac{\mathrm{rcycle} - 1}{\mathrm{rcycle}} \rfloor}\\
        &\leq \frac{k'}{\lfloor k'' \frac{\mathrm{rcycle} - 1}{\mathrm{rcycle}} \rfloor}\\
        &\leq \frac{k'}{((1 - \tau)k - 1) \frac{\mathrm{rcycle} - 1}{\mathrm{rcycle}} - 1}\\
        &\leq \frac{k'}{(1 - \tau)k \frac{\mathrm{rcycle} - 1}{\mathrm{rcycle}} - 2}.
\end{align*}
We can then lower bound the divisor by
\begin{align*}
    (1 - \tau)k \frac{\mathrm{rcycle} - 1}{\mathrm{rcycle}} - 2 &\geq \left(1 - \tau - \frac{1}{\mathrm{rcycle}} - \frac{2}{k}\right) k\\
    &\geq \left(1 - \left(\frac{1}{4350} + \frac{1}{2} + \frac{1}{8}\right) \frac{1}{d}\right) k\\
    &\geq \left(1 - \left(\frac{1}{4350} + \frac{1}{2} + \frac{1}{8}\right) \frac{1}{d}\right) \left(1 - \frac{2}{4350 d}\right) \left(1 + \frac{2}{4350 d}\right) k\\
    &\geq \left(1 - \left(\frac{3}{4350} + \frac{1}{2} + \frac{1}{8}\right) \frac{1}{d}\right) \left(1 + \frac{2}{4350 d}\right) k\\
    &\geq \left(1 - \left(1 - e^{-1}\right) \frac{1}{d}\right) \cdot \frac{1}{1 - \tau} \cdot k\\
    &\geq \frac{e^{-1 / d}}{1 - \tau} \cdot k.
\end{align*}
Where we use $\frac{1}{1 - \tau} \leq \frac{1}{1 - 1 / 4350 d} \leq \left(1 + \frac{2}{4350 d}\right)$, that $1 - \left(1 - e^{-1}\right) y \geq e^{-y}$ for $0 \leq y \leq 1$, and that $k \geq 16 d$.

Now, plugging in the values of $x$ and $x_0$ and adding $\frac{n_0}{\rho} = x / 2$ to both sides, we finally get
\begin{align*}
C_{k', k'} + \frac{n_0}{\rho} &\leq x \cdot \left(k' \cdot 2 H_k\right) + x_0 \cdot \left(\frac{k'}{k} \cdot e^{1 / d} (1 - \tau)\right) + \frac{x}{2}\\
    &\leq \left(x \cdot \left(\frac{2.5 H_k}{1 - \tau}\right) + x_0 \cdot \left(\frac{1}{k} \cdot e^{1 / d}\right)\right) \cdot (1 - \tau) k'\\
    &\leq \left(\frac{5 H_k}{(1 - \tau) \rho} + e^{1/d} \cdot \left(\frac{|\sigma| e^{|\sigma| / d}}{40e \cdot d}\right)\right) \cdot (1 - \tau) k' n_0\\
    &\leq \left(\frac{1}{40 e \cdot d} + \frac{|\sigma| e^{(|\sigma| + 1) / d}}{40e \cdot d}\right) \cdot (1 - \tau) k' n_0\\
    &\leq \left(\frac{(|\sigma| + 1) e^{(|\sigma| + 1) / d}}{40e \cdot d}\right) \cdot (1 - \tau) k' n_0,
\end{align*}
as desired. In the second-to-last inequality above, we simply use
\begin{equation*}
    \rho = \left\lceil 2H_k \cdot \mathrm{rcycle}  / \tau \right\rceil \geq 13050 d^2 H_k \geq \frac{200 e \cdot H_k \cdot d}{1 - \tau}.
\end{equation*}
\end{proof}

\simplemarkratio

\begin{proof}
    We show the claim by induction on decreasing $|\sigma|$. First, suppose that $\sigma$ is a leaf cluster, i.e. $|\sigma| = d$. Then, unless the parent cluster $\sigma'$ of $\sigma$ is degenerate, the only vertex of $\sigma$ can never be marked, as there is an edge from $\sigma$ to every other nonempty child of the parent of $\sigma$. Otherwise, if $\sigma'$ is degenerate, it must have been degenerate at the moment $\sigma$ became the current target cluster. At that moment, the ratio of marked vertices in $\sigma'$ was strictly less than 1 by the first property of \Cref{lem:strong-mark-ratio}, and as the cluster only contained one vertex, it must have been zero. If the only vertex of $\sigma$ later became marked, that marking would immediately have been followed by a call to $\Call{Trim}{\sigma}$ and $\sigma$ would not be nonempty.

    Now, suppose that $|\sigma| < d$. If the cluster $\sigma$ is degenerate, the ratio of marks is exactly the same as the ratio of marks of the only child cluster of $\sigma$, with a ratio of at most $\frac{1}{20}$ by induction. Otherwise, If the cluster $\sigma$ is not degenerate, by the second property of \Cref{lem:strong-mark-ratio},
    \begin{equation*}
        \frac{|M_\sigma \cap V_\sigma|}{|V_\sigma|} \leq \frac{(|\sigma| + 1) e^{(|\sigma| + 1) / d}}{20e \cdot d} \leq \frac{de}{20e \cdot d} = \frac{1}{20}.
    \end{equation*}
\end{proof}

\bibliographystyle{alpha}
\bibliography{aux/refs}

\appendix

\newpage

\section{Basic Load Balancer}\label{sec:appendix-loadbalancer}

In this appendix section, we show the following:

\loadbalancerlemma*

As mentioned before, maintaining this data structure is involved but not interesting.

\begin{proof} The data structure is given in \Cref{alg:load-balancer}. This data structure maintains 
\begin{itemize}
    \item A map $B : I \rightarrow E$ of the bucket assignment.
    \item An inverse map $B^{-1}$ with $B^{-1}(e)$ being the set of clients assigned to the bucket $e$.
    \item The target load $L := \lfloor |I| / |E| \rfloor$.
    \item The subset of buckets $E^{\mathrm{low}} \subseteq E$ with load $L$.
    \item The subset of buckets $E^{\mathrm{high}} \subseteq E$ with load $L + 1$. This set has size $|I| \text{ mod $|E|$}$.
\end{itemize}
on update $I^+, I^-, E^+, E^-$, the data structure does the following:
\begin{enumerate}
    \item[1-3.] Remove the buckets $E^-$, moving the clients they contained into $I^+$. Then, add to $B^{-1}$ the new buckets $E^+$ with initially no assignments. Then, remove the clients $I^-$ from the sets that contain them (or from $I^+$ if they were assigned to some removed bucket).
    
    After these first three steps, the bucket set $\Call{keys}{B^{-1}}$ contains the bucket set that should exist after the update, the buckets $B^{-1}(e)$ contain disjoint subsets of clients that were not removed in the update, and $I^+$ contains exactly the set of clients that should exist after the update that are not assigned to a bucket.

    Step 2 takes work $\tilde{O}(|E^+|)$, and each of steps 2 and 3 take $\tilde{O}(\mathrm{recourse})$ work. 
\end{enumerate}

\newpage

\begin{algorithm}[H]
    \caption{Basic load balancer data structure (part 1/2)} \label{alg:load-balancer}
    \begin{algorithmic}[1]
        \Class{LoadBalancer}
            \Data
                \State Target load $L \in \mathbb{N}$
                \State Set $E^{\mathrm{low}} \subseteq E$ of buckets with load $L$
                \State Set $E^{\mathrm{high}} \subseteq E$ of buckets with load $L + 1$
                \State Map $B : I \rightarrow E$ of the bucket assignment
                \State Map $B^{-1} : E \rightarrow 2^I$ of clients assigned to the bucket
            \EndData

            \Function{Initialize}{\null}
                \State $L \gets 0$
                \State $E^{\mathrm{low}}, E^{\mathrm{high}}, B, B^{-1} \gets \emptyset$
            \EndFunction

            \State

            \Function{UpdateLoadBalancer}{$I^+, I^-, E^+, E^-$}
                \State Let $L^{\mathrm{new}} := \lfloor (|B| + |I^+| + |I^-|)\ /\ (|B^{-1}| + |E^+| + |E^-|) \rfloor$
                \State
                \LeftComment{1. Remove the buckets $E^-$, add newly unassigned clients not in $I^-$ to $I^+$}
                \State Let $I^{A-} \gets \left(\bigcup_{e^- \in E^-} B^{-1}(e^-)\right)$
                \State Let $B^{\mathrm{rem}} = \{(i \rightarrow B(i)) : i \in I^{A-}\}$
                \State Let $I^{\cap} := I^- \cap I^{A-}$
                \State $I^+ \gets I^+ \cup (I^{A-} \setminus I^{\cap})$
                \State $I^- \gets I^- \setminus I^{\cap}$
                \State $B^{-1} \gets B^{-1} \setminus E^-$
                \State
                \LeftComment{2. Add to $B^{-1}$ the buckets $E^+$ with initially no assignments}
                \State $B^{-1} \gets B^{-1} \cup \{(e^+ \rightarrow \emptyset) : e^+ \in E^+\}$
                \State
                \LeftComment{3. Remove the removed clients $I^-$}
                \LeftComment{Add every bucket with a removed client to $E^+$}
                \State $I^{A-} \gets I^{A-} \cup I^-$
                \State $B^{\mathrm{rem}} \gets B^{\mathrm{rem}} \cup \{(i \rightarrow B(i)) : i \in I^-\}$
                \State Let $B^{--} := \bigcup_{i^- \in I^-} \{(B(i^-) \rightarrow i^-)\}$
                \For{$e \in \Call{keys}{B^{--}}$}
                    \State $B^{-1}(e) \gets B^{-1}(e) \setminus B^{--}(e)$
                \EndFor
                \State $E^+ \gets E^+ \cup \Call{keys}{B^{--}}$
                \State $B \gets B \setminus (I^\cap \cup I^-)$
                \State
                \LeftComment{4. Update the sets $E^{\mathrm{high}}$ and $E^{\mathrm{low}}$}
                \State $E^{\mathrm{low}} \gets E^{\mathrm{low}} \setminus (E^+ \cup E^-)$
                \State $E^{\mathrm{high}} \gets E^{\mathrm{high}} \setminus (E^+ \cup E^-)$
                \If{$L^{\mathrm{new}} = L + 1$}
                    \State $E^+ \gets E^+ \cup E^{\mathrm{low}}$
                    \State $E^{\mathrm{low}} \gets E^{\mathrm{high}}$
                    \State $E^{\mathrm{high}} \gets \emptyset$
                \ElsIf{$L = L^{\mathrm{new}} - 1$}
                    \State $E^+ \gets E^+ \cup E^{\mathrm{high}}$
                    \State $E^{\mathrm{high}} \gets E^{\mathrm{low}}$
                    \State $E^{\mathrm{low}} \gets \emptyset$
                \ElsIf{$|L - L^{\mathrm{new}}| \geq 2$}
                    \State $E^+ \gets E^+ \cup E^{\mathrm{low}} \cup E^{\mathrm{high}}$
                    \State $E^{\mathrm{low}} \gets \emptyset$
                    \State $E^{\mathrm{high}} \gets \emptyset$
                \EndIf
            \algstore{LoadBalancer}
    \end{algorithmic}
\end{algorithm}

\addtocounter{algorithm}{-1}
\begin{algorithm}[H]
    \caption{Basic load balancer data structure (part 2/2)} 
    \begin{algorithmic}[1]
            \algrestore{LoadBalancer}
                \LeftComment{5. Move excess clients into $I^+$ from too large load buckets in $E^+$}
                \For{$e^+ \in E^+$}
                    \If{$|B^{-1}(e^+)| \geq L^{\mathrm{new}} + 1$}
                        \State Let $I^{\mathrm{trim}}$ be an arbitrary subset of $B^{-1}(e^+)$ of size $|B^{-1}(e^+)| - (L^{\mathrm{new}} + 1)$
                        \State $B^{-1}(e^+) \gets B^{-1}(e^+) \setminus I^{\mathrm{trim}}$
                        \State $I^+ \gets I^+ \cup I^{\mathrm{trim}}$
                        \State $I^{A-} \gets I^{A-} \cup I^{\mathrm{trim}}$
                        \State $B^{\mathrm{rem}} \gets B^{\mathrm{rem}} \cup \{(i \rightarrow e^+) : i \in I^{\mathrm{trim}}\}$
                        \State $E^{\mathrm{high}} \gets E^{\mathrm{high}} \cup \{e^+\}$
                        \State $E^+ \gets E^+ \setminus \{e^+\}$
                    \EndIf
                \EndFor
                \State
                \LeftComment{6. If $I^+$ is too small to fill buckets in $E^+$ up to $L^{\mathrm{new}}$, move to $I^+$ from $E^{\mathrm{high}}$}
                \State Let $\mathrm{tot\_need} := \sum_{e^+ \in E^+} L^{\mathrm{new}} - |B^{-1}(e^+)|$
                \If{$\mathrm{tot\_need} < |I^+|$}
                    \State Let $E^{\mathrm{trim}}$ be an arbitrary subset of $E^{\mathrm{high}}$ of size $\mathrm{tot\_need} - |I^+|$
                    \For{$e \in E^{\mathrm{trim}}$}
                        \State Let $i^{\mathrm{trim}}$ be an arbitrary element of $B^{-1}(e)$
                        \State $B^{-1}(e) \gets B^{-1}(e) \setminus \{i^{\mathrm{trim}}\}$
                        \State $I^+ \gets I^+ \cup \{i^{\mathrm{trim}}\}$
                        \State $I^{A-} \gets I^{A-} \cup \{i^{\mathrm{trim}}\}$
                        \State $B^{\mathrm{rem}} \gets B^{\mathrm{rem}} \cup \{(i^{\mathrm{trim}} \rightarrow e)\}$
                    \EndFor
                    \State $E^{\mathrm{high}} \gets E^{\mathrm{high}} \setminus E^{\mathrm{trim}}$
                    \State $E^{\mathrm{low}} \gets E^{\mathrm{low}} \cup E^{\mathrm{trim}}$
                \EndIf
                \State
                \LeftComment{Get clients that get assigned to a bucket}
                \State Let $I^{A+} \gets I^+$
                \State
                \LeftComment{7. Fill buckets in $E^+$ from $I^+$}
                \For{$e^+ \in E^+$}
                    \State Let $I^{\mathrm{add}}$ be an arbitrary subset of $I^+$ of size $L^{\mathrm{new}} - |B^{-1}(e^+)|$
                    \State $I^+ \gets I^+ \setminus I^{\mathrm{add}}$
                    \State $B^{-1}(e^+) \gets B^{-1}(e^+) \cup I^{\mathrm{add}}$
                    \For{$i \in I^{\mathrm{add}}$}
                        \State $B(i) \gets e^+$
                    \EndFor
                \EndFor
                \State $E^{\mathrm{low}} \gets E^{\mathrm{low}} \cup E^+$
                \State
                \LeftComment{8. Add remaining elements in $I^+$ to some buckets in $E^{\mathrm{low}}$}
                \State Let $E^{\mathrm{grow}}$ be an arbitrary subset of $E^{\mathrm{low}}$ of size $|I^+|$
                \For{$e \in E^{\mathrm{grow}}$}
                    \State Let $i$ be an arbitrary element of $I^+$
                    \State $I^+ \gets I^+ \setminus \{i\}$
                    \State $B(i) \gets e$
                    \State $B^{-1}(e) \gets B^{-1}(e) \cup \{i\}$
                \EndFor
                \State $E^{\mathrm{low}} \gets E^{\mathrm{low}} \setminus E^{\mathrm{grow}}$
                \State $E^{\mathrm{high}} \gets E^{\mathrm{high}} \cup E^{\mathrm{grow}}$
                \State
                \State $L \gets L^{\mathrm{new}}$
                \State \Return $(I^{A+}, I^{A-}, B^{\mathrm{rem}})$
            \EndFunction
        \EndClass
    \end{algorithmic}
\end{algorithm}

\begin{enumerate}
    \item[4.] Update the sets $E^{\mathrm{low}}$ and $E^{\mathrm{high}}$ with the possibly changed value $L$. We first remove from both the set of buckets $E^+ \cup E^-$ that have been modified so far. Since these buckets had at least one assigned client change, the work for this is at most $\tilde{O}(\mathrm{recourse})$.
    
    Then if $L$ increased by one, the old $E^{\mathrm{high}}$ becomes the new $E^{\mathrm{low}}$ and the old $E^{\mathrm{low}}$ gets pushed into $E^+$ as it needs to receive clients later. Similarly if $L$ decreased by one, the old $E^{\mathrm{low}}$ becomes the new $E^{\mathrm{high}}$ and the old $E^{\mathrm{high}}$ gets pushed into $E^+$. If $L$ changes by at least two, both get pushed into $E^+$, otherwise neither set is changed.  Once again, as we only push out buckets that will have to receive a change in assigned clients, the work to do this is $\tilde{O}(\mathrm{recourse})$

    After this point, $E^+$ contains every bucket $e$ for which the number of assigned clients $|B^{-1}(e)|$ is not either $L^{\mathrm{new}}$ or $L^{\mathrm{new}} + 1$, and may contain some buckets with those numbers of assignees as well. Of the buckets that should exist after the update and are not in $E^+$, the set $E^{\mathrm{low}}$ contains exactly the set with number of currently assigned clients $L^{\mathrm{new}}$ and $E^{\mathrm{high}}$ exactly the set with number of currently assigned clients $L^{\mathrm{new}} + 1$.
    \item[5.] Some buckets in $E^+$ have $L^{\mathrm{new}} + 1$ or even strictly more currently assigned clients, which can specifically occur when $L^{\mathrm{new}} < L$. Remove from each if those buckets the surplus clients, adding them to $I^+$, and move the buckets to $E^{\mathrm{high}}$ from $E^+$.

    Afterwards, $E^+$ contains only buckets with at most $L^{\mathrm{new}}$ currently assigned clients. This step takes $\tilde{\bigO}(|E^+| + \mathrm{recourse})$ work, as each client removed from buckets this step will be part of $I^\Delta$.
    \item[6.] We now want to fill buckets in $E^+$ up until they contain $L^{\mathrm{new}}$ clients, so that each bucket contains exactly $L^{\mathrm{new}}$ or $L^{\mathrm{new}} + 1$ clients. However, the set $I^+$ of unassigned clients might not be large enough to do this, in which case we have to fill the surplus by taking single clients from buckets in $E^{\mathrm{high}}$, moving them to $E^{\mathrm{low}}$.
    
    This step takes work $\tilde{O}(\mathrm{tot\_need}) = \tilde{O}(\mathrm{recourse})$, where $\mathrm{tot\_need}$ is the deficit in $I^+$.
    \item[7.] In step 7, we fill each bucket in $E^+$ from $I^+$ to contain at least $L^{\mathrm{new}}$ clients. For each bucket, we add to it the deficit from $I^+$, then add it to $E^{\mathrm{low}}$ as it now has exactly $L^{\mathrm{new}}$ assigned clients.

    After this step, $E^{\mathrm{low}}$ is exactly the set of buckets with $L^{\mathrm{new}}$ assigned clients, $E^{\mathrm{high}}$ is exactly the set of buckets with $L^{\mathrm{new}} + 1$ assigned clients, and every bucket is in one of these sets.

    This step takes $\tilde{O}(|E^+| + \mathrm{recourse})$ work, as each added client is one that changes its assigned bucket. 
    \item[8.] Finally, the set $I^+$ of unassigned clients might not be empty, which occurs exactly when $\mathrm{tot\_need}$ was negative, and in which case we have here $|I^+| = -\mathrm{tot\_need}$. We take $|I^+|$ sets from $E^{\mathrm{low}}$ and add a client from $I^+$ to each of them, moving them from $E^{\mathrm{low}}$ to $E^{\mathrm{high}}$. This takes $\tilde{O}(\mathrm{recourse})$ work.
\end{enumerate}
It remains to analyze the recourse.
\begin{itemize}
    \item For each bucket removed, the up to $\lceil |I| / |E| \rceil$ clients that were in the bucket have to be reassigned, thus appearing in both $I^{A+}$ and $I^{A-}$. The total contribution to the recourse of bucket removals is thus at most $2\lceil |I| / |E| \rceil \cdot |E^-|$ (for the pre-update ratio $|I| / |E|$).
    \item Each client added is added to some bucket, appearing in $I^{A+}$. Thus, client additions contribute $|I^+|$ to the recourse.
    \item For each bucket added, up to $L^{\mathrm{new}}$ clients have to be moved into it. Note that if $L^{\mathrm{new}} > L$, then every client moved to the bucket comes from either a removed bucket or was added in the same update, and we can charge the recourse to bucket removals and client additions. Otherwise, $L^{\mathrm{new}} \leq L = \lfloor |I| / |E| \rfloor$ (pre-update), and thus the contribution to the recourse of bucket additions is at most $2\lceil |I| / |E| \rceil \cdot |E^+|$. 
    \item For each client removed from a bucket, up to that many clients have to be moved back to the bucket from other buckets. These moved clients count for both $I^{A+}$ and for $I^{A-}$, while the removed clients themselves appear in $I^{A-}$. Thus, the contribution to the recourse of client removals is at most $3 |I^-|$.
\end{itemize}
Thus, $|I^{A+}| + |I^{A-}| \leq |I^+| + 3|I^-| + 2\lceil |I| / |E| \rceil \cdot (|E^-| + |E^+|)$, as desired. 
\end{proof}

\section{Corollaries of \Cref{sec:pruning-through-embedding} on Unstructured Routers} \label{sec:appendix-general-graph-embedding}

In this section, we finally combine the theorems of \Cref{sec:pruning-through-embedding} with a length-constrained expander routing result for embedding the initial semi-hypercube, obtaining pruning results for unstructured routers.

\pruningobliviousrouting

\begin{proof}
    Let $k := \lfloor n^{1 / d} \rfloor$, and note that
    \begin{equation*}
        k^d \geq n \cdot \left(1 - \frac{1}{n^{1 / d}}\right)^d \geq n \cdot e^{- 2d / n^{1 / d}} = \Omega(n).
    \end{equation*}
    For initialization, compute a routing of the embedding demand of a $(k, d)$-semi-hypercube $H = (V'', E_H)$ on some size-$k^d$ vertex subset plus an edge from every vertex in $V \setminus V''$ to some vertex in $V''$, such that every vertex in $V''$ is the endpoint of at most $\bigO(1)$ such edges. Note that this demand has load $\bigO(d \cdot n^{1 / d}) \leq n^{\bigO(1 / d)}$. For the embedding paths $P$ of the additional edges, append $\Call{concat}{P, \Call{reverse}{P}}$ to the embedding path of some edge in $E_H$ incident to the endpoint of $P$ in $V''$. This increases the maximum length of any embedding path by at most a constant factor.
    
    Now, we have an embedding of $E_H$ into $G$ of congestion $\kappa' = \kappa \cdot n^{\bigO(1 / d)}$ and length $h' = \bigO(h)$, such that every vertex in $V$ is on some embedding path. Apply \Cref{thm:embedded-oblivious-routing} on $G, H, V, V'', E, E'', k, d, \kappa', h'$, giving
    \begin{itemize}
        \item congestion of the oblivious routing of $\kappa' h' \cdot k \cdot 2^{\bigO(d)} = \kappa h \cdot n^{\bigO(1 / d)}$,
        \item oblivious routing path length of $\bigO(h' \cdot d^3) = \bigO(h \cdot d^3)$,
        \item worst case pruning ratio of $\kappa' h' \cdot k \cdot \bigO(\log n)^{2d} = \kappa h \cdot n^{\bigO(1 / d)}$,
        \item and work per update of $\kappa' h' \cdot \poly(k \log^d(n)) = \kappa h \cdot n^{\bigO(1 / d)}$,
    \end{itemize}
    noting that as $d = \bigO(\sqrt{\frac{\log n}{\log \log n}})$, we have
    \begin{equation*}
        (\log n)^{\bigO(d)} = 2^{\bigO(d \log \log n)} \leq 2^{\bigO(\log(n) / d)} = n^{\bigO(1 / d)}.
    \end{equation*}
\end{proof}

\pruningdeterministicrouting

\begin{proof}
    Let $k := \lfloor n^{1 / d} \rfloor$, and note that
    \begin{equation*}
        k^d \geq n \cdot \left(1 - \frac{1}{n^{1 / d}}\right)^d \geq n \cdot e^{- 2d / n^{1 / d}} = \Omega(n).
    \end{equation*}
    For initialization, compute a routing of the embedding demand of a $(k, d)$-semi-hypercube $H = (V'', E_H)$ on some size-$k^d$ vertex subset plus an edge from every vertex in $V \setminus V''$ to some vertex in $V''$, such that every vertex in $V''$ is the endpoint of at most $\bigO(1)$ such edges. Note that this demand has load $\bigO(d \cdot n^{1 / d}) \leq n^{\bigO(1 / d)}$. For the embedding paths $P$ of the additional edges, append $\Call{concat}{P, \Call{reverse}{P}}$ to the embedding path of some edge in $E_H$ incident to the endpoint of $P$ in $V''$. This increases the maximum length of any embedding path by at most a constant factor.
    
    Now, we have an embedding of $E_H$ into $G$ of congestion $\kappa' = \kappa \cdot n^{\bigO(1 / d)}$ and length $h' = \bigO(h)$, such that every vertex in $V$ is on some embedding path. Apply \Cref{thm:embedded-deterministic-routing} on $G, H, V, V'', E, E'', k, d, \kappa', h'$, giving
    \begin{itemize}
        \item congestion of the explicit routing of $L \cdot \kappa' h' \cdot k \cdot 2^{\bigO(d)} = L \cdot \kappa h \cdot n^{\bigO(1 / d)}$,
        \item length of the explicit routing of $h' \cdot 2^{\bigO(d)} = h \cdot 2^{\bigO(d)}$,
        \item worst case pruning ratio of $\kappa' h' \cdot k \cdot \bigO(\log n)^{2d} = \kappa h \cdot n^{\bigO(1 / d)}$,
        \item rerouting step recourse of
        \begin{equation*}
            L \cdot \kappa' h' \cdot k^2 \cdot \bigO(\log n)^{2d} + |D^-| \cdot 2^{\bigO(d)} + |D^+| = L \cdot \kappa h \cdot n^{\bigO(1 / d)} + |D^-| \cdot 2^{\bigO(d)} + |D^+|,
        \end{equation*}
        \item and work per update of
        \begin{equation*}
            (L \cdot \poly(\kappa' h' \cdot k \log^d(n))) \cdot (L + |D^+| + |D^-|) = (L \cdot \poly(\kappa h) \cdot n^{\bigO(1 / d)}) \cdot (L + |D^+| + |D^-|),
        \end{equation*}
    \end{itemize}
    noting that as $d = \bigO(\sqrt{\frac{\log n}{\log \log n}})$, we have
    \begin{equation*}
        (\log n)^{\bigO(d)} = 2^{\bigO(d \log \log n)} \leq 2^{\bigO(\log(n) / d)} = n^{\bigO(1 / d)}.
    \end{equation*}
\end{proof}

\section{Batched Amortized Pruning of Non-Critical Semi-Hypercubes} \label{sec:amortized-formal-pruning}

In this section, we show the following Lemma:

\batchedpruning*

\begin{proof}
To do this, we maintain that after processing the $i$\textsuperscript{th} batch, the graph is a noncritical $(k, d, \tau \frac{i}{b})$-semi-hypercube. The pruning algorithm for handling the $i$\textsuperscript{th} batch of deletions $V^{\mathrm{del}}$, simply first deletes the vertices $V^{\mathrm{del}}$, then while there is a $\tau \frac{i}{b}$-critical cluster, prunes every vertex in that cluster.

First, for the pruning ratio, denote by $\mathrm{rat}_{d'}$ the maximum ratio of vertices deleted from a cluster $\sigma$ of depth $|\sigma| = d'$ during the processing of some update $i$ to the number of vertices pruned from the same cluster in the same update, except for vertices pruned due to some ancestor cluster of $\sigma$ becoming critical. For $d' = d$, we have $\mathrm{rat}_{d'} = 0$. Consider now some $d' < d$.

First, if the cluster $\sigma$ did not become $\tau \frac{i}{b}$-critical during the processing of the batch, the pruning ratio in the cluster is at most $\mathrm{rat}_{d' + 1}$. Otherwise, since the cluster $\sigma$ was not $\tau \frac{i - 1}{b}$-critical before processing the batch, strictly more than $\frac{\tau}{b} k^{|\sigma|}$ vertices have to be removed from the cluster while processing the batch, either as direct deletions or prunes to maintain descendant clusters of $\sigma$. For this to occur, strictly more than $\left(\frac{\tau}{b} k^{|\sigma|}\right) / (\mathrm{rat}_{d' + 1} + 1)$ vertices have to have been deleted from the cluster in that batch. Thus,
\begin{equation*}
    (\mathrm{rat}_{d'} + 1) \leq (\mathrm{rat}_{d' + 1} + 1) / \left(\frac{\tau}{b}\right) 
\end{equation*}
and we get $(\mathrm{rat}_{d'} + 1) \leq (b / \tau)^{d - d'}$, and in particular that the pruning ratio is at most $(b / \tau)^{d} - 1$.

Finally, we need to consider the work and depth of processing the update. We can implement the pruning operation as a recursive function that deletes $V^{\mathrm{del}} \cap V_\sigma$ from the cluster $\sigma$ and returns the set of vertices removed (either deleted or pruned) from the cluster. First, it partitions the set $V^{\mathrm{del}} \cap V_\sigma$ into $k$ subsets for each child cluster with work $\bigO(k \cdot |V^{\mathrm{del}} \cap V_\sigma|)$ and depth $\bigO(\log n)$, then recursively processes deletions in child clusters of the cluster in parallel, then updates the size of the current cluster with the vertices removed during processing of the child clusters, and if the cluster has became $\frac{\tau i}{b}$-critical, returns the vertex set of the entire cluster, otherwise returning the union of the removed vertex subsets of the child clusters. This takes work $\bigO(k \cdot |V^{\mathrm{rem}}_\sigma|)$ where $V^{\mathrm{rem}}_\sigma$ is the returned set of removed vertices. By the bound on the pruning ratio, we know that $\sum_\sigma |V^{\mathrm{rem}}_\sigma| \leq d \cdot (b / \tau)^d |V^{\mathrm{del}}|$. Thus, the total work of the algorithm is $\bigO(kd \cdot (\tau / b)^{-d} \cdot |V^{\mathrm{del}}|)$, and the depth of the algorithm is $\bigO(d \log n)$, as each of the $d$ depths of clusters are processed sequentially, and each depth takes depth $\bigO(\log n)$ to process.
\end{proof}

\end{document}